\documentclass[12pt]{article}

\usepackage{etoolbox}
\usepackage{arxiv_package}

\usepackage{mathrsfs}

\usepackage{enumitem}

\usepackage[utf8]{inputenc}
\usepackage[T1]{fontenc}
\usepackage{natbib}
\usepackage{url}            
\usepackage{booktabs}       
\usepackage{amsfonts}       
\usepackage{nicefrac}       
\usepackage{microtype}      
\usepackage{xcolor} 
\usepackage{bm}
\usepackage{amsthm}
\usepackage{comment}
\usepackage{hyperref}       
\usepackage{cleveref}
\usepackage{amsmath}
\usepackage{multicol, multirow, colortbl}
\usepackage{pifont}
\usepackage{caption}
\usepackage{subcaption}
\usepackage{tabularx}
\usepackage{multirow}
\usepackage{makecell}
\usepackage[hang,flushmargin]{footmisc}
\usepackage{wrapfig}
\usepackage{diagbox,graphicx}
\usepackage{float}

\usepackage{appendix}

\usepackage{bbm}
\AtBeginDocument{%
  \DeclareFontFamily{U}{bbold}{}%
  \DeclareFontShape{U}{bbold}{m}{n}{<-> bbold12}{}%
}
\usepackage{dsfont}

\DeclareMathAlphabet{\mathbbold}{U}{bbold}{m}{n}

\newcommand{\fz}[1]{{\color{blue}[Fangnan: #1]}}

\newcommand{\includeorplaceholder}[3][]{%
  \IfFileExists{#2}{%
    \includegraphics[#1]{#2}%
  }{%
    {\setlength{\fboxsep}{0pt}\setlength{\fboxrule}{0.4pt}%
     \fbox{\parbox[c][#3][c]{\linewidth}{}}%
    }%
  }%
}

\hypersetup{colorlinks=true,citecolor=blue,linkcolor=blue}

\makeatletter
\let\original@footnotemark\footnotemark
\newcommand{\align@footnotemark}{%
  \ifmeasuring@
    \chardef\@tempfn=\value{footnote}%
    \original@footnotemark
    \setcounter{footnote}{\@tempfn}%
  \else
    \iffirstchoice@
      \original@footnotemark
    \fi
  \fi}
\pretocmd{\start@align}{\let\footnotemark\align@footnotemark}{}{}
\makeatother

\usepackage{mathtools}

\newcolumntype{C}{>{\centering\arraybackslash}X}

\usepackage[algoruled]{algorithm2e}

\usepackage[normalem]{ulem} 
\usepackage{cancel}

\title{
Fit CATE Once:
Model-Assisted Randomization Tests Without Sample Splitting
}
\author{Fangnan Zheng  
\quad and \quad 
Yao Zhang\thanks{Department of Statistics and Data Science, National University of Singapore.}}

\begin{document}

\maketitle

\begin{abstract}
Randomization tests and flexible treatment-effect models offer complementary strengths for analyzing data from randomized panel experiments: the former provide valid inference under the known assignment mechanism, while the latter can capture complex patterns of effect heterogeneity. We develop model-assisted randomization tests that combine these strengths without sample splitting.
The key idea is to estimate an unsigned version of the conditional average treatment effect (CATE) from the covariance structure of residualized outcomes, while leaving the realized assignments for randomization inference. 
We establish identification and consistency for the proposed unsigned CATE estimators, as well as validity for the CATE-assisted randomization tests.
Synthetic data experiments under staggered-adoption, parallel-group,
and crossover designs show that the proposed CATE-assisted randomization
test controls Type I error and achieves higher power than
covariate-adjusted and sample-split alternatives. In a real-data application to Chicago's procedural-justice training
program, it generally yields smaller \(p\)-values and more rejections
when testing for treatment effects within racial and ethnic subgroups
and for different training cohorts.
\end{abstract}

\section{Introduction}

Randomized experiments run over time produce sequences of outcomes that are well suited to studying dynamic causal effects. For example, a government rolling out an employment subsidy across regions can track annual earnings to see how the policy's impact unfolds. Similarly, a clinical trial that staggers treatment start times across patients can record symptom scores over multiple visits to learn whether effects strengthen, plateau, or fade.
To recover these patterns, standard analyses often use regression models to separate treatment effects from other variation, but valid inference can become sensitive to modeling assumptions that are not always easy to justify.

This concern motivates design-based or randomization-based inference, where validity is grounded in the known assignment mechanism rather than in strong outcome-modeling assumptions. For example, randomization-based central limit theorems support asymptotically valid confidence intervals for treatment effects without an i.i.d.\ assumption \citep{li2017general}, and related ideas have been developed for time-series and panel experiments \citep{athey2022design,bojinov2019time,bojinov2021panel}. When baseline covariates are available, ordinary least squares adjustment with heteroskedasticity-robust standard errors yields asymptotically valid intervals even under model misspecification \citep{lin2013agnostic}, and follow-up work extends this to nonlinear adjustment \citep{guo2023generalized,cohen2024no} and to more complex designs \citep{lin2025unifying,gao2025causal}.

In this article, we focus on randomization tests \citep{fisher1935design} for causal inference. They offer finite-sample validity and allow flexible test statistics, including statistics that use regression models to detect complex treatment-effect patterns.
Existing work typically adjusts the statistic with a conditional-mean model \citep{rosenbaum2002covariance,zhao2021covariate}; a natural next step is to use a CATE model, which captures effect heterogeneity across units and time. However, naively fitting a CATE model on the realized assignments and using it in the test may break validity. One must either refit the model for every randomized assignment \citep{guo2025ml}, which is usually expensive, or use sample splitting, which sacrifices power. Recently, \citet{zhang2025adaptive} proposes an adaptive sample splitting method, which aims to reduce the power loss in Bernoulli trials. The adaptive procedure splits the assignments for nuisance estimating and randomization inference based on iteratively fitted CATE estimates.

In this article, we develop an assignment-free method to construct unsigned CATE estimates for randomization tests in the setting of time-series experiments, where we can observe multiple outcomes of each unit over time. The key observation of this method is that, under an additive outcome model, residualized outcomes can be written as centered treatment-timing indicators multiplied by lagged treatment effects, plus noise, so their second moments encode treatment-effect magnitudes. 
Under the moment and identification conditions developed below, diagonal residual moments identify a common magnitude for lag-invariant effects, while off-diagonal moments identify the lagged CATE vector up to sign when residual errors are serially uncorrelated. The sign is common across lags at a fixed covariate value and may differ across covariate values. Because these moments depend only on covariates, outcomes, and the known assignment mechanism, they can be estimated \emph{without} fitting on the realized assignments. Given the unsigned CATE estimates, the sign can simply be chosen to fit the observed outcomes unit-wisely, under the simulated and observed assignments. This yields a CATE-assisted and computationally efficient  randomization test with the standard finite-sample validity.

On the theoretical side, we establish local and global identification results for our assignment-free estimation strategy and prove consistency for the proposed unsigned CATE estimator. On the empirical side, we verify that the unsigned CATE estimator is consistent and can warm-start CATE estimation. In other words, our approach can also be used independently as a two-step approach for CATE estimation in time-series experiments.
Through intensive experiments, we show that our CATE-assisted randomization tests control Type I error and achieve higher power than unadjusted, covariate-adjusted, and sample-split alternatives across various designs.

The rest of the paper is organized as follows. Section~\ref{sec:related} discusses related work. Section~\ref{sect:sad} introduces the experimental setup and the randomization-test framework. Section~\ref{sec:te_var} develops the diagonal and off-diagonal residual-moment estimators for unsigned CATEs and establishes identification and consistency results. Section~\ref{subsec:te_in_rt} explains how to construct CATE-assisted randomization tests. \Cref{sect:exp} presents the experiments, and \Cref{sect:discussion} concludes with a discussion.

\section{Related work}
\label{sec:related}

Our work connects to several strands of the literature on randomized inference. The closest is the line on randomization tests with covariate adjustment. A classical starting point is \citet{rosenbaum2002covariance}, who studies randomization tests based on residualized outcomes from a conditional mean model. For inference on the average treatment effect, \citet{zhao2021covariate} propose a robust \(t\)-statistic based on a linear model with treatment--covariate interactions. \citet{hennessy2016conditional} develop conditional randomization tests that restrict the assignment space to account for covariate imbalance between treated and control groups. Our work differs from this line in two ways. First, we study experiments run over time, where treatment timing creates repeated outcomes. Second, instead of using only outcome adjustment, we use residual covariances to extract treatment-effect information before randomization tests.

A second strand is the literature on design-based, or randomization-based, inference for panel and time-series experiments. \citet{bojinov2019time} develop causal estimands and exact randomization tests for temporal interventions in single-unit settings, while \citet{bojinov2021panel} extend this design-based framework to panel experiments with repeated observations on multiple units. Most relevant to us is \citet{zhang2025multiple}, who develop conditional randomization tests for lagged and spillover treatment effects without using outcome or treatment-effect models. Building on their conditioning idea, we study how to construct more powerful test statistics by estimating CATEs from residual covariances without using the realized assignments.

A third strand is the literature on staggered-adoption designs. \citet{athey2022design} analyze randomized staggered-adoption experiments from a design-based perspective. In observational settings, a large literature has shown that standard two-way fixed-effects event-study regressions can fail to recover dynamic effects under treatment-effect heterogeneity \citep{goodmanbacon2021difference,sun2021estimating}, prompting a range of alternative estimators \citep{callaway2021did,borusyak2024revisiting,roth2023trending}. Our setting is different: treatment timing is randomized and known from the design, so we do not need parallel-trends assumptions. Instead, we exploit randomization directly for exact or conditional inference.

\section{Setup}\label{sect:sad}

Suppose we observe a panel with \(N\) units and \(T\) time periods.
For each unit \(i\in[N]\), we observe baseline covariates
\(X_i\in\mathbb R^d\), outcomes
\(Y_{i,[T]}=(Y_{i,1},\ldots,Y_{i,T})\), and a treatment start time
\(A_i\in[T+1]\), where \(A_i=T+1\) denotes a never-treated unit.
Write \(O_i=(X_i,Y_{i,[T]},A_i)\) for the observed data
of unit \(i\) and \(O_{[N]}=(O_1,\ldots,O_N)\) for
the full observed data.
Let \(\mathcal X\) denote the support of the covariates.

Unless stated otherwise, we use a completely randomized
staggered-adoption design. 
We extend our method to parallel-group trials and crossover
trials with randomized treatment sequences in
Appendix~\ref{sect:designs}.
Fix cohort sizes
\(N_1,\ldots,N_{T+1}\) summing to \(N\), and draw \(A_{[N]}\)
uniformly over allocations with exactly \(N_a\) units in cohort \(a\),
independently of the covariates and potential outcomes. Define
\[
N_{>s}:=\sum_{a=s+1}^{T+1}N_a,\qquad s=0,\ldots,T,
\]
the number of units remaining untreated at time \(s\). Write
\[
\pi_a(x):=\mathbb P(A_i=a\mid X_i=x)\qquad \text{and} \qquad
\pi_{\le s}(x):=\sum_{a=1}^s\pi_a(x).
\]
In staggered-adoption design, \(\pi_a(x)=N_a/N\) and
\(1-\pi_{\le s}(x)=N_{>s}/N\). We retain the argument \(x\) in this
notation to allow more general known assignment mechanisms.

We impose two standard assumptions throughout. First, we assume no interference, meaning that unit \(i\)'s outcome at time \(t\) depends only on its own treatment start time: for any \(a\in[T+1]\), \(Y_{i,t}=Y_{i,t}(A_i)=Y_{i,t}(a)\) whenever \(A_i=a\).
Second, we assume no anticipation, meaning that treatment starting after time \(t\) has no effect on the time-\(t\) outcome: \(Y_{i,t}(a)=Y_{i,t}(T+1)\) for all \(a>t\).

In this article, we mainly consider testing a global sharp null and lag-specific nulls.
The sharp null of no treatment effect is
\begin{equation}
\label{eq:H_0}
H_0:\quad
Y_{i,t}(a)=Y_{i,t}(T+1)
\qquad
\text{for all } i\in[N],\ 1\le a\le t\le T.
\end{equation}
This null asks whether treatment has any effect at all. Because it is sharp, all missing potential outcomes can be imputed under \(H_0\), making it the basic benchmark for randomization inference. In many panel experiments, however, the sharp null is too coarse: rejecting \(H_0\) does not reveal whether the effect appears immediately, emerges only after several periods, or varies with exposure length. We therefore refine the global null by introducing lag-specific null hypotheses.

Fix a time \(t\in[T-l]\) and a lag \(l\in\{0,\dots,T-1\}\). Define the time-specific lag-\(l\) null
\begin{equation}
\label{eq:H_0tl}
H_{0,t,l}:\quad
Y_{i,t+l}(t)=Y_{i,t+l}(T+1)
\qquad
\text{for all } i\in[N].
\end{equation}
This null asks whether treatment initiated at time \(t\) has any effect on the outcome \(l\) periods later, at time \(t+l\). The lag-\(l\) null \(H_{0,l}\) is the intersection of the time-specific nulls \(H_{0,t,l}\) over all \(t=1,\dots,T-l\). In other words, \(H_{0,l}\) asserts that treatment has no effect exactly \(l\) periods after adoption, uniformly over all  treatment start times.
Both nulls \(H_0\) and \(H_{0,l}\) can be extended to fixed-effect null hypotheses with prespecified effect sizes, and can also be restricted to prespecified subgroups.

The same framework applies to sharp nulls specifying a constant
effect \(\vartheta\in\mathbb R\). The values of \(\vartheta\)
not rejected by the corresponding randomization tests form a
confidence set for the constant effect \citep{rosenbaum2002covariance}.

\subsection{Model-assisted randomization tests}\label{sec:rt}

A randomization test is defined through a test statistic
\(T(O_{[N]};\hat\eta)\), which maps the data \(O_{[N]}\) into
evidence against the null. Here \(\hat\eta\) denotes a nuisance
model, such as regression models for the outcomes and/or treatment
effects. This model is often used to adjust for covariates and
construct a lower-variance statistic by making the alternative
distribution easier to distinguish from the null distribution.

Let \(\tilde A_{[N]}\) be an independent copy of \(A_{[N]}\)
drawn from the same assignment distribution. Let
\(\tilde{\mathbb P}\) denote the distribution of
\(\tilde A_{[N]}\), and let \(\tilde O_{[N]}\) denote the
simulated dataset obtained by replacing the observed assignments
in \(O_{[N]}\) with \(\tilde A_{[N]}\) while keeping the observed
covariates and outcomes fixed. We test the global null \(H_0\)
using the \(p\)-value
\begin{equation}\label{equ:p_value}
p(O_{[N]};\hat\eta)
=
\tilde{\mathbb P}_{\tilde A_{[N]}}
\big\{
T(\tilde O_{[N]};\hat\eta)
\ge T(O_{[N]};\hat\eta)
\big\}.
\end{equation}

\begin{theorem}\label{thm:validity}
Assume that
\(
\hat\eta\independent A_{[N]}
\mid X_{[N]},Y_{[N],[T]}.
\)
Then rejecting \(H_0\) whenever
\(p(O_{[N]};\hat\eta)\le\alpha\) controls the Type I error
at level \(\alpha\in(0,1)\):
\begin{equation}\label{equ:validity}
\mathbb P_{H_0}
\big\{p(O_{[N]};\hat\eta)\le\alpha
\mid X_{[N]},Y_{[N],[T]}\big\}
\le\alpha.
\end{equation}
\end{theorem}

A Monte Carlo version using the plus-one correction also controls
the Type I error \citep{hemerik2018exact}.
The conditional independence assumption may fail if \(\hat\eta\)
is fitted using \(A_{[N]}\). Refitting \(\hat\eta\) for every
simulated \(\tilde A_{[N]}\) can restore validity, but is often
computationally expensive, especially when \(\hat\eta\) is a
machine learning model.

The hypothesis \(H_{0,t,l}\) in \eqref{eq:H_0tl} is weaker
than \(H_0\). Under \(H_{0,t,l}\) and no anticipation, 
\[
Y_{i,t+l}(a)=Y_{i,t+l}(a')=Y_{i,t+l}(T+1),
\quad
a,a'\in\mathcal A_{t,l}:=\{t,t+l+1,\ldots,T+1\}.
\]
Define the subset of units treated at time $t$ or after time $t+l$ by 
\begin{equation}\label{equ:set_I}
\mathcal I:=\mathcal I_{t,l}(A_{[N]})
=\{i\in[N]:A_i\in\mathcal A_{t,l}\},
\qquad
\mathcal J=[N]\setminus\mathcal I.
\end{equation}
Thus \(\mathcal I\) contains units that either start treatment
at time \(t\) or remain untreated at time \(t+l\).
Following \citet{zhang2025multiple}, we test \(H_{0,t,l}\) using
the conditional \(p\)-value
\begin{equation}\label{equ:pvaluetl}
p_{t,l}(O_{[N]};\hat\eta)
=
\tilde{\mathbb P}_{\tilde A_{[N]}}
\big\{
T_{t,l}(\tilde O_{[N]};\hat\eta)
\ge T_{t,l}(O_{[N]};\hat\eta)
\mid
\mathcal I_{t,l}(\tilde A_{[N]})=\mathcal I,\,
\tilde A_{\mathcal J}=A_{\mathcal J}
\big\}.
\end{equation}
We hold \(\mathcal I\) and the other units' assignments
\(A_{\mathcal J}\) fixed. The $p$-value uses the data
\begin{equation}\label{eq:lag_invariant_data}
\mathcal D_{t,l}
:=
\left(
X_{[N]},\,
Y_{\mathcal J,[T]},\,
\{Y_{i,r}:i\in\mathcal I,\ r<t\text{ or }r=t+l\}
\right).
\end{equation}
Under \(H_{0,t,l}\), no interference, and no anticipation,
these data remain unchanged across all permitted assignments.
The outcomes of units in \(\mathcal J\) are fixed because
their assignments are fixed. For units in \(\mathcal I\),
outcomes before \(t\) are fixed by no anticipation, and outcomes
at \(t+l\) are fixed by the null and no anticipation.
We can fit \(\hat\eta\) using
\((\mathcal D_{t,l},\mathcal I,A_{\mathcal J})\), with any fitting
randomness independent of the experiment, and hold the fitted
functions fixed across randomization draws. The statistic is
a fixed function of these data, the fitted functions, and the
candidate assignment.

\begin{proposition}
\label{prop:lag_conditional_validity}
Under no interference and no anticipation,
rejecting \(H_{0,t,l}\) whenever
\(p_{t,l}(O_{[N]};\hat\eta)\le\alpha\) controls the Type I error
at level \(\alpha\in(0,1)\):
\[
\mathbb P_{H_{0,t,l}}\!\left\{
p_{t,l}\le\alpha
\,\middle|\,
\mathcal D_{t,l},\mathcal I,A_{\mathcal J}
\right\}
\le\alpha,
\qquad \alpha\in(0,1).
\]
\end{proposition}

\section{Assignment-free estimation of unsigned CATEs}
\label{sec:te_var}

We focus on tests of the global sharp null \(H_0\), for which
the full outcome panel can be used for assignment-free nuisance
fitting. We briefly discuss tests of \(H_{0,t,l}\) at the end
of this section and give further details in
Appendix~\ref{app:lag_cate_estimation}.

We define the outcome function and the residuals at time $t$ as
\[
\mu_t(x):=\mathbb E[Y_{i,t}\mid X_i=x] \quad \text{ and }\quad  R_{i,t}:=Y_{i,t}-\mu_t(X_i).
\]
We define the expected control outcome at time $t$ as
\[
\mu_{0,t}(x)
:= \mathbb E \big[Y_{i,t}(T+1)\mid X_i=x\big] = \mathbb E\big[Y_{i,t}\mid X_i=x,\,A_{i}>t\big]
\]
by randomization and no anticipation, whenever
\(1-\pi_{\le t}(x)>0\).
For each lag $l\in\{0,1,\dots,T-1\}$ and time $t\in [T-l]$, 
we define the expected outcome at time \(t+l\), \(l\) steps after treatment initialization at time \(t\), as
\[
\mu_{1,t+l}(x):= \mathbb E \big[Y_{i,t+l}(t)\mid X_i=x\big],
\]
For simplicity, we suppress the dependence on the adoption
time \(t\) in the notation \(\mu_{1,t+l}(x)\).
The same convention applies to its estimator.

Define the lagged CATE vector $\tau(x):=\big(\tau_0(x),\dots,\tau_{T-1}(x)\big)^\top$, where $\tau_l(x)$ is
the lag-$l$ conditional average treatment effect (CATE), defined as
\begin{equation}\label{eq:tau_l_def}
\tau_l(x)
:=\mu_{1,t+l}(x)  - \mu_{0,t+l}(x).
\end{equation}
The definition in \eqref{eq:tau_l_def} assumes that, conditional on \(X_i\), the treatment effect depends only on the exposure length \(l\), but not on the time at which the treatment begins. This assumption holds when there is no {\it time-specific factor} that modify treatment effects, such as seasonality, aggregate shocks, or changes in treatment implementation.

Without using the treatment assignments \(A_{[N]}\), we estimate the \emph{second-order moments} of the residual vector \(R_i=(R_{i,1},\dots,R_{i,T})^\top\), defined by
\begin{equation}\label{equ:second_R}
\Sigma_R(x)
:=\mathbb E\!\left[R_i R_i^\top \mid X_i=x\right],
\quad
c_{t,s}(x):=[\Sigma_R(x)]_{t,s}
=\mathbb E\!\left[R_{i,t}R_{i,s}\mid X_i=x\right].
\end{equation}
Under the model introduced below, estimating the CATE vector \(\tau(X_i)\) using the treatment assignments \(A_{[N]}\) is straightforward and admits a closed-form expression; we will discuss this estimation strategy in \Cref{subsec:first_moment}.

\subsection{Additive outcome model}

\begin{assumption}\label{ass:additive_outcome_model}
We posit the additive model for observed outcomes over time:
\begin{equation}\label{eq:yit}
Y_{i,t}
= \mu_{0,t}(X_i)
+ \sum_{l=0}^{t-1}\mathbf{1}_{\{A_i=t-l\}}\tau_l(X_i)
+ \epsilon_{i,t}, \qquad t\in[T],
\end{equation}
where \(\epsilon_{i,t}\) satisfies
\(
\mathbb E[\epsilon_{i,t}\mid X_i,A_i]=0, t\in[T]
\).
\end{assumption}

Assumption~\ref{ass:additive_outcome_model} is flexible in several respects. The untreated mean \(\mu_{0,t}(x)\) is allowed to vary arbitrarily with both covariates and calendar time, and the treatment effect \(\tau_l(x)\) may be heterogeneous in \(x\) and depend on the exposure length \(l\). Thus, the model does not impose linearity or parametric structure on either the baseline outcome or the dynamic treatment effect. Its main restriction is an additive mean structure: conditional on \(X_i\) and \(A_i\), treatment shifts the outcome mean by a lag-specific amount, while the remaining error has mean zero. 
Additional restrictions on the second moments of \(\epsilon_i\) are imposed below only where needed.

In \eqref{eq:yit}, the sum has at most one nonzero
term because $A_i$ takes a single value: if $A_i=t-l$ then the sum is
the lag-$l$ effect $\tau_l(X_i)$, while if $A_i>t$ (not yet treated) then the entire sum is zero.
Following the transform of \citet{robinson1988root}, 
we residualize \eqref{eq:yit} by subtracting off its expectation given $X_i$, which yields
\begin{equation}\label{eq:yit_expression}
R_{i,t}= Y_{i,t} - \mu_t(X_i)
=\sum_{l=0}^{t-1}\Big(\mathbf 1_{\{A_i=t-l\}}-\pi_{t-l}(X_i)\Big)\tau_l(X_i)+\epsilon_{i,t}.
\end{equation}
Equation \eqref{eq:yit_expression} expresses $R_{i,t}$ as the sum of a treatment component and a noise component.
The treatment component is a linear combination of the centered indicators
$\mathbf 1_{\{A_i=t-l\}}-\pi_{t-l}(X_i)$, which have mean zero conditional on $X_i$.
Rewrite \eqref{eq:yit_expression} as
\[
R_i=\Pi(X_i,A_i)\,\tau(X_i)+\epsilon_i,
\]
where $\epsilon_i:=(\epsilon_{i,1},\dots,\epsilon_{i,T})^\top$ and $\Pi(X_i,A_i)$ is a $T\times T$ \emph{random} assignment matrix with
\[
[\Pi(X_i,A_i)]_{t,l+1}
=
\begin{cases}
\mathbf 1_{\{A_i=t-l\}}-\pi_{t-l}(X_i),
& 0\le l\le t-1,\\
0, & t\le l\le T-1,
\end{cases}
\qquad t\in[T].
\]
Substituting this expression for \(R_i\) into \(\Sigma_R(x)\) in
\eqref{equ:second_R}, we obtain the decomposition
\begin{equation}\label{equ:decomposition}
\begin{split}
\Sigma_R(x) &  = \Sigma_{A}(x;\tau)+\Sigma_\epsilon(x),\\[5pt]
\Sigma_{A}(x;\tau)
& :=\mathbb E\!\left[
\Pi(X_i,A_i)\tau(X_i)\tau(X_i)^\top
\Pi(X_i,A_i)^\top
\mid X_i=x
\right],\\[5pt]
\Sigma_\epsilon(x)
& :=\mathbb E\!\left[\epsilon_i\epsilon_i^\top \mid X_i=x\right].
\end{split}
\end{equation}
The first term \(\Sigma_{A}(x;\tau)\) is the part of the residual covariance
generated by randomized treatment timing, while the second term
\(\Sigma_\epsilon(x)\) is the residual error covariance.

Following \eqref{equ:second_R}, \(\Sigma_R(x)\) can be
estimated by regressing the residual products \(R_iR_i^\top\) on the covariates
\(X_i\), using estimated residuals in practice. Matching this estimated
covariance with the decomposition in \eqref{equ:decomposition} gives conditional
moment equations for the CATE vector \(\tau(x)\). Although
\(\Pi(X_i,A_i)\) is random, the expectation defining \(\Sigma_A(x;\tau)\)
depends only on the known assignment probabilities. Thus no treatment-assignment
model needs to be estimated. These  moment equations cannot
distinguish \(\tau(x)\) from \(-\tau(x)\), because
\(\tau(x)\tau(x)^\top\) is unchanged when \(\tau(x)\) is replaced by
\(-\tau(x)\). Below, we show how different assumptions on
\(\Sigma_\epsilon(x)\) lead to different estimation strategies.

\subsection{Estimation from diagonal moments}\label{subsubsec:static_est}

We begin with a simple setting in which the treatment effect is constant across lags. Under this assumption, the diagonal entries of \(\Sigma_R(x)\) take a particularly simple form and already suffice to identify the treatment-effect magnitude.

\begin{assumption}\label{assumption:diagonal}
For any \(x\in\mathcal X\) and any \(l\in\{0,\dots,T-1\}\), the lag-\(l\) effect is constant:
\[
\tau_l(x)\equiv \tau_{*}(x).
\]
\end{assumption}

\begin{assumption}[Time-invariant residual variance]\label{assumption:diag_variance}
For any \(x\in\mathcal X\),
\[
\mathbb E[\epsilon_{i,t}^2\mid X_i=x]=\sigma_\epsilon^2(x),
\qquad t\in[T].
\]
\end{assumption}

Assumption~\ref{assumption:diagonal} reduces the lagged
CATE vector to a single effect \(\tau_*(x)\), while
Assumption~\ref{assumption:diag_variance} makes the
residual noise contribution common across time. Together, they imply that changes in the diagonal residual moments across \(t\) are driven by the known \(v_t(x)\), which allows us to identify \(\tau_*^2(x)\). If residual variances vary over time, these diagonal moments would confound treatment-effect magnitude with time-varying noise.

\begin{proposition}\label{prop:diagonal_moment}
Suppose Assumptions~\ref{ass:additive_outcome_model}, \ref{assumption:diagonal}, and \ref{assumption:diag_variance} hold. Then, for each \(t\in[T]\),
\[
c_{t,t}(x)
:=
[\Sigma_R(x)]_{t,t}
=
\tau_*^2(x)\,v_t(x)+\sigma_\epsilon^2(x),
\qquad
v_t(x):=\pi_{\le t}(x)\bigl(1-\pi_{\le t}(x)\bigr).
\]
\end{proposition}

Proposition~\ref{prop:diagonal_moment} shows that the conditional
variance \(c_{t,t}(x)=\mathbb E[R_{i,t}^2\mid X_i=x]\)
has two components: the variation \(\tau_*^2(x)v_t(x)\)
generated by randomized treatment timing and the residual
error variance \(\sigma_\epsilon^2(x)\).
The unknowns are the squared magnitude
\(\tau_*^2(x)\) and the variance
\(\sigma_\epsilon^2(x)\).
Because \(\sigma_\epsilon^2(x)\) is common across \(t\),
we can eliminate it by comparing diagonal moments over time.
Averaging the expression gives
\[
\bar c(x):=\frac{1}{T}\sum_{t=1}^T c_{t,t}(x)
=
\tau_*^2(x)\,\bar v(x)+\sigma_\epsilon^2(x),
\qquad
\bar v(x):=\frac{1}{T}\sum_{t=1}^T v_t(x).
\]

\begin{proposition}\label{prop:diagonal_identification}
Suppose Assumptions~\ref{ass:additive_outcome_model}, \ref{assumption:diagonal}, and \ref{assumption:diag_variance} hold. Fix \(x\in\mathcal X\). If there exists \(t\in[T]\) such that \(v_t(x)\neq \bar v(x)\), the magnitude \(m_*(x) := |\tau_*(x)|\) is point identified, with
\[
m_*(x)
=
\sqrt{\frac{c_{t,t}(x)-\bar c(x)}{v_t(x)-\bar v(x)}}.
\]
\end{proposition}
Proposition~\ref{prop:diagonal_identification} shows that, after eliminating 
\(\sigma_\epsilon^2(x)\) by differencing across time, the remaining variation in the diagonal moments identifies \(\tau_*^2(x)\). The key requirement is that the design generate nontrivial variation in \(v_t(x)\) across \(t\). Since the moment equation identifies only the squared effect, \(\tau_*(x)\) is identified up to sign.
We let
\begin{equation}\label{eq:tau_diagonal}
\hat\tau_{\pm}^{\mathrm{diag}}(x)
:=\hat m_*(x)\mathbf 1_T,
\qquad
\mathbf 1_T=(1,\ldots,1)^\top\in\mathbb R^T.
\end{equation}
where we estimate \(m_*(x)\) by plugging in estimators of the diagonal moments:
\[
\hat{m}_*(x)
:=
\sqrt{
\max\left\{
\frac{\hat c_{t,t}(x)-\hat{\bar c}(x)}{v_t(x)-\bar v(x)},0
\right\}
},
\qquad 
\hat{\bar c}(x):=\frac{1}{T}\sum_{t=1}^T \hat c_{t,t}(x).
\]
\begin{proposition}\label{prop:diag_consistency}
Suppose Assumptions~\ref{ass:additive_outcome_model}, \ref{assumption:diagonal}, and \ref{assumption:diag_variance} hold. Fix \(x\in\mathcal X\), and suppose there exists \(t\in[T]\) such that \(v_t(x)\neq \bar v(x)\). If
\(
\hat c_{t,t}(x)\xrightarrow{\;p\;} c_{t,t}(x)\) for all \(t\in[T]\),
then
\[
\hat{m}_*(x)\xrightarrow{\;p\;}m_*(x).
\]
\end{proposition}
The plug-in magnitude estimator \(\hat m_*(x)\) is consistent for \(m_*(x)=|\tau_*(x)|\). 
Note that \eqref{eq:tau_diagonal} yields one estimator for each \(t\in[T]\) with \(v_t(x)\neq \bar v(x)\). These estimators may differ in finite samples, so we aggregate them to obtain the final estimator of \(m_*(x)\).

\subsection{Estimation from off-diagonal moments}\label{subsubsec:cross_est}

Instead of assuming constant error variances over time,
we now assume that the residual errors are uncorrelated
across time. Then any off-diagonal residual covariance in $\Sigma_R(x)$ comes from treatment timing and lag-specific treatment effects.
\begin{assumption}\label{assumption:offdiagonal}
For any \(x\in\mathcal X\), \(\Sigma_\epsilon(x)\) is diagonal, i.e.,
\[
\mathbb E[\epsilon_{i,t}\epsilon_{i,s}\mid X_i=x]=0
\qquad \text{for all } t\neq s.
\]
\end{assumption}
Let \(\mathcal P:=\{(t,s):1\le s<t\le T\}\) and vectorize the off-diagonal moments from \eqref{equ:second_R}:
\[
c(x):=\bigl(c_{t,s}(x)\bigr)_{(t,s)\in\mathcal P}\in \mathbb R^{T(T-1)/2}.
\]

\begin{proposition}\label{prop:offdiag_moment}
Under Assumptions~\ref{ass:additive_outcome_model} and \ref{assumption:offdiagonal}, \(c(x)\) has the quadratic representation
\begin{equation}\label{equ:key_identity}
c(x)=F_x(\tau):=\bigl(\tau^\top H_{t,s}(x)\tau\bigr)_{(t,s)\in\mathcal P},
\end{equation}
where each \(H_{t,s}(x)\in\mathbb R^{T\times T}\) is a symmetric matrix depending only on the treatment assignment probabilities \(\pi_a(x)\), \(a\in[T+1]\).
\end{proposition}
Proposition~\ref{prop:offdiag_moment} shows that the off-diagonal residual covariances depend on the CATE vector \(\tau(x)\) through quadratic forms. 
They encode products of lag-specific effects. Whether these
products identify the full vector depends on the design and
the values of these effects.
Since Assumption~\ref{assumption:offdiagonal} removes the off-diagonal error covariance, any remaining covariance between residuals at different time points must come from the shared treatment-timing indicators and the lag-specific treatment effects. This turns estimation of the lagged CATE vector into the following moment-matching problem.

\subsubsection{Estimators}

To approximate \eqref{equ:key_identity}, we estimate \(c(x)\) by
\(\hat c(x)\), obtained by regressing
\(\hat R_{i,t}\hat R_{i,s}\) on \(X_i\) for all
\((t,s)\in\mathcal P\). For a prespecified finite \(M>0\), define
\[
\Theta:=\{u\in\mathbb R^T:\|u\|_2\le M\}.
\]
We obtain the pointwise nonlinear least-squares (NLS) estimator by
minimizing the following penalized objective over the compact set \(\Theta\):
\begin{align}
\ell^{\,\mathrm{nls}}(u;\hat c,x,\rho)
&:=\|\hat c(x)-F_x(u)\|_2^2+\rho\|u\|_2^2,
\qquad \rho\ge0, \notag\\[7pt]
\hat\tau_{\pm}^{\,\mathrm{nls}}(x)
&\in\arg\min_{u\in\Theta}
\ell^{\,\mathrm{nls}}(u;\hat c,x,\rho).
\label{eq:opt2_nls_pointwise}
\end{align}
The moment map takes the same value at \(u\) and \(-u\), so it
cannot identify their common sign. This is the same symmetry that
arises in other quadratic inverse problems
\citep{shechtman2015phase,candes2013phaselift,candes2015wirtinger}.
The next subsection gives conditions under which these are the
only two matching vectors. The sign is common to all lag
coordinates at a fixed \(x\); it may differ across covariate values.

The NLS objective is nonconvex. Numerical methods such as
Gauss--Newton or Levenberg--Marquardt may find only a local
minimum. We use the following convex relaxation to construct an
initial vector for NLS. Define
\[
\mathcal A_x(B)
:=
\bigl(\operatorname{tr}\{H_{t,s}(x)B\}\bigr)_{(t,s)\in\mathcal P}.
\]
For \(B=uu^\top\), we have \(\mathcal A_x(B)=F_x(u)\), so we estimate
\(B\) by minimizing the following convex objective with trace penalty \(\lambda>0\):
\begin{align}
\ell^{\,\mathrm{cvx}}(B;\hat c,x,\lambda)
&:=\|\hat c(x)-\mathcal A_x(B)\|_2^2+\lambda\operatorname{tr}(B), \notag\\[7pt]
\hat B(x)
&\in\arg\min_{B\succeq0}\ell^{\,\mathrm{cvx}}(B;\hat c,x,\lambda).
\label{eq:opt2_convex_pointwise}
\end{align}
The trace penalty encourages a low-rank solution. The program
has \(T(T+1)/2\) free matrix entries and can be solved using a
standard semidefinite programming solver. From a leading
eigenpair \((\hat\lambda_1(x),\hat u_1(x))\), form the spectral
initialization
\begin{equation}\label{eq:opt2_tau_sp}
\hat\tau_{\pm}^{\,\mathrm{sp}}(x)
:=
\sqrt{\hat\lambda_1(x)}\,\hat u_1(x).
\end{equation}
Its Euclidean projection onto \(\Theta\) is a feasible initial
vector for the constrained NLS problem.
Optimization and eigenvector tie rules are fixed without using
the realized assignments. The spectral vector provides an
initialization, while the consistency result below concerns a
global minimizer of the constrained NLS problem
\eqref{eq:opt2_nls_pointwise}. It does not certify that a
particular numerical solver has found such a minimizer.

\subsubsection{Identification and consistency}
\label{sect:identify_consistent}

The pointwise identification results below use the population relation
\(c(x)=F_x(\tau(x))\), which follows from Assumptions~\ref{ass:additive_outcome_model}
and \ref{assumption:offdiagonal}. Identification up to sign means that the moments
determine \(\{\tau(x),-\tau(x)\}\), with one common sign across all lag coordinates at
a given \(x\). We write the Jacobian of the moment map \(F_x\) with respect to \(u\) as
\[
J_x(u):=\frac{\partial F_x(u)}{\partial u^\top}
\in\mathbb R^{|\mathcal P|\times T}.
\]
\begin{proposition}[Local identification up to sign]
\label{prop:opt2_identification}
Fix \(x\in\mathcal X\). Suppose \(\tau(x)\ne0\) and \(\operatorname{rank}\{J_x(\tau(x))\}=T\).
Then there exists an open neighborhood \(\mathcal N_x\) of \(\tau(x)\) such that
\[
\{u\in\mathcal N_x\cup(-\mathcal N_x):F_x(u)=c(x)\}
=\{\tau(x),-\tau(x)\}.
\]
The vectors, \(\tau(x)\) and \(-\tau(x)\),
give the same matrix \(B(x)=\tau(x)\tau(x)^\top\).
\end{proposition}

The rows of \(J_x(u)\) are \(2\{H_{t,s}(x)u\}^\top\), so full column rank is equivalent to
\[
\operatorname{span}\{H_{t,s}(x)\tau(x):(t,s)\in\mathcal P\}
=\mathbb R^T.
\]
The condition \(T\ge3\) provides at least as many off-diagonal moments as lag coordinates, but does
not ensure full column rank. The smallest singular value of \(J_x(\hat\tau_{\pm}^{\,\mathrm{nls}}(x))\)
is a useful diagnostic of local identifiability. A value well separated from zero provides evidence for local identification near the fitted solution,
but not for global identification:
other vectors, apart from the sign reversal, may have the same moments.

For global identification up to sign, we use the first-column moments in
\(c(x)\) to reduce the problem to one scalar. Fix \(x\in\mathcal X\), and assume
\begin{equation}\label{eq:global_id_conditions}
c(x)=F_x(\tau(x)),\quad
0<\pi_1(x)<1,\quad
c_{t,1}(x)\ne0\ \text{for some }t\ge2.
\end{equation}
We suppress \(x\) in the following definitions. The first-column moments link every lag effect to
the lag-0 effect. For any candidate vector \(u\) satisfying \(F_x(u)=c(x)\), 
\[
c_{l+1,1}/\pi_1
=
(1-\pi_1)u_0u_l
-\sum_{a=2}^{l+1}\pi_a u_0u_{l+1-a},
\qquad l=1,\ldots,T-1.
\]
The nonzero moment assumption implies \(u_0\ne0\). Set \(z=u_0^2>0\) and write
\(d_l(z)=u_0u_l\). Once \(z\) is specified, the equations above determine
\(d_1(z),\ldots,d_{T-1}(z)\) successively. Starting from \(d_0(z)=z\), we obtain the recursion
\begin{equation}\label{eq:scalar_id_recursion}
d_l(z)
:=
(1-\pi_1)^{-1}\Big[
c_{l+1,1}/\pi_1+
\displaystyle\sum_{a=2}^{l+1}\pi_a d_{l+1-a}(z)\Big]
\qquad l=1,\ldots,T-1.
\end{equation}
The recursion shows that every \(d_l(z)\) is affine in \(z\). Since
\(d_l(z)=u_0u_l\) and \(u_0=\pm\sqrt z\), 
\[
u_l=\frac{d_l(z)}{u_0}
=\pm\frac{d_l(z)}{\sqrt z},
\qquad l=0,\ldots,T-1,
\]
with the same sign for all coordinates. Writing \(d(z)=(d_0(z),\ldots,d_{T-1}(z))^\top\), we obtain
\(u=\pm d(z)/\sqrt z\). Substituting this expression into the remaining moment equations \(u^\top
H_{t,s}u=c_{t,s}\) and multiplying by \(z>0\) shows that \(z\) must satisfy \(P_{t,s}(z)=0\), where
\begin{equation}\label{eq:scalar_id_polynomial}
P_{t,s}(z)
:=
d(z)^\top H_{t,s}d(z)-zc_{t,s}
=A_{t,s}z^2+B_{t,s}z+C_{t,s},
\qquad 1<s<t\le T.
\end{equation}
Let \(\mathcal Z_x\) denote the set of positive common roots of these polynomials:
\[
\mathcal Z_x
:=
\{z>0:P_{t,s}(z)=0\text{ for every }1<s<t\le T\}.
\]

\begin{theorem}[Global identification]
\label{prop:global_id}
Under the conditions in \eqref{eq:global_id_conditions},
\[
\{u\in\mathbb R^T:F_x(u)=c(x)\}
=
\left\{\pm\frac{d(z)}{\sqrt z}:z\in\mathcal Z_x\right\}.
\]
Thus \(c(x)\) identifies \(\{\tau(x),-\tau(x)\}\) if \(|\mathcal Z_x|=1\).
\end{theorem}
Let \((A_j,B_j,C_j)\) and \((A_k,B_k,C_k)\) be the coefficients for two distinct moment pairs \(j,k\) in
\eqref{eq:scalar_id_polynomial}. Their polynomial equations form a linear system in \(z^2\) and \(z\):
\[
\begin{pmatrix}
A_j & B_j\\
A_k & B_k
\end{pmatrix}
\begin{pmatrix}
z^2\\
z
\end{pmatrix}
=
-\begin{pmatrix}
C_j\\
C_k
\end{pmatrix}.
\]
A sufficient condition for uniqueness is \(A_jB_k-A_kB_j\ne0\), which makes the coefficient
matrix invertible. Since the population relation guarantees that \(z=\tau_0(x)^2>0\)
satisfies all the polynomial equations, it is the unique common positive root, given by
\[
z=\frac{A_kC_j-A_jC_k}{A_jB_k-A_kB_j}
=\tau_0(x)^2.
\]
The determinant condition is sufficient but not necessary. If any one polynomial \(P_{t,s}\)
has exactly one positive root, that root must be \(\tau_0(x)^2\), so \(|\mathcal Z_x|=1\).

The next result establishes global identification up to sign for geometrically
decaying effects under the stated assignment distribution. We consider this
class because uniqueness of the scalar solution can be verified explicitly.
\begin{proposition}[Geometric effects]
\label{prop:rank_sufficient_design}
Fix \(T\ge4\) and \(x\in\mathcal X\). Suppose
\[
\pi_1(x)=\cdots=\pi_T(x)=q,\qquad
\pi_{T+1}(x)=1-Tq,
\qquad q\in(0,1/T).
\]
Suppose the true lagged CATE vector satisfies
\[
\tau_l(x)=\theta(x)\rho^l,\qquad l=0,\ldots,T-1,
\qquad
\theta(x)\ne0,\quad \rho\in(0,1).
\]
Then \(\tau(x)\) is globally identified up to sign from \(c(x)\), that is,
\[
\{u\in\mathbb R^T:F_x(u)=c(x)\}
=\{\tau(x),-\tau(x)\}.
\]
\end{proposition}
These identification results do not by themselves guarantee recovery by the convex relaxation in
\eqref{eq:opt2_convex_pointwise}. We use its solution to construct an initial value for the NLS estimator
in \eqref{eq:opt2_nls_pointwise}. For consistency, we keep \(T\), the known assignment probabilities,
and the radius \(M\) of \(\Theta\) fixed as \(N\) increases. We use the sign-invariant distance
\begin{equation}\label{equ:unsigned_dist}
d_{\pm}(u,v):=\min\{\|u-v\|_2,\|u+v\|_2\}.
\end{equation}

\begin{assumption}\label{assump:opt2_moment_cons}
As \(N\to\infty\), assume that
\begin{enumerate}
    \item for any \(x\in\mathcal X\), \(\hat c(x)\xrightarrow{\;p\;} c(x)\);
    \item for any compact \(\mathcal X_0\subset\mathcal X\),
    \(
    \sup_{x\in\mathcal X_0}\|\hat c(x)-c(x)\|_2 \xrightarrow{\;p\;} 0.
\)
\end{enumerate}
\end{assumption}

\begin{assumption}[Regularity and boundedness]\label{assump:opt2_reg}
Fix a compact \(\mathcal X_0\subset\mathcal X\). Assume:
\begin{enumerate}
\item \(H_{t,s}(x)\) is continuous in \(x\) for all \((t,s)\in\mathcal P\);
\item \(\sup_{x\in\mathcal X_0}\max_{(t,s)\in\mathcal P}\|H_{t,s}(x)\|_{\mathrm{op}}<\infty\);
\item \(\sup_{x\in\mathcal X_0}\|\tau(x)\|_2\le M<\infty\).
\end{enumerate}
\end{assumption}

\begin{theorem}
\label{thm:opt2_nls_consistency}
Fix \(x\in\mathcal X_0\). Suppose Assumptions~\ref{assump:opt2_moment_cons}
and \ref{assump:opt2_reg} hold, and that
\[
\{u\in\Theta:F_x(u)=c(x)\}
=\{\tau(x),-\tau(x)\}.
\]
Then any global minimizer \(\hat\tau_{\pm}^{\,\mathrm{nls}}(x)\) in
\eqref{eq:opt2_nls_pointwise}, with \(\rho=\rho_N\ge0\) and \(\rho_N\to0\), satisfies
\[
d_{\pm}\bigl(\hat\tau_{\pm}^{\,\mathrm{nls}}(x),\tau(x)\bigr)
\xrightarrow{\;p\;}0.
\]
\end{theorem}
The identification premise holds under \eqref{eq:global_id_conditions} when \(|\mathcal Z_x|=1\),
or under the conditions of \Cref{prop:rank_sufficient_design}. Consistency also requires consistent
moment estimation and a global solution of the constrained NLS problem. A spectral initialization
can help numerical optimization, but an arbitrary local minimum need not satisfy the theorem.

\subsection{Summary of assumptions}
\label{subsec:unsigned_assumptions_summary}

Unsigned estimation assumes finite conditional second moments, no interference,
no anticipation, and the additive model in Assumption~\ref{ass:additive_outcome_model}.
This model requires the effect at each lag to be unchanged across treatment start times.

The diagonal estimator requires lag-invariant effects, equal conditional error
variances across time, and variation in \(v_t(x)\), but allows serial correlation
(\Cref{prop:diagonal_identification}). The off-diagonal estimator allows effects
and error variances to vary across time, but requires zero error covariances
between distinct times conditional on the covariates.

Consistency of the constrained NLS estimator requires global identification up to
sign (\Cref{prop:global_id,prop:rank_sufficient_design}), the moment-consistency
and regularity conditions in \Cref{assump:opt2_moment_cons,assump:opt2_reg},
a vanishing ridge penalty, and a global minimizer over \(\Theta\).
The convex and spectral estimates are used to initialize NLS.
These conditions concern effect estimation. Randomization validity follows from
using invariant fitting and scoring inputs and the correct assignment distribution,
as described in Section~\ref{sec:rt}.

\section{CATE-assisted randomization tests}
\label{subsec:te_in_rt}

\subsection{Orienting the unsigned CATE estimator}
\label{sect:sign_cate}

Let \(\hat\tau_{\pm}(x)
=(\hat\tau_{\pm,0}(x),\ldots,\hat\tau_{\pm,T-1}(x))^\top\)
denote an unsigned CATE estimator. This may be the diagonal estimator
\(\hat\tau_{\pm}^{\mathrm{diag}}(x)\) from \Cref{subsubsec:static_est},
or the off-diagonal NLS estimator \(\hat\tau_{\pm}^{\,\mathrm{nls}}(x)\).
For the global null, we choose a sign for each unit \(i\) and candidate
treatment start time \(a\) by minimizing the squared error in the
residual equation \eqref{eq:yit_expression}:
\begin{equation}
\label{equ:s_x}
\hat s_i(a)
\in
\arg\min_{s\in\{-1,1\}}
\left\{
\sum_{t=1}^{T}
\left[
\hat R_{i,t}
-
\sum_{l=0}^{t-1}
\Bigl(
\mathbf 1_{\{a=t-l\}}-\pi_{t-l}(X_i)
\Bigr)
s\,\hat\tau_{\pm,l}(X_i)
\right]^2
\right\}.
\end{equation}
For each unit and candidate start time, we compare the two losses
for \(s=-1\) and \(s=1\). Let \(\tilde A_i^{(1)},\ldots,\tilde A_i^{(M)}\) be the simulated
start times for unit \(i\in[N]\). The signed CATE estimates under
the observed and simulated assignments are
\[
\hat\tau_i=\hat s_i(A_i)\hat\tau_{\pm}(X_i),
\qquad
\hat\tau_i^{(m)}
=\hat s_i(\tilde A_i^{(m)})\hat\tau_{\pm}(X_i),
\quad m\in[M].
\]
The next subsection uses these estimates to construct CATE-assisted
test statistics. The following proposition bounds the probability
of selecting the wrong sign under the observed assignment when
the residuals and unsigned CATEs are known.

Suppose the outcome model in Assumption~\ref{ass:additive_outcome_model}
holds. Let \(\tau_\pm(x)\) be a fixed representative of the CATE vector
up to sign, so that \(\tau(x)=s(x)\tau_\pm(x)\) for
\(s(x)\in\{-1,1\}\), and write \(s_i=s(X_i)\).
Assume that, conditional on \(X_i,A_i\), \(v^\top\epsilon_i\) is
mean-zero sub-Gaussian with variance proxy \(\sigma^2\|v\|_2^2\)
for every \(v\in\mathbb R^T\), where \(\sigma>0\).

For a possible observed start time \(a\), define
\(\nu_i(a):=\Pi(X_i,a)\tau_\pm(X_i)\), where \(\Pi(X_i,a)\)
is the assignment matrix defined after \eqref{eq:yit_expression},
evaluated at \(A_i=a\). Thus \(\nu_i(a)\) is the treatment contribution
to the residual vector implied by \(\tau_\pm(X_i)\).
When \(A_i=a\), the true residual satisfies
\(R_i=s_i\nu_i(a)+\epsilon_i\), so sign selection compares
the two possible residual means \(\nu_i(a)\) and \(-\nu_i(a)\).
Let \(\hat s_i^{*}(a)\) denote the sign selected by
\eqref{equ:s_x} using the true residual \(R_i\),
\(\tau_\pm(X_i)\), and \(A_i=a\),
with a fixed rule for resolving ties.

\begin{proposition}[Oracle sign recovery]
\label{prop:gaussian_sign_recovery}
For every \(a\) with \(\pi_a(X_i)>0\),
\begin{equation}\label{eq:sign_margin_bound}
\mathbb P\!\left\{\hat s_i^{*}(a)\ne s_i
\mid X_i,A_i=a\right\}
\le
\exp\!\left\{-\frac{\|\nu_i(a)\|_2^2}{2\sigma^2}\right\}.
\end{equation}
If \(\epsilon_i\mid X_i,A_i=a\sim N(0,\sigma^2 I_T)\) and
\(\|\nu_i(a)\|_2>0\), then
\begin{equation}\label{eq:sign_gaussian_error}
\mathbb P\!\left\{\hat s_i^{*}(a)\ne s_i
\mid X_i,A_i=a\right\}
=
\Phi\!\left(-\frac{\|\nu_i(a)\|_2}{\sigma}\right).
\end{equation}
\end{proposition}

Here \(\Phi\) is the standard normal distribution function.
Both the upper bound and the exact Gaussian error probability
decrease as \(\|\nu_i(a)\|_2/\sigma\) increases.
To illustrate this, consider a balanced design with
\(\pi_c(x)=1/(T+1)\) for \(c\in[T+1]\), and suppose
\(\tau_{\pm,l}(X_i)=\theta(X_i)\rho^l\), where
\(\theta(X_i)>0\) and \(0<\rho<1\).
For \(a\le T\), we have
\(\|\nu_i(a)\|_2\ge\theta(X_i)\{1-a/(T+1)\}\);
see the end of the proof of \Cref{prop:gaussian_sign_recovery}
in Appendix~\ref{sect:proof_sign_recover} for a derivation.
Thus the bound \eqref{eq:sign_margin_bound} gives
\begin{equation}\label{eq:sign_geometric_bound}
\mathbb P\!\left\{\hat s_i^{*}(a)\ne s_i
\mid X_i,A_i=a\right\}
\le
\exp\!\left\{
-\frac{\theta(X_i)^2}{2\sigma^2}
\left(1-\frac{a}{T+1}\right)^2
\right\}.
\end{equation}
A higher signal-to-noise ratio  \(\theta(X_i)^2/\sigma^2\) reduces the upper bound \eqref{eq:sign_geometric_bound} on sign error. Earlier treatment starting time $a$ gives a smaller upper bound and
provides more post-treatment outcomes for assessing whether the selected
sign fits the outcome trajectory. When \(a\) is close to \(T\), fewer
post-treatment outcomes are available, and the bound gives a weaker
sign-recovery guarantee.

\paragraph{Sign selection for a lag-specific test.}
For the lag-specific null \(H_{0,t,l}\) in \eqref{eq:H_0tl},
we use only the target-time residual in place of the
full-trajectory loss in \eqref{equ:s_x}.
Assume \(N_t>0\) and \(N_{>t+l}>0\).
We fit the outcome mean \(\hat\mu_{\mathcal I,t+l}\) on the
eligible units \(\mathcal I\) defined in \eqref{equ:set_I},
using their covariates and outcomes at time \(t+l\).
Set \(\hat r_i=Y_{i,t+l}-\hat\mu_{\mathcal I,t+l}(X_i)\) and
\(e_{t,l}=N_t/|\mathcal I|=N_t/(N_t+N_{>t+l})\).
Let \(\hat m_i\ge0\) be a fitted magnitude based on the data
\(\mathcal D_{t,l}\) in \eqref{eq:lag_invariant_data}, as required by
\Cref{prop:lag_conditional_validity}.
For \(i\in\mathcal I\) and \(a\in\mathcal A_{t,l}\), use
\begin{equation}\label{eq:lag_target_sign}
\hat s_i(a)\in\arg\min_{v\in\{-1,1\}}
\{\hat r_i-(\mathbf1_{\{a=t\}}-e_{t,l})v\hat m_i\}^2.
\end{equation}
The fitted outcome means and magnitudes remain fixed across randomization
draws, while the signs are recomputed under each candidate assignment.

\subsection{CATE-assisted test statistics}\label{sect:cate_stat}

Let
\(
\hat\eta=
(\hat\mu_1,\ldots,\hat\mu_T,
\hat\tau_{\pm,0},\ldots,\hat\tau_{\pm,T-1},
\pi_1,\ldots,\pi_{T+1})
\)
denote the collection of fixed nuisance functions used
to construct the test statistic in the randomization test.
For each candidate treatment start time \(a\in[T+1]\), define
\begin{equation}\label{eq:candidate_signed_cate}
\hat\tau(X_i;a)
:=\hat\tau_i(a)
=\hat s_i(a)\hat\tau_\pm(X_i).
\end{equation}
In the formulas below, we write \(\hat\tau(X_i)\) for
\(\hat\tau(X_i;a)\) in \eqref{eq:candidate_signed_cate},
with \(a\) denoting the candidate treatment start time.
The signed CATE and fitted control means are recomputed
for each candidate assignment in a randomization test, while the unsigned CATE
estimate \(\hat\tau_\pm(X_i)\) remains fixed. Define
\[
\hat\mu_{0,t}(X_i)
:=
\hat\mu_t(X_i)
-
\sum_{s=0}^{t-1}\pi_{t-s}(X_i)\hat\tau_s(X_i),
\quad t\in[T],
\]
and the lag-\(l\) treated outcome functions by
\[
\hat\mu_{1,t+l}(X_i)
:=
\hat\mu_{0,t+l}(X_i)+\hat\tau_l(X_i),
\quad l\in\{0,\dots,T-t\}.
\]
Let \(\psi_{i,t,l}(O_i;\hat\eta)\) be a unit-level score
designed to detect the lag-\(l\) effect.
For example, one can use the augmented inverse probability weighting (AIPW) score \citep{robins1994estimation} for the average lag-\(l\) treatment effect at time \(t\):
\[
\psi_{i,t,l}^{\mathrm{AIPW}}(O_i;\hat\eta)
:=
\hat\tau_l(X_i)
+\frac{\mathbf 1_{\{A_i=t\}}}{\pi_t(X_i)}
\Bigl\{Y_{i,t+l}-\hat\mu_{1,t+l}(X_i)\Bigr\}
-\frac{\mathbf 1_{\{A_i>t+l\}}}{1-\pi_{\le t+l}(X_i)}
\Bigl\{Y_{i,t+l}-\hat\mu_{0,t+l}(X_i)\Bigr\}.
\]
Here, units with \(A_i=t\) form the treated group at lag \(l\), while units with \(A_i>t+l\) remain untreated at time \(t+l\) and serve as controls.

\begin{algorithm}[t]

\caption{CATE-assisted randomization test for the global null}
\label{alg:rc}
\vspace{3pt}
\begin{minipage}{\textwidth}
\begin{enumerate}[leftmargin=*,label=\arabic*.]
    \item For each \(t\in [T]\), fit \(\hat\mu_t(\cdot)\) by regressing \(Y_{i,t}\) on \(X_i\), and compute the residuals
    \[
    \hat R_{i,t}:=Y_{i,t}-\hat\mu_t(X_i),
    \quad i\in[N].
    \]
    \item Estimate the residual covariance matrix \(\Sigma_R(x)\) in \eqref{equ:second_R} by regressing  \(\hat{R}_{i,t}\hat{R}_{i,s}\) on $X_i$:
    \[
    \hat c_{t,s}(x)
    =
    \widehat{\mathbb E}[\hat R_{i,t}\hat R_{i,s}\mid X_i=x],
    \quad t,s\in[T].
    \]
    \item Estimate the CATE vector up to a common sign across lags at each covariate value:
    \begin{enumerate}[leftmargin=*,label=(\alph*)]
        \item If the effect is lag-invariant, use the diagonal estimator \(\hat\tau_{\pm}^{\mathrm{diag}}(x)\) in \eqref{eq:tau_diagonal}.
        \item If the effect varies across lags, use the spectral fit in \eqref{eq:opt2_tau_sp} to initialize the NLS fit in \eqref{eq:opt2_nls_pointwise}. The consistency result applies to a global NLS minimizer.
    \end{enumerate}
\item
For each candidate start time \(a\), compute \(\hat s_i(a)\)
and \(\hat\tau(X_i;a)\) using \eqref{equ:s_x} and
\eqref{eq:candidate_signed_cate}.

\item
Compute the \(p\)-value in \eqref{equ:p_value}
using a CATE-assisted statistic from \Cref{sect:cate_stat}.
\end{enumerate}
\end{minipage}
\vspace{2pt}
\end{algorithm}

Alternatively, suppose we want a test statistic that is more directly sensitive to heterogeneous treatment effects across units. Motivated by a Gaussian working model with common variance, we can use the following likelihood-ratio score for the alternative \(\tau_l(X_i)=\hat\tau_l(X_i)\) against the null \(\tau_l(X_i)=0\):
\begin{equation}\label{equ:lr_statistic}
\begin{split}
\psi_{i,t,l}^{\mathrm{LR}}(O_i;\hat\eta)
&= -
\bigl[ Y_{i,t+l} - \hat\mu_{0,t+l}(X_i)-
\mathbf 1_{\{A_i=t\}}\hat\tau_l(X_i)\bigr]^2
+
\bigl[Y_{i,t+l}-\hat\mu_{0,t+l}(X_i)\bigr]^2
\\[5pt]
&=
\mathbf 1_{\{A_i=t\}}\left\{2\hat\tau_l(X_i)
\bigl[
Y_{i,t+l}-\hat\mu_{0,t+l}(X_i)
\bigr]-
\hat\tau_l^2(X_i)
\right\}.
\end{split}
\end{equation}
For comparison, in experiments we also consider a covariate-adjusted (CA)
mean-difference statistic that uses the outcome model but not the CATE
estimates. Using the cohort counts defined in \Cref{sect:sad}, we define
\begin{equation}\label{equ:ca_statistic}
\begin{split}
\psi_{i,t,l}^{\mathrm{CA}}(O_i;\hat\eta)
&=
\frac{\mathbf1_{\{A_i=t\}}}{N_t/N}
\bigl[Y_{i,t+l}-\hat\mu_{t+l}(X_i)\bigr]-
\frac{\mathbf1_{\{A_i>t+l\}}}{N_{>t+l}/N}
\bigl[Y_{i,t+l}-\hat\mu_{t+l}(X_i)\bigr].
\end{split}
\end{equation}
Substituting this score into \eqref{equ:statistic_form} gives an estimate of the average treatment effect, using the
difference between the average residuals of units starting treatment
at time \(t\) and those remaining untreated at time \(t+l\).
The outcome models in the score are fitted 
without using the realized assignments. This provides a
benchmark for assessing the additional gain of CATE-assisted randomization test.

To detect the lag-\(l\) effect at adoption time \(t\),
we use a statistic of the form
\begin{equation}\label{equ:statistic_form}
T_{t,l}(O_{[N]};\hat\eta) =
\frac{1}{N}\sum_{i=1}^N \psi_{i,t,l}(O_i;\hat\eta),
\end{equation}
To test the global sharp null \(H_0\) in \eqref{eq:H_0}, we use the
same aggregation for the scores:
\begin{equation}\label{equ:global_statistic}
T(O_{[N]};\hat\eta)
=
\sum_{l=0}^{T-1}w_l
\sum_{t=1}^{T-l}T_{t,l}(O_{[N]};\hat\eta),
\qquad
w_l:=\frac{1}{T(T-l)}.
\end{equation}

\subsection{Power comparison}
\label{subsec:cate_gaussian_power}

We compare three oracle tests for the global sharp null \(H_0\):
the LR test based on the true outcome likelihood with signed CATEs,
the sign-selected LR test using the true outcome means and unsigned
CATEs, and the CA test using the true outcome means.
The latter two use the scores in \eqref{equ:lr_statistic} and
\eqref{equ:ca_statistic}, respectively, with the aggregation in
\eqref{equ:global_statistic}. The sign-selected test applies
\eqref{equ:s_x} separately under every candidate assignment
without using the true signs.
By Bayes' rule and the Neyman--Pearson lemma, the signed-oracle
likelihood test is most powerful against each specified
distribution \(P\) among tests of conditional level at most
\(\alpha\) given \(X_{[N]}\) and \(Y_{[N],[T]}\), for every
\(N\).\footnote{The finite-sample optimality statement allows
randomization at the critical value; the asymptotic results
also hold for nonrandomized randomization tests.}

\paragraph{Geometric effects.}
Retain the balanced design with fixed \(T\ge3\) and
\(N_c=N/(T+1)\) for every \(c\in[T+1]\), taking \(N\) through
multiples of \(T+1\). The covariates \(X_i\) are independent and
identically distributed, independently of assignment.
Under the model in \Cref{ass:additive_outcome_model}, assume
square-integrable control means and independent \(N(0,\sigma^2)\)
unit--time errors, independent of covariates and assignment. Here
\(\sigma>0\) is fixed, and the same error is used across potential
adoption times. Suppose
\begin{equation}\label{eq:sign_geometric_effects}
g_l(x)=u(x)\rho^l,\quad l=0,\ldots,T-1,\quad
0<\rho\le\tfrac12,\quad \mathbb E [u(X)]=0,\quad \mathbb E [u(X)^2]=1,
\end{equation}
where \(u\) is fixed, bounded, and measurable. Consider
\begin{equation}\label{eq:sign_local_effects}
\tau_{l,N}(x)=\{\delta_N+\theta_Nu(x)\}\rho^l,
\qquad l=0,\ldots,T-1,
\end{equation}
where \(\delta_N\in\mathbb R\), \(\theta_N>0\), and
\[
\frac{\sqrt N\,|\delta_N|}{\sigma}\to0,\qquad
\frac{\theta_N}{\sigma}\to0,\qquad
\frac{\sqrt N\,\theta_N^2}{\sigma^2}
\to h\in[0,\infty).
\]
Each individual's effect falls in magnitude by at least one half
in each successive period. At lag \(l\), the average effect is
\(\delta_N\rho^l\) and the CATE variance is
\(\theta_N^2\rho^{2l}\).
Both \(\delta_N\) and \(\theta_N\) tend to zero, so the treatment
effects vanish as \(N\) increases.

In this balanced design, the assignment matrix \(\Pi(x,c)\)
defined after \eqref{eq:yit_expression} does not depend on the
covariates \(x\). We write it as \(\Pi(c)\) for
treatment start time \(c\) and define
\[
A_g=(T+1)^{-1}\sum_{c=1}^{T+1}
\mathbb E[\|\Pi(c)g(X)\|_2^2]>0.
\]
Let \(B_g>0\) be the standardized quadratic mean response of the
sign-selected LR statistic. Its exact definition is given in the
proof, which also establishes its positivity under the preceding
conditions.
All three tests use upper-tail calibration as in
\eqref{equ:p_value}. Denote their powers by
\(\beta_{\mathrm{LR},N}^{*}\), \(\beta_{\mathrm{LR},N}^{\pm}\),
and \(\beta_{\mathrm{CA},N}^{*}\), respectively.
Let \(\Phi\) be the standard normal distribution function and
\(z_{1-\alpha}=\Phi^{-1}(1-\alpha)\).

\begin{theorem}[Power comparison]
\label{thm:oracle_power}
\label{thm:sign_selected_local_power}
Under the balanced design and Gaussian outcome assumptions above,
suppose the treatment effects satisfy
\eqref{eq:sign_geometric_effects}--\eqref{eq:sign_local_effects}
and the stated rate conditions on \(\delta_N\) and \(\theta_N\).
Then, for every fixed \(\alpha\in(0,1)\), 
\[
\beta_{\mathrm{LR},N}^{*}
=
1-\Phi(z_{1-\alpha}-\sqrt{A_gN}\,\theta_N/\sigma)+o(1),
\]
\[
\beta_{\mathrm{LR},N}^{\pm}
=
1-\Phi(z_{1-\alpha}-B_g\sqrt N\,\theta_N^2/\sigma^2)+o(1),
\qquad
\beta_{\mathrm{CA},N}^{*}=\alpha+o(1).
\]
\end{theorem}

The limiting power of the CA test is \(\alpha\), because the
average effects are negligible on the \(N^{-1/2}\) scale.
If the limit \(h>0\), the sign-selected LR test has limiting power
\(1-\Phi(z_{1-\alpha}-hB_g)>\alpha\), giving a strict
power gain over CA. When \(h=0\), its limiting power is also
\(\alpha\). In this case, effects may still vary across units,
but \(\theta_N=o(\sigma N^{-1/4})\), so this variation is too
weak to affect its limiting power.
Thus, under the stated assumptions, the
sign-selected oracle LR test can detect heterogeneity on the
\(N^{-1/4}\) scale even when the average effects are too small
for the oracle CA test to detect.

\section{Numerical experiments}\label{sect:exp}

This section evaluates the size and power of the CATE-assisted
randomization tests in simulated staggered-adoption experiments, examines the convergence of the unsigned CATE estimators and their use as warm starts for estimating signed CATEs, and applies the tests
to the Chicago procedural-justice training program.
For the staggered-adoption size and power experiments,
we compute rejection rates at nominal level \(\alpha=0.1\)
over 1000 replications, with each \(p\)-value computed from
1000 randomization draws. Shaded bands in the corresponding
figures show the rejection rate plus or minus two Monte Carlo
standard errors.
\Cref{subsec:simulations_additional} reports additional
size and power comparisons under parallel-group and crossover designs.
Code to reproduce the experiments is available at
\url{https://github.com/fangnanzheng/fit-cate-once}.

\subsection{Synthetic data experiments}
\label{subsec:synthetic_data}

We consider two settings: lag-invariant effects, for which
\(\tau_l(x)\) is the same at every lag, and lagged effects,
for which \(\tau_l(x)\) varies with \(l\).
Both settings use the completely randomized staggered-adoption
design of \Cref{sect:sad}, with \(T=5\) periods, sample sizes
\(N\in\{100,200,300,400,500\}\), and cohort proportions
\[
\pi=(\pi_1,\ldots,\pi_6)=(0.10,0.15,0.15,0.15,0.15,0.30),
\]
where \(\pi_6\) is the never-treated share. All cohort counts
\(N_a=N\pi_a\) are integers. Outcomes follow the additive model
\eqref{eq:yit} with untreated mean
\[
\mu_{0,t}(X_i)
=(0.7+0.1t)(-0.3X_{i,1}+0.2X_{i,2}),
\qquad t\in[5],
\]
and Gaussian error vectors
\(\epsilon_i\sim N_5(0,\Sigma_\epsilon)\), independent of
\((X_i,A_i)\). The error covariance is specified separately
for each setting below.

\paragraph{Lag-invariant setting.}
Each unit has covariates \(X_i=(X_{i,1},X_{i,2},X_{i,3})\),
where \(X_{i,1}\sim\mathrm{Unif}[-1,1]\),
\(X_{i,2}\) takes values \(-1\) and \(1\) with equal probability,
and \(X_{i,3}=X_{i,2}V_i\).
Here \(V_i\in\{-1,1\}\) with \(\mathbb P(V_i=1)=0.53\),
and \(X_{i,1}\), \(X_{i,2}\), and \(V_i\) are mutually
independent. For the size experiments, we set \(\tau_l\equiv0\) at every
lag in the lag-invariant setting, so that \(H_0\) holds.
For the power experiments, the effect at each lag
\(l=0,\ldots,4\) is invariant:
\[
\tau_l(X_i)
=
\bigl(0.80+0.04X_{i,1}+0.09\max(X_{i,1},0)^2\bigr)
X_{i,2}X_{i,3}
+0.25X_{i,2}+0.04X_{i,1}.
\]
The error covariance is
\(\Sigma_\epsilon=\bigl((-0.8)^{|t-s|}\bigr)_{t,s=1}^{5}\),
so each error has variance one.

\paragraph{Lagged setting.}
Each unit has covariates \(X_i=(X_{i,1},X_{i,2})\), where
\(X_{i,1}\) and \(V_i\) are independent variables taking values
in \(\{-1,1\}\), with
\(\mathbb P(X_{i,1}=1)=\mathbb P(V_i=1)=0.575\), and
\(X_{i,2}=X_{i,1}V_i\).
For the power experiments, we set
\[
\tau_l(X_i)
=
\Bigl(0.5+\frac{l(l+1)}{40}\Bigr)
\Bigl(X_{i,1}X_{i,2}+\frac{3l}{40}X_{i,1}\Bigr),
\qquad l=0,\ldots,4.
\]
The error covariance is \(\Sigma_\epsilon=0.6^2I_5\),
so each error has variance \(0.36\).
Each unit's effect has the sign of \(V_i=X_{i,1}X_{i,2}\)
at every lag, and its magnitude increases with the lag.
For the size experiments,we set the tested lag effect
\(\tau_l\equiv0\) while retaining nonzero effects at the
other lags, so that \(H_{0,t,l}\) holds while \(H_0\) does not.

\subsubsection{Size and Power}
\label{subsubsec:size_and_power}

We compare four randomization tests (RTs) in simulated staggered-adoption experiments: one without covariate adjustment and three with different forms of adjustment. All tests use the statistic form in \eqref{equ:statistic_form},
with different scores and fitting procedures. The nuisance models and the samples on which they are fitted are specified in Appendix~\ref{app:sim_details}.
\begin{itemize}
    \item \textbf{RT (DM)} uses the difference between the treated and not-yet-treated means:
    \[
    \psi_{i,t,l}^{\mathrm{DM}}(O_i)
    =
    \frac{N}{N_t}\mathbf{1}_{\{A_i=t\}}Y_{i,t+l}
    -
    \frac{N}{N_{>t+l}}\mathbf{1}_{\{A_i>t+l\}}Y_{i,t+l},
    \]
 using the cohort counts $N_t$ and $N_{>t+l}$ defined at the start \Cref{sect:sad}.

\item \textbf{RT (CA)} uses the covariate-adjusted mean-difference
score \(\psi_{i,t,l}^{\mathrm{CA}}\) in
\eqref{equ:ca_statistic}. The outcome models are fitted using all
units without using the realized assignments.
    \item \textbf{RT (SS)} uses the likelihood-ratio score
    \(\psi_{i,t,l}^{\mathrm{LR}}(O_i;\hat\eta)\) in \eqref{equ:lr_statistic}, where the CATE estimator is constructed using the realized assignments of 50\% of the units, following the method in Section~\ref{subsec:first_moment}. This method is based on Assumption~\ref{ass:additive_outcome_model} and is the closest comparison to our assignment-free CATE estimation method. The remaining 50\% of the units are used for the randomization test.
\item \textbf{RT (CATE)} is the proposed randomization test, implemented as in Algorithm~\ref{alg:rc}
with the LR score \(\psi_{i,t,l}^{\mathrm{LR}}\) in
\eqref{equ:lr_statistic}. We use the diagonal estimator in the
lag-invariant setting and the off-diagonal estimator in the
lagged setting.
\end{itemize}

\begin{figure}[t]
\vspace{-15pt}
    \centering
    \begin{minipage}{0.9\textwidth}
        \centering
        \begin{subfigure}[t]{0.38\linewidth}
            \centering
            \includegraphics[width=\linewidth]{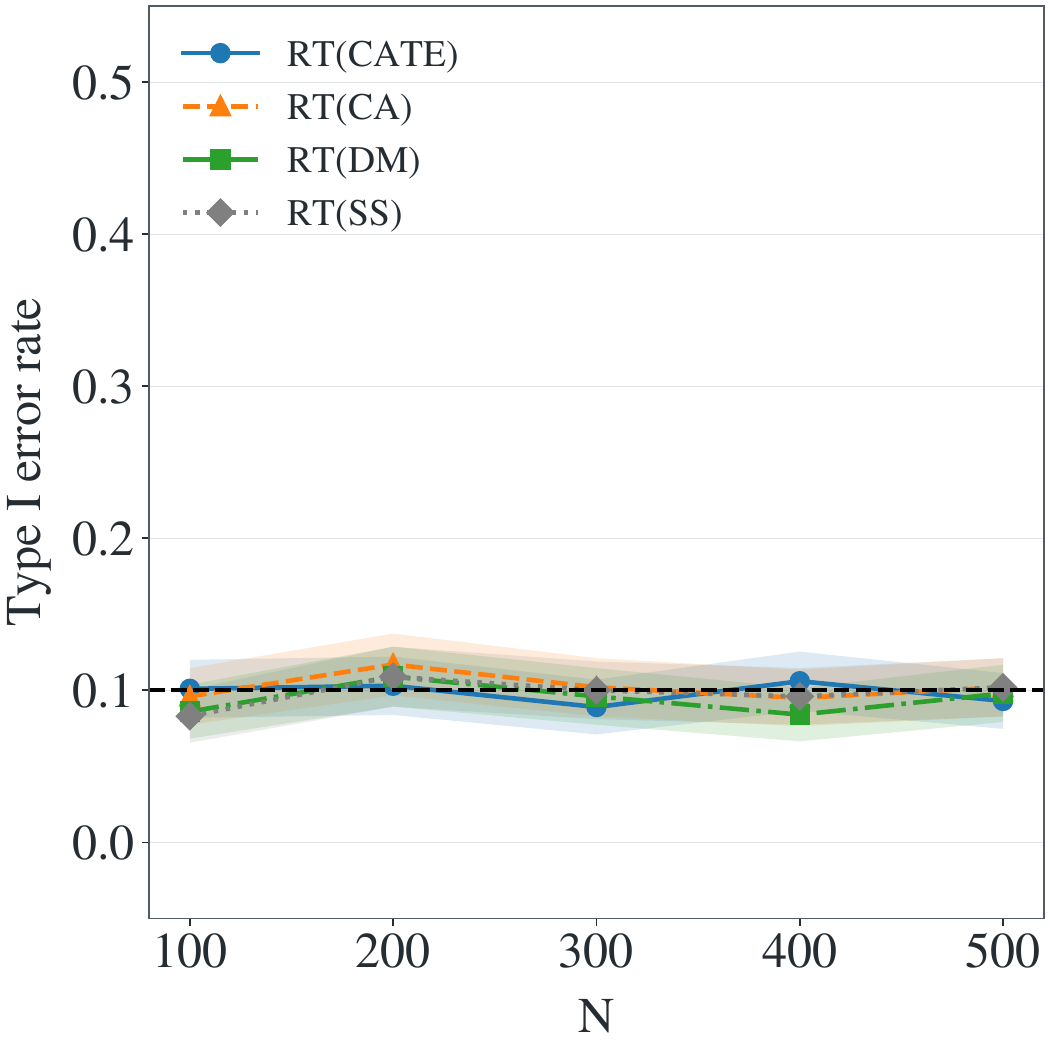}
            \caption{Size for the global null.}
            \label{fig1:sub1}
        \end{subfigure}
        \hspace{0.02\linewidth}
        \begin{subfigure}[t]{0.38\linewidth}
            \centering
            \includegraphics[width=\linewidth]{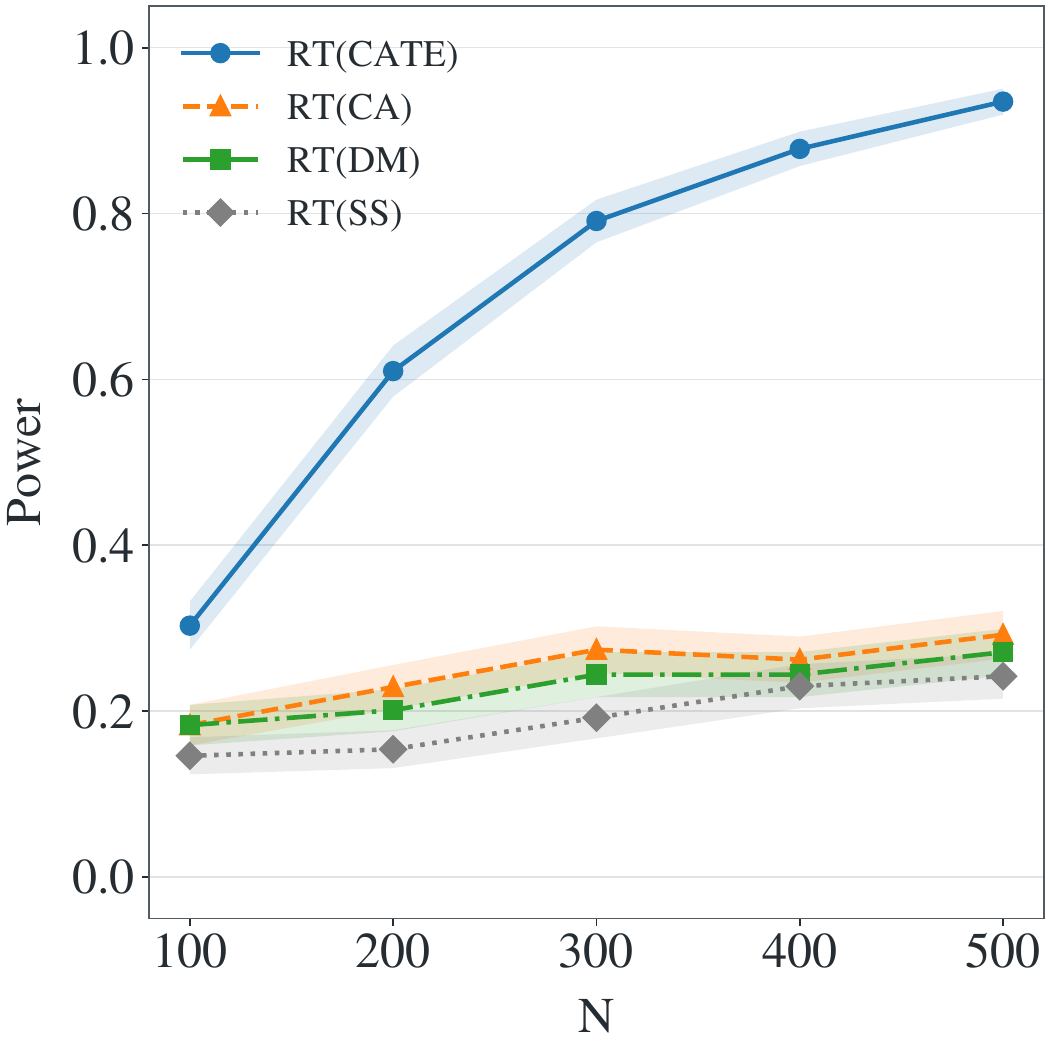}
            \caption{Power for the global null.}
            \label{fig1:sub2}
        \end{subfigure}
        \begin{subfigure}[t]{0.8\linewidth}
        \vspace{6pt}
            \centering
            \includegraphics[width=\linewidth]{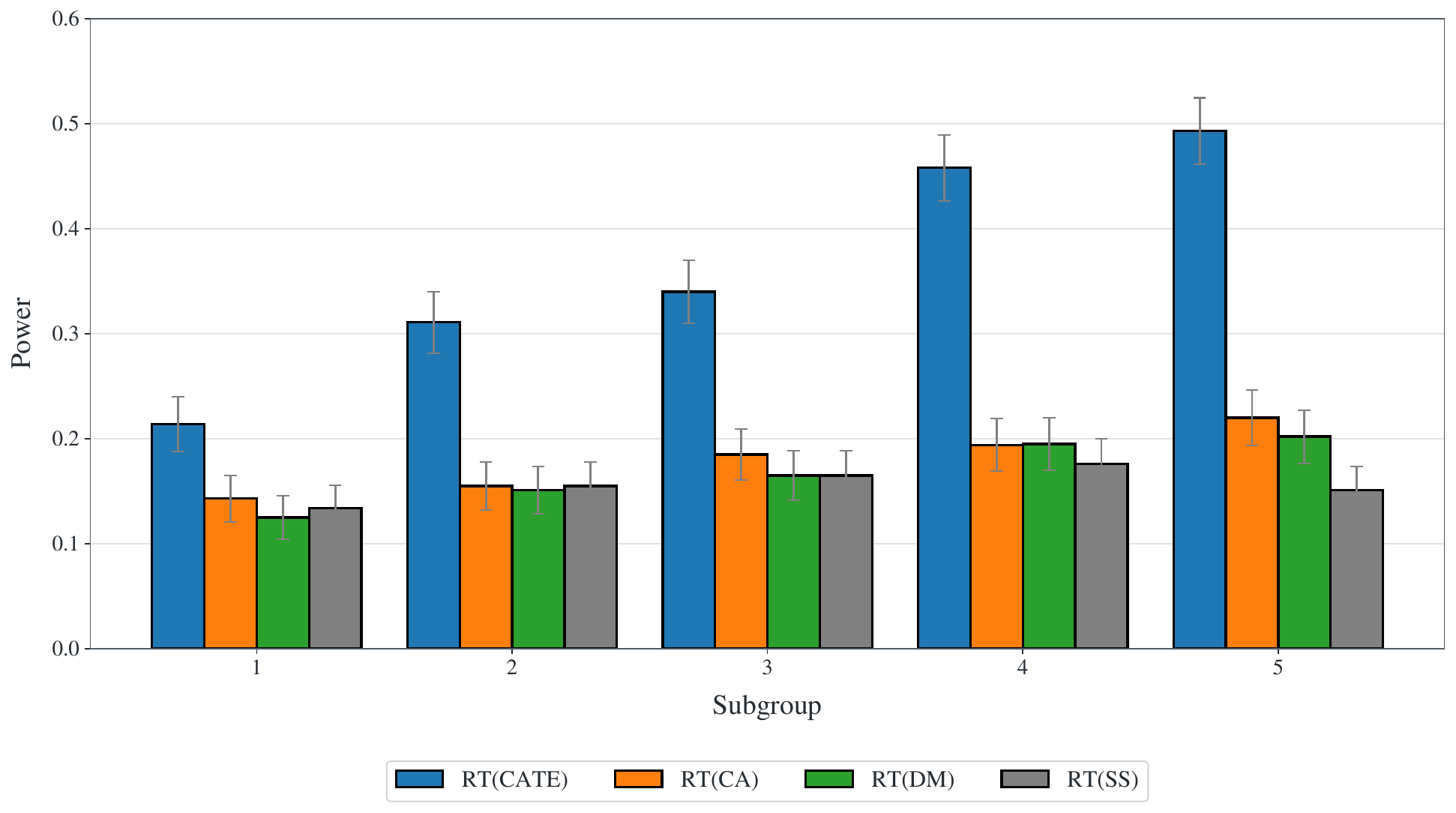}
            \caption{Power for the subgroup nulls at \(N=300\).}
            \label{fig1:sub3}
        \end{subfigure}
    \end{minipage}
    \caption{Comparison of randomization tests in the lag-invariant setting. }
    \label{fig:static}
\end{figure}

\Cref{fig1:sub1} confirms that in the lag-invariant setting, all four randomization tests control Type I error near the nominal level \(0.1\).  \Cref{fig1:sub2} shows that RT (CATE) achieves the highest power across all sample sizes. 
RT (CATE) uses unsigned CATE estimates with sign selection to
detect both average effects and treatment effect heterogeneity.
Its power advantage over RT (CA) is consistent with our power
analysis in \Cref{subsec:cate_gaussian_power}.

\begin{figure}[t]
    \centering
    \begin{minipage}{0.9\textwidth}
        \centering
        \begin{subfigure}[t]{0.38\linewidth}
            \centering
            \includegraphics[width=\linewidth]{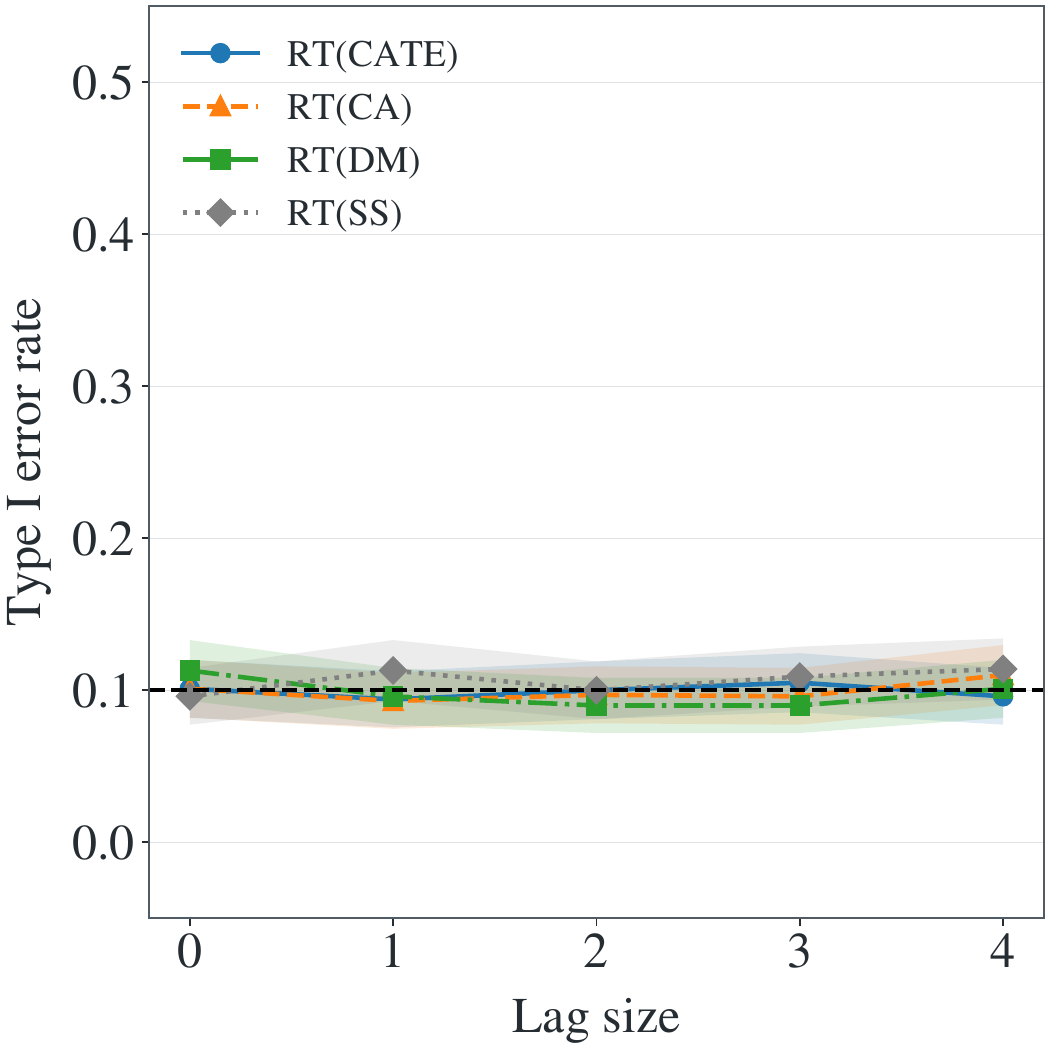}
            \caption{Size across lags. }
            \label{fig2:sub1}
        \end{subfigure}
        \hspace{0.02\linewidth}
        \begin{subfigure}[t]{0.38\linewidth}
            \centering
            \includegraphics[width=\linewidth]{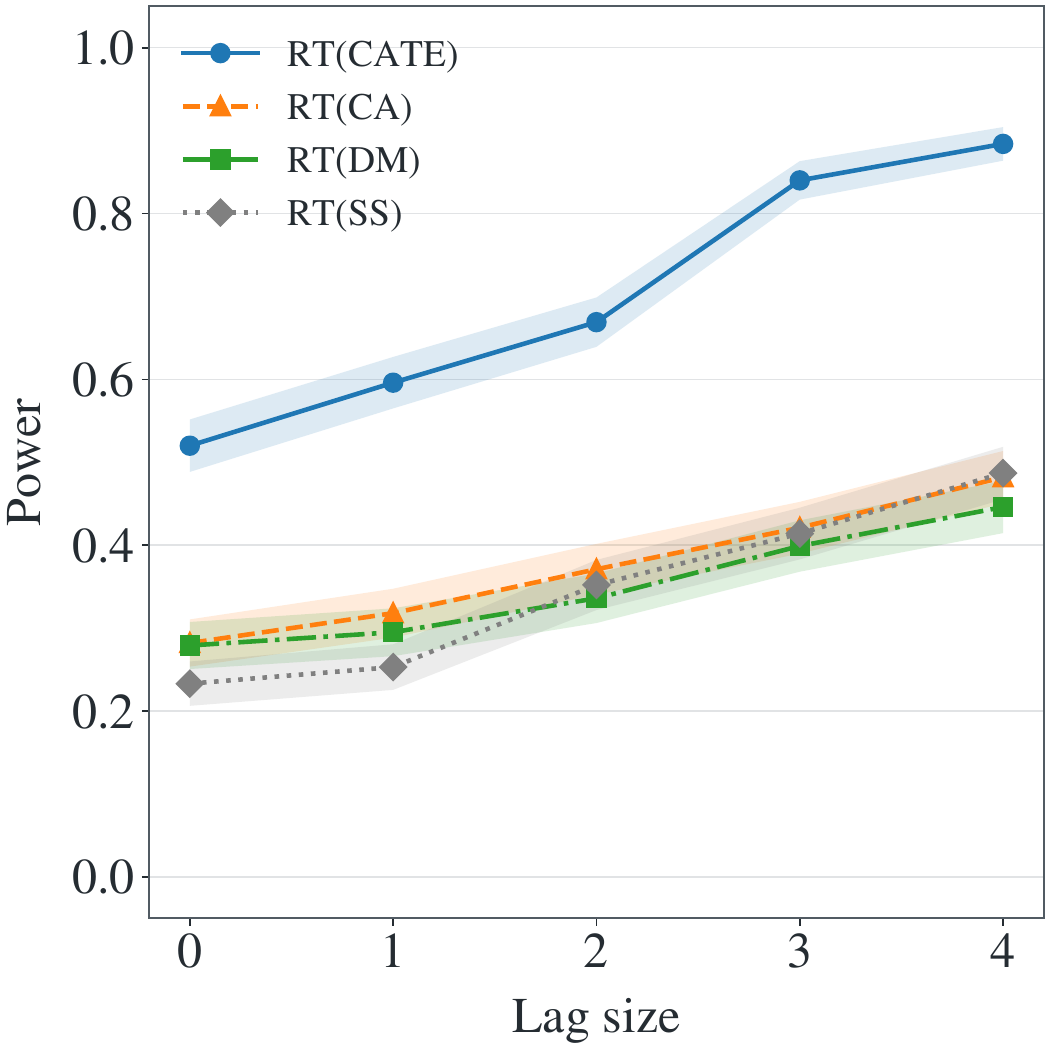}
            \caption{Power across lags.}
            \label{fig2:sub2}
        \end{subfigure}
        \begin{subfigure}[t]{0.38\linewidth}
                \vspace{0.5em}
            \centering
            \includegraphics[width=\linewidth]{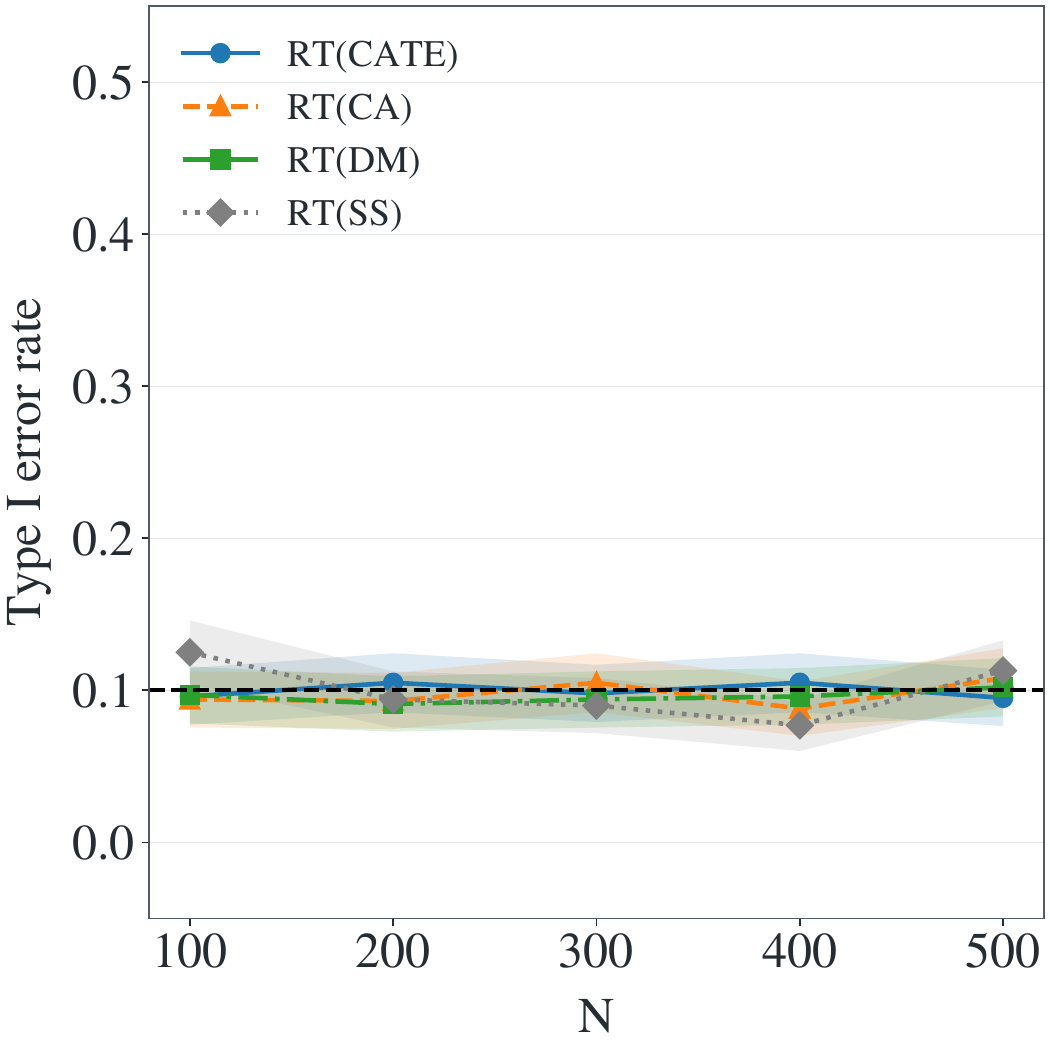}
            \caption{Size with Fisher's method.}
            \label{fig2:sub3}
        \end{subfigure}
        \hspace{0.02\linewidth}
        \begin{subfigure}[t]{0.38\linewidth}
         \vspace{0.5em}
            \centering
            \includegraphics[width=\linewidth]{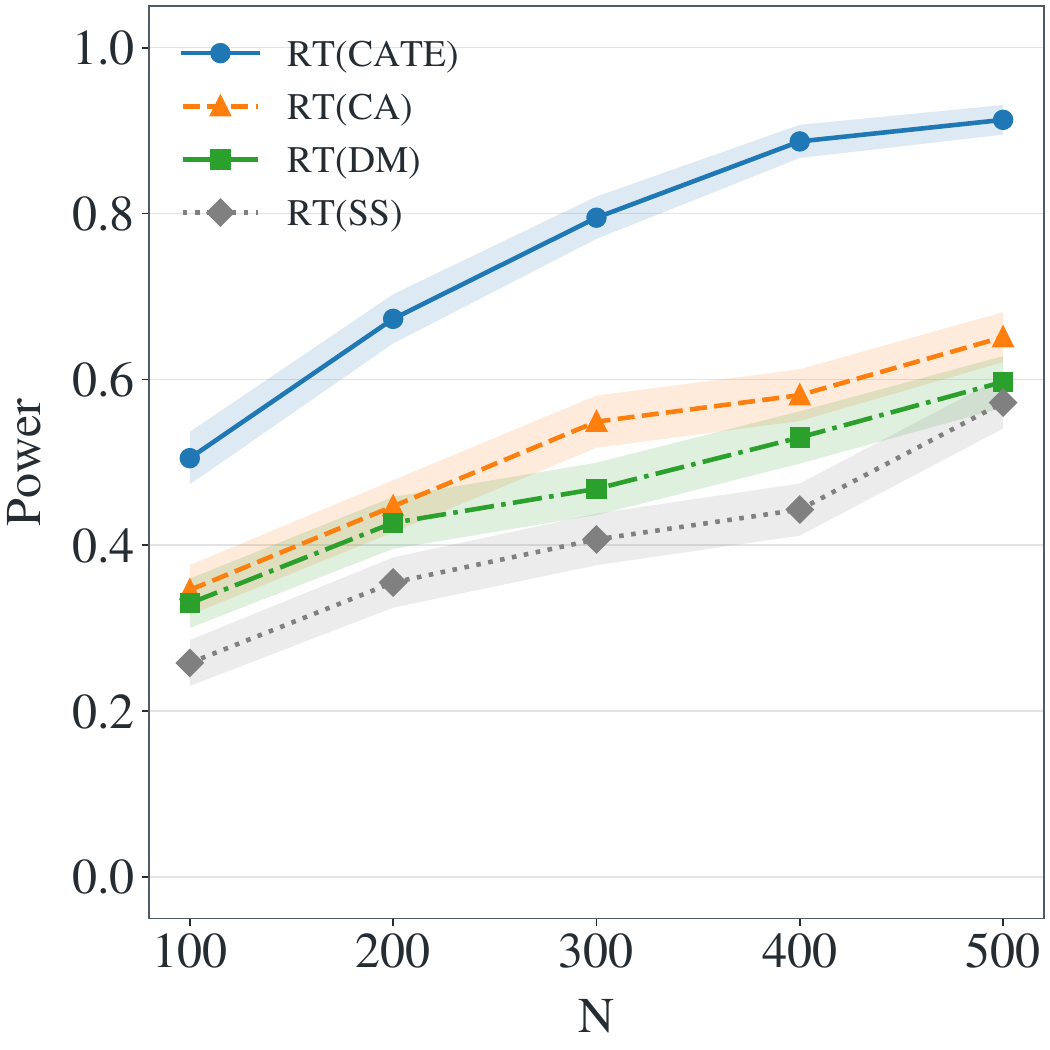}
            \caption{Power with Fisher's method.}
            \label{fig2:sub4}
        \end{subfigure}
    \end{minipage}
    \vspace{5pt}
    \caption{Comparison of randomization tests in the lagged setting. 
    The top row tests lagged effects at time $t=1$.
    The bottom row reports the results obtained by using  Fisher's method combining $p$-values across time \(t\) for \(l=0\).
    }
    \label{fig:lagged}
\end{figure}

In addition to the global null, we test for effects within
\(K=5\) subgroups. We partition the covariate space into regions
\(\mathcal X_1,\ldots,\mathcal X_K\) using the four quintile
cut-points of \(X_{i,1}\), and define
\(\mathcal S_k:=\{i\in[N]:X_i\in\mathcal X_k\}\), \(k\in[K]\).
We test the subgroup null
\begin{equation}\label{eq:H_0k}
H_{0,k}:\quad
Y_{i,t}(a)=Y_{i,t}(T+1)
\qquad
\text{for all }i\in\mathcal S_k,\ 1\le a\le t\le T.
\end{equation}
This null states that treatment has no effect on any unit in
subgroup \(k\).
For these baseline-covariate subgroups, a valid conditional
randomization test (CRT) fixes \(A_{[N]\setminus \mathcal S_k} \) and
randomizes the assignments within \(\mathcal S_k\).
Under \(H_{0,k}\) and no interference, the entire outcome panel
is then invariant, so nuisance functions fitted using that panel
can remain fixed throughout the CRT.
The randomization argument of \Cref{thm:validity} applies
conditional on the assignments outside the subgroup.
\Cref{fig1:sub3} shows that RT (CATE) also achieves the highest
power for the subgroup nulls in \eqref{eq:H_0k}. In the
data-generating process, treatment effects become stronger from
subgroup 1 to subgroup 5. Accordingly, the power of RT (CATE)
and RT (CA) increases across subgroups, while the power of
RT (DM) and RT (SS) remains relatively stable.

Figure~\ref{fig:lagged} reports the size and power results
for all tests in the lagged setting. Following
\citet{zhang2025multiple}, we first test the lag-\(l\) effect
at adoption time \(t=1\) for \(l=0,\ldots,4\), using their
conditional randomization test. We then test the global
lag-\(0\) null by combining the time-specific lag-\(0\)
\(p\)-values across all adoption times using Fisher's method.
They show that under \(H_{0,0}\), the time-specific tests
can be ordered by adoption time such that each \(p_{t,0}\)
is super-uniform conditional on the covariates, potential
outcomes, and preceding cohort allocations. Since the preceding \(p\)-values are fixed
given this information and previous randomization draws,
it follows by successive conditioning that
$-2\sum_{t=1}^T\log p_{t,0}$
is stochastically dominated by a \(\chi^2_{2T}\) random
variable. Thus Fisher's method controls the Type I error
for \(H_{0,0}\).

\Cref{fig2:sub1,fig2:sub3} show that rejection rates under the
null are close to the nominal level for both the individual
tests and the combined test.
\Cref{fig2:sub2,fig2:sub4} show that RT (CATE) achieves higher
power than the other methods.

\subsubsection{Consistency and warm-starting experiments}
\label{subsubsec:warm_start}

We next provide an empirical check of the identification
and consistency theory in \Cref{sect:identify_consistent}.
We ask two practical questions: whether the assignment-free
unsigned CATE estimator can converge to the true CATE vector
up to sign at each covariate value, and whether it can serve
as a warm start for learning the signed CATE vector.

We focus on the off-diagonal unsigned CATE estimator in \Cref{subsubsec:cross_est}. In Algorithm~\ref{alg:rc}, we use gradient-boosted regression trees to estimate the outcome functions and the residual covariance matrix. We first compute the spectral initialization in \eqref{eq:opt2_tau_sp}, and then refine it using the NLS objective. We set the trace penalty in the convex relaxation to decrease with the sample size, while always using no ridge penalty in the NLS refinement.
Other details of this experiment, including the data-generating process, are given in Appendix~\ref{app:warm_start_dgp}. We also repeat the consistency experiment with a finite-dimensional parametric version of the off-diagonal estimator in Appendix~\ref{app:parametric}.

Let \(\hat\tau_{\pm,N}\) denote the unsigned estimator at sample size \(N\). \Cref{fig3:sub1} shows that the normalized mean squared error (NMSE) of \(\hat\tau_{\pm,N}\) decreases toward zero as \(N\) increases. The NMSE is defined using the unsigned distance \(d_{\pm}(\cdot,\cdot)\) in \eqref{equ:unsigned_dist}:
\[
\mathrm{NMSE}(\hat \tau_{\pm,N},\tau)
=
\frac{
\mathbb E\!\left[d_{\pm}\{\hat\tau_{\pm,N}(X),\tau(X)\}^2\right]
}{
\mathbb E\!\left[\|\tau(X)\|_2^2\right]
}.
\]

\begin{figure}[t]
    \centering
    \begin{subfigure}[b]{0.4\textwidth}
        \centering
        \includegraphics[width=\linewidth]{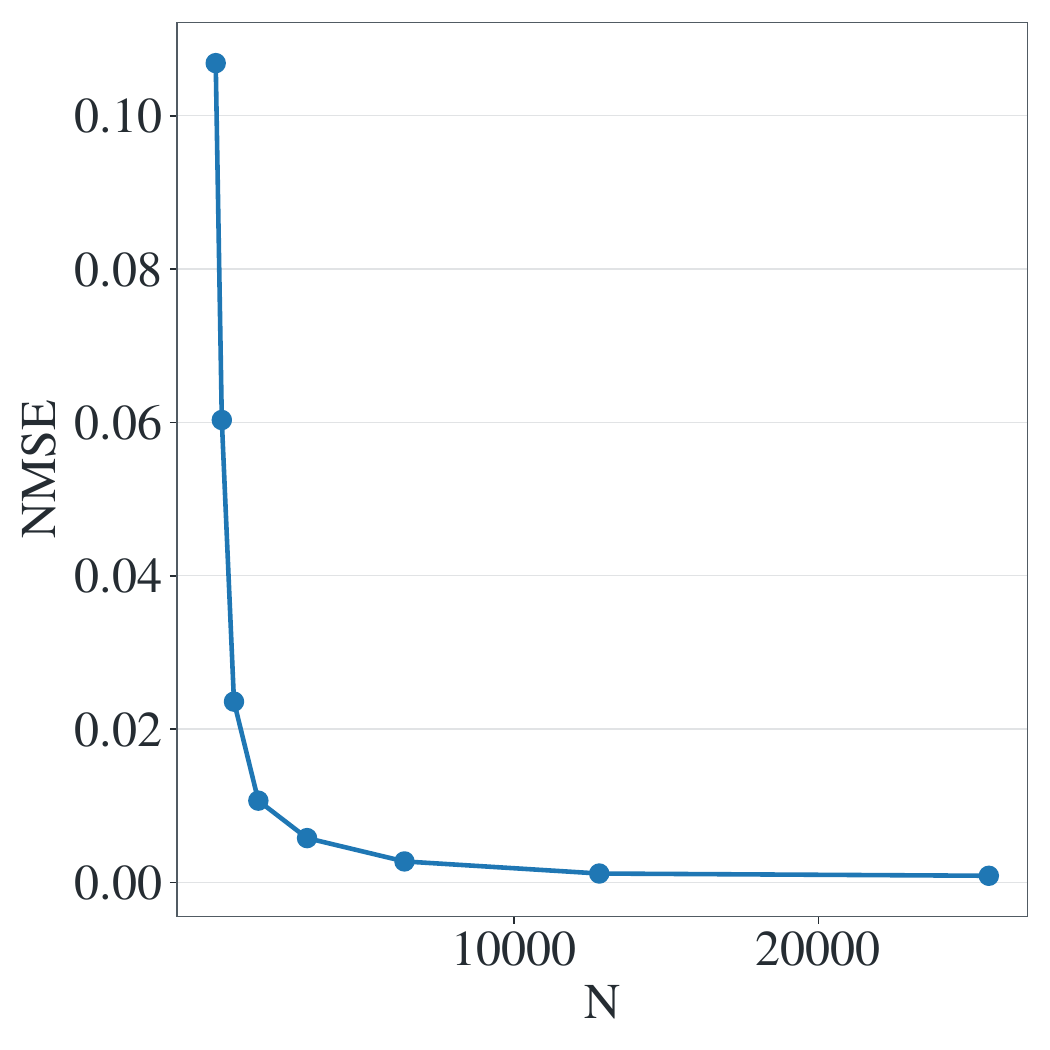}
        \caption{Unsigned estimator.}
        \label{fig3:sub1}
    \end{subfigure}
        \begin{subfigure}[b]{0.4\textwidth}
        \centering
        \includegraphics[width=\linewidth]{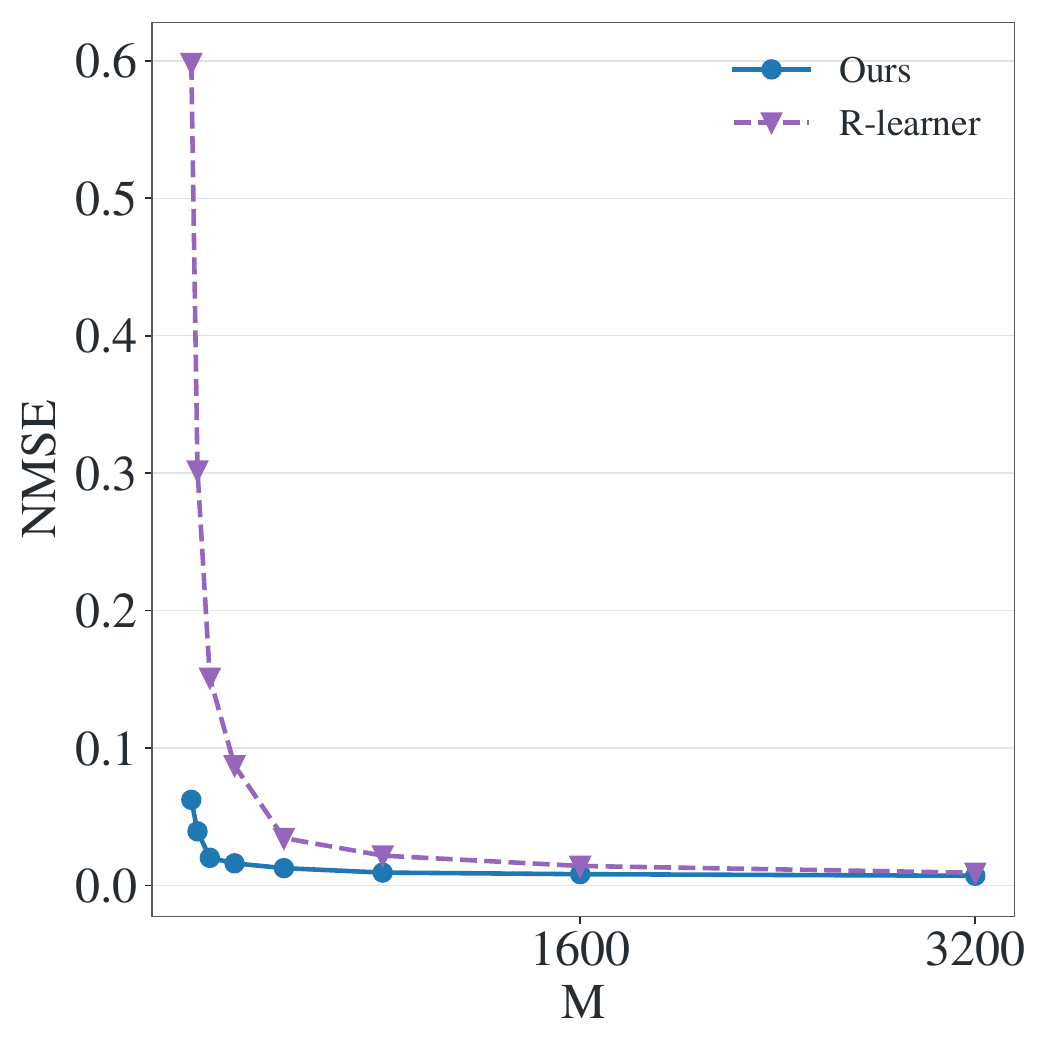}
        \caption{Signed estimators.}
        \label{fig3:sub2}
    \end{subfigure}
    \caption{Normalized MSE (NMSE) of the unsigned and signed CATE estimators.}
    \label{fig:consistency}
\end{figure}

For the warm-start experiment, we fix a full dataset of size \(N=25{,}600\). We first fit the off-diagonal estimator using all covariates and outcomes. We then reveal assignments for \(M\in\{25,50,100,200,400,800,1600,3200\}\) units.
We compare two estimators: a direct R-learner \citep{nie2020quasi} that learns the full CATE vector from these \(M\) assignment-revealed units, and a warm-start estimator that learns only the missing sign on top of our unsigned CATE estimator. Both estimators are implemented using the same gradient-boosted regression trees. The warm-start estimator uses the same second-stage loss as the direct R-learner, but restricts the fitted CATE to be a signed version of the unsigned estimator; see Appendix~\ref{app:warm_start_dgp} for details.

\Cref{fig3:sub2} reports the NMSE of the direct R-learner and our warm-start estimator, denoted by Ours, computed using the usual \(\ell_2\) distance. When \(M\) is small, the warm-start estimator is substantially more accurate. This shows that the unsigned CATE estimator can greatly reduce the number of revealed assignments needed to learn a signed CATE vector. As \(M\) grows, the direct learner improves and catches up.
The main takeaway is that learning the sign of the CATE can be much easier than learning the entire CATE vector from scratch.

\vspace{5pt}

\subsection{The Chicago procedural-justice training program}
\label{subsec:real_data}
\newcommand*{\ChicagoTableHead}[1]{%
\ensuremath{\vcenter{\hbox{\shortstack[c]{#1}}}}}

The Chicago Police Department introduced a one-day procedural-justice
course through staggered adoption. Between January 2012 and March
2016, 8,480 officers took the course in successive groups. The course
emphasized listening to civilians and explaining police decisions.
\citet{wood2020procedural} estimated that training reduced complaints
received, sustained or settled complaints, and reported uses of force
over the following two years. A later reanalysis of monthly
officer-level outcomes retained evidence of reduced use of force,
while the 95\% confidence intervals for the simple weighted average
effects on both complaint outcomes included zero
\citep{wood2021revised}.

\providecommand*{\ChicagoTableHead}[1]{%
  \ensuremath{\vcenter{\hbox{\shortstack[c]{#1}}}}}

\begin{table}[p]
\centering
\begingroup
\footnotesize
\setlength{\tabcolsep}{3.5pt}
\renewcommand{\arraystretch}{1.05}
\setlength{\aboverulesep}{0pt}
\setlength{\belowrulesep}{0pt}
\captionsetup{font=small,skip=3pt}

\caption{\(p\)-values for group tests in the Chicago data.}
\label{tab:chicago_groups}

\begin{tabular*}{\linewidth}{@{}c|c@{\extracolsep{\fill}}cccc@{}}
\toprule
\raisebox{\dimexpr2em-\height\relax}[2em][1.3em]{%
  \diagbox[trim=l,height=5.1em]{Method}{Group}}
& \ChicagoTableHead{Asian/Pacific\\Islander}
& \ChicagoTableHead{Black}
& \ChicagoTableHead{Hispanic}
& \ChicagoTableHead{Native American/\\Alaskan Native}
& \ChicagoTableHead{White} \\
\cline{2-6}
\rule[-0.5em]{0pt}{1.8em}
& \multicolumn{5}{c}{\textit{All complaints}} \\
\midrule
RT(SS)   & 0.330 & 0.527 & 0.559 & 0.483 & 0.826 \\
RT(CA)   & 0.344 & 0.854 & 0.403 & 0.288 & 0.546 \\
RT(DM)   & 0.548 & 0.880 & 0.399 & 0.330 & 0.699 \\
RT(CATE) & \textbf{0.121} & \textbf{0.369} & \textbf{0.284} & \textbf{0.270} & \textbf{0.371} \\
\midrule
\rule[-0.5em]{0pt}{1.8em}
& \multicolumn{5}{c}{\textit{Sustained or settled complaints}} \\
\midrule
RT(SS)   & 0.600 & 0.749 & 0.508 & 0.647 & 0.830 \\
RT(CA)   & 0.645 & 0.813 & 0.268 & 0.957 & 0.219 \\
RT(DM)   & 0.789 & 0.818 & 0.310 & 1.000 & 0.352 \\
RT(CATE) & \textbf{0.224} & \textbf{0.710} & \textbf{0.200} & \textbf{0.638} & \textbf{0.003} \\
\midrule
\rule[-0.5em]{0pt}{1.8em}
& \multicolumn{5}{c}{\textit{Use of force}} \\
\midrule
RT(SS)   & 0.588 & 0.890 & 0.377 & 0.435 & 0.195 \\
RT(CA)   & \textbf{0.261} & 0.483 & 0.709 & 0.725 & 0.793 \\
RT(DM)   & 0.591 & 0.637 & 0.656 & 0.765 & 0.880 \\
RT(CATE) & 0.945 & \textbf{0.300} & \textbf{0.194} & \textbf{0.394} & \textbf{0.015} \\
\bottomrule
\end{tabular*}

\bigskip
\bigskip

\caption{\(p\)-values for cohort tests in the Chicago data.}
\label{tab:chicago_cohorts}

\begin{tabular*}{\linewidth}{@{}c|c@{\extracolsep{\fill}}ccccc@{}}
\toprule
\raisebox{\dimexpr2em-\height\relax}[2em][1.3em]{%
  \diagbox[trim=l,height=5.1em]{Method}{Cohort}}
& \ChicagoTableHead{Q2\\2012}
& \ChicagoTableHead{Q3\\2012}
& \ChicagoTableHead{Q4\\2012}
& \ChicagoTableHead{Q1\\2013}
& \ChicagoTableHead{Q2\\2013}
& \ChicagoTableHead{Q3\\2013} \\
\cline{2-7}
\rule[-0.5em]{0pt}{1.8em}
& \multicolumn{6}{c}{\textit{All complaints}} \\
\midrule
RT(SS)   & 0.525 & 0.622 & 0.992 & 0.982 & 0.256 & 0.129 \\
RT(CA)   & 0.543 & 0.450 & 0.984 & 0.981 & 0.002 & 0.129 \\
RT(DM)   & 0.663 & 0.591 & 0.999 & 0.999 & 0.004 & 0.219 \\
RT(CATE) & \textbf{0.402} & \textbf{0.363} & \textbf{0.970} & \textbf{0.974} & \textbf{0.001} & \textbf{0.081} \\
\midrule
\rule[-0.5em]{0pt}{1.8em}
& \multicolumn{6}{c}{\textit{Sustained or settled complaints}} \\
\midrule
RT(SS)   & 0.967 & 0.224 & 0.574 & 0.728 & 0.761 & 0.416 \\
RT(CA)   & 0.675 & 0.185 & 0.905 & 0.651 & 0.361 & 0.349 \\
RT(DM)   & 0.753 & 0.292 & 0.966 & 0.770 & 0.374 & 0.590 \\
RT(CATE) & \textbf{0.046} & \textbf{0.133} & \textbf{0.362} & \textbf{0.569} & \textbf{0.238} & \textbf{0.278} \\
\midrule
\rule[-0.5em]{0pt}{1.8em}
& \multicolumn{6}{c}{\textit{Use of force}} \\
\midrule
RT(SS)   & 0.661 & 0.702 & 0.924 & 0.888 & 0.157 & 0.994 \\
RT(CA)   & 0.998 & 0.837 & 0.935 & 0.513 & 0.003 & 0.015 \\
RT(DM)   & 0.999 & 0.900 & 0.960 & 0.703 & \textbf{0.001} & 0.007 \\
RT(CATE) & \textbf{0.530} & \textbf{0.029} & \textbf{0.765} & \textbf{0.479} & 0.002 & \textbf{0.004} \\
\bottomrule
\end{tabular*}

\endgroup
\end{table}

We use the officer-level data distributed with Roth and
Sant'Anna's \texttt{staggered} package in R%
\footnote{\url{https://cran.r-project.org/package=staggered}}
and the incident records accompanying
\citet{wood2021revised}.%
\footnote{\url{https://github.com/george-wood/procedural_justice_revisited}}
We restrict the analysis to 3,504 officers appointed in
2000 or later and aggregate monthly outcomes by calendar
quarter from January 2011 through December 2016.
We analyze the rollout under the working assumption that
training quarters were assigned as if at random, uniformly
over allocations with the observed cohort sizes.
For each outcome and each of five racial and ethnic groups,
we test the sharp null \(H_{0,k}\) in \eqref{eq:H_0k} that
training had zero effect on every officer in the group.
\citet[Figure~4]{wood2020procedural} define early adopters
as officers trained during the first 17-month phase of
the rollout. Motivated by this distinction, we examine
six early quarterly cohorts, from April--June (Q2) 2012
through July--September (Q3) 2013.
For each training quarter \(t\), we test \(H_{0,t,1}\)
in \eqref{eq:H_0tl} by comparing outcomes in quarter
\(t+1\) for officers trained in \(t\) and officers
still untrained at the end of \(t+1\).

We compare the four tests in \Cref{subsubsec:size_and_power}.
The covariates used for nuisance fitting are birth year,
appointment year, race, gender, the number of use-of-force
reports in October--December 2010, the number of complaints
filed in 2010, and the number of complaints closed with a
sustained finding in 2010.
Within each comparison, RT(CATE), RT(CA), and RT(SS)
share the same fitted outcome means.
We use regression trees to estimate outcome means and
residual moments.
We select the models based on prediction performance,
using only the data permitted for nuisance fitting under
each null.
For RT(SS), we randomly split the eligible officers into
two halves for each test.
Each \(p\)-value is computed using 1000 randomization draws.

The two tables show that RT(CATE) gives the smallest \(p\)-value
in most comparisons. At the nominal level \(\alpha=0.10\), only
RT(CATE) rejects two nulls among the group tests in
\Cref{tab:chicago_groups}: the sharp nulls of no effect on
sustained or settled complaints (\(p=0.003\)) and use of force
(\(p=0.015\)) for White officers.
Among the cohort tests in \Cref{tab:chicago_cohorts}, only RT(CATE)
rejects the one-quarter effect null for all complaints in the
2013 Q3 cohort (\(p=0.081\)), for sustained or settled complaints
in the 2012 Q2 cohort (\(p=0.046\)), and for use of force in the
2012 Q3 cohort (\(p=0.029\)).
RT(CATE), RT(CA), and RT(DM)
also reject the null for all complaints in the 2013 Q2 cohort
and for use of force in the 2013 Q2 and Q3
cohorts; RT(SS) rejects
none of these three nulls.

These results illustrate the practical benefits of incorporating
CATE estimates without sample splitting. Compared with RT(SS),
RT(CATE) retains more observations for testing, which may be
particularly useful when subgroup samples are limited. Compared
with RT(CA), it uses treatment-effect information beyond outcome
adjustment.

\section{Discussion}\label{sect:discussion}

We developed a general method for constructing CATE-assisted
randomization tests without sample splitting. The main idea is
simple: in panel experiments, residualized outcomes contain more
information than their conditional means alone. Under suitable
assumptions, their covariance matrix carries information about
treatment effects, which can be extracted without fitting on the
realized assignments. The resulting tests have finite-sample
validity under the stated conditions. The experiments, including
those in Appendix~\ref{app:designs}, show higher power than the
baselines in the settings studied, covering staggered-adoption,
parallel-group, two-period crossover, and repeated crossover
designs with different treatment paths and CATE functions.

The assignment-free estimates also support exploratory subgroup
analysis. They can guide the search for groups with large estimated
effect magnitudes, followed by tests of the corresponding subgroup
sharp nulls. This connects our method to work on subgroup discovery
and inference \citep{athey2016recursive,hsu2015strong}.
When the search uses only covariates and assignment-free estimates,
the selected groups remain unchanged under each subgroup sharp
null when assignments outside that group are held fixed.
The same conditional randomization argument therefore allows the
observations to be used for both subgroup discovery and testing.

Residual covariances may also support assignment-free estimation of
unsigned treatment effects beyond the panel setting studied here.
For example, similar structure may arise when several distinct
outcomes are observed for each unit. More broadly, our results
suggest a class of randomization tests based on assignment-free
treatment-effect estimation. The basic ingredients are a known
treatment assignment distribution, a suitable outcome model, and
enough second-order moment restrictions to identify the effects
up to sign. The remaining sign ambiguity can then be addressed efficiently when evaluating the statistic under each candidate assignment.

\bibliographystyle{plainnat}
\bibliography{references}

\appendix

\section{Results in other panel designs}
\label{app:designs}
\label{sect:designs}

We describe four randomized panel designs and provide additional
simulation results for the parallel-group, two-period crossover,
and repeated crossover designs. These experiments examine the
performance of our method beyond staggered adoption.

\subsection{Assignment designs}

\begin{example}[Staggered-adoption design]\label{ex:staggered}
In the staggered-adoption design, the treatment assignment space
is \(\mathcal A=\{1,\ldots,T+1\}\).
Fix group sizes \(N_1,\ldots,N_{T+1}\) such that
\(\sum_{t=1}^{T+1}N_t=N\).
Here \(N_t\) denotes the number of units assigned to begin treatment
at time \(t\in[T]\), and \(N_{T+1}\) denotes the number of
never-treated units. The assignment vector \(A_{[N]}\) is drawn
uniformly from all assignments with these fixed group sizes:
\[
\mathbb P(A_{[N]}=a_{[N]})
=
\frac{\prod_{t=1}^{T+1}N_t!}{N!}
\mathbbm 1\!\left\{
a_{[N]}\in[T+1]^N,\;
\sum_{i=1}^N \mathbf 1_{\{a_i=t\}}=N_t
\ \text{for all }t\in[T+1]
\right\}.
\]
Under this design, exactly \(N_t\) units are assigned \(A_i=t\)
for each \(t\in[T+1]\).
\end{example}

\begin{example}[Parallel-group design]\label{ex:crt}
Fix a treatment time \(t^\star\in[T]\).
Exactly \(N_{\mathrm t}\) units are randomly assigned to begin
treatment at time \(t^\star\), while the remaining
\(N_{\mathrm c}=N-N_{\mathrm t}\) units are never treated.
Thus \(A_i\in\mathcal A=\{t^\star,T+1\}\) for all \(i\),
with exactly \(N_{\mathrm t}\) units assigned \(A_i=t^\star\).
The assignment vector \(A_{[N]}\) follows
\[
\mathbb P(A_{[N]}=a_{[N]})
=
\binom{N}{N_{\mathrm t}}^{-1}
\mathbf 1\!\left\{
a_{[N]}\in\{t^\star,T+1\}^N,\;
\sum_{i=1}^N \mathbf 1_{\{a_i=t^\star\}}=N_{\mathrm t}
\right\}.
\]
This is complete randomization between two parallel groups,
with a common treatment start time.
Before \(t^\star\), all units are untreated; from \(t^\star\)
onward, the treatment group remains treated and the control
group remains untreated.
\end{example}

\begin{example}[Two-period crossover design]\label{ex:crossover}
Let \(\mathsf A\) denote treatment and \(\mathsf B\) no treatment.
Fix a switching boundary \(t_0\in\{1,\ldots,T-1\}\).
Period~1 contains times \(t\le t_0\), and period~2 contains
times \(t>t_0\).
In the usual two-sequence \(\mathsf{AB}/\mathsf{BA}\)
crossover design \citep{dwan2019crossover}, each unit
receives treatment in one period and no treatment in the
other. Its switching times
\(S_i=(A_i^{\mathrm{on}},A_i^{\mathrm{off}})\) are
\[
S_i=
\begin{cases}
(1,t_0+1), & \text{for sequence }\mathsf{AB},\\
(t_0+1,T+1), & \text{for sequence }\mathsf{BA}.
\end{cases}
\]
Let \(\mathcal S_{\mathrm{CO}}\) be the set of assignments
with exactly \(N_{\mathsf{AB}}\) units receiving
\(\mathsf{AB}\) and the remaining units receiving
\(\mathsf{BA}\). We draw uniformly from this set:
\[
\mathbb P(S_{[N]}=s_{[N]})
=
\binom{N}{N_{\mathsf{AB}}}^{-1}
\mathbf 1\!\left\{s_{[N]}\in\mathcal S_{\mathrm{CO}}\right\}.
\]
The treatment indicator is
\(D_{i,t}=\mathbf 1_{\{A_i^{\mathrm{on}}\le t<A_i^{\mathrm{off}}\}}\).
A washout interval, if used, is specified separately
between the two treatment periods.
\end{example}

\begin{example}[Repeated crossover design]\label{ex:repeated_crossover}
Let \(T=2C\), and group the treatment periods into cycles
\((1,2),\ldots,(2C-1,2C)\).
Each cycle uses randomized \(\mathsf{AB}/\mathsf{BA}\)
pairs \citep{vohra2015consort}. We additionally fix the
allocation counts: in cycle \(c\), exactly \(n_c\) units
receive \(\mathsf{AB}\) and the remaining units receive
\(\mathsf{BA}\), where \(n_c\in\{0,\ldots,N\}\).

Write \(B_{ic}=1\) for \(\mathsf{AB}\) and \(B_{ic}=0\)
for \(\mathsf{BA}\). The treatment indicators are
\(D_{i,2c-1}=B_{ic}\) and \(D_{i,2c}=1-B_{ic}\).
Independent uniform allocation across cycles gives
\[
\mathbb P(B=b)
=
\prod_{c=1}^{C}
\binom{N}{n_c}^{-1}
\mathbf 1\!\left\{\sum_{i=1}^{N}b_{ic}=n_c\right\},
\qquad b\in\{0,1\}^{N\times C}.
\]
Balanced allocation sets \(n_c=N/2\) in every cycle,
for even \(N\). Under balanced allocation with three
cycles, each unit can receive any of the eight sequences
formed by concatenating \(\mathsf{AB}\) and \(\mathsf{BA}\).
The counts of these complete sequences are random;
only the allocation counts within each cycle are fixed.
\end{example}

\subsection{Additional simulation results}
\label{subsec:simulations_additional}

We repeat the four-method comparison in
Section~\ref{subsubsec:size_and_power} under parallel-group,
two-period crossover, and repeated crossover designs.
All experiments test the global sharp null \(H_0\) of no
treatment effect at any occasion under any admissible
treatment sequence.
We consider \(T=6\) measurement occasions and sample sizes
\(N\in\{100,200,300,400,500\}\).

For each unit \(i\in[N]\), we observe covariates
\(X_i=(X_{i,1},\ldots,X_{i,5})\).
The five covariates are mutually independent:
\(X_{i,1},X_{i,4},X_{i,5}\sim\mathrm{Unif}[-1,1]\),
while \(X_{i,2}\) and \(X_{i,3}\) each take the values
\(-1\) and \(1\) with equal probability.
The covariate vectors are independent across units.
For \(t\in[6]\), the untreated conditional mean is
\[
\mu_{0,t}(X_i)
=
(0.7+0.1t)
\left(
0.6X_{i,1}-0.3X_{i,2}+0.2X_{i,3}
+0.4X_{i,4}-0.2X_{i,5}
\right).
\]
Let \(D_{i,t}\) indicate whether unit \(i\) receives treatment
at occasion \(t\). Outcomes follow
\[
Y_{i,t}=\mu_{0,t}(X_i)+D_{i,t}\tau(X_i)+\epsilon_{i,t},
\qquad
\epsilon_i\sim\mathcal N_6(0,\Sigma_\epsilon),
\]
where $\Sigma_\epsilon
=
\bigl(0.25^2+0.6^2\,0.4^{|t-s|}\bigr)_{t,s=1}^{6}.$
The errors are independent across units and independent of
the covariates and treatment assignments.
This covariance combines a Gaussian random intercept with
serially correlated errors.
The treatment effect is constant across treated occasions,
and there is no carryover.
For the power experiments, we use the CATE functions specified below.
For the size experiments, we set \(\tau(X_i)=0\) for all units \(i\),
while keeping the other parts of the data distribution unchanged.
To write the alternative DGPs, first define
\[
H_c(x)=\frac{(0.15+\mathbf 1_{\{x_1>c\}})x_2x_3}
{\sqrt{0.0225+1.3(1-c)/2}}.
\]
For the parallel-group design in \Cref{ex:crt}, we set
\(t^\star=3\) and \(N_{\mathrm t}=0.4N\). The treated units follow
the path \((0,0,1,1,1,1)\). The remaining \(0.6N\) units
are never treated.
Set
\[
\tau(x)=\frac{0.325}{\sqrt{1+0.20^2}}\{0.20+H_{0.3}(x)\}.
\]
For the two-period crossover design in \Cref{ex:crossover}, we set $t_0=4$, so the first treatment period contains four measurement occasions and the second contains two. We randomly assign $N_{\mathsf{AB}}=\max\{18,\operatorname{round}(0.08N)\}$ units to $\mathsf{AB}$, with path $(1,1,1,1,0,0)$, and the remaining units to $\mathsf{BA}$, with path $(0,0,0,0,1,1)$. The five sample sizes thus have $18,18,24,32,40$ units assigned to $\mathsf{AB}$, respectively. We set
\[
\tau(x)=0.045+0.39\{0.40x_2+\sqrt{1-0.40^2}\,H_{0.6}(x)\}.
\]
For the repeated crossover design in \Cref{ex:repeated_crossover},
we set $C=3$ and $n_c=0.10N$ for $c=1,2,3$.
The three cycles contain times $(1,2)$, $(3,4)$, and
$(5,6)$, respectively.
In each cycle, exactly $0.10N$
units receive $\mathsf{AB}$ and the remaining units receive
$\mathsf{BA}$, with assignments drawn independently across cycles.
We set
\[
\tau(x)=\frac{0.196(0.08)}{\sqrt{1+0.08^2}}
        +0.1045x_2+0.222H_{0.6}(x).
\]
RT (CA), RT (SS), and RT (CATE) use linear ridge regression
to estimate outcome means. RT (CATE) fits conditional
second moments by gradient-boosted regression trees on
the full sample. RT (SS) estimates signed effects using
gradient-boosted regression trees on a randomly selected
half of the sample and uses the remaining half for testing.
All regressions use the five original covariates.
The unsigned CATE estimator uses residual moments adapted to
the design. 
Under the outcome model above, the residual
moments in \eqref{equ:second_R} satisfy
$c_{t,s}(x)
=
\operatorname{Cov}(D_{i,t},D_{i,s})\tau(x)^2
+(\Sigma_\epsilon)_{t,s}.$
For the parallel-group design, RT (CATE) uses diagonal
residual-variance contrasts between occasions before and
after treatment begins. The common error variance cancels
in these contrasts.
For both crossover designs, it uses off-diagonal contrasts
of adjacent residual covariances.
In the two-period crossover design, these compare adjacent
occasions within a treatment period with the pair across
the switching boundary.
In the repeated crossover design, they compare adjacent
occasions across cycle boundaries with pairs within cycles.
The common adjacent error covariance cancels, while the
known differences in treatment covariances identify
\(\tau(x)^2\).

Under the global sharp null, outcomes remain fixed across permitted
assignments. Following Algorithm~\ref{alg:rc}, RT (CATE) fits the outcome
models and unsigned CATEs once on the full panel without using the
realized assignments, then holds these fits fixed across randomization draws.
Signs, fitted control means, and the statistic are evaluated
under each candidate treatment sequence.
The draws follow the distribution of the relevant design:
they preserve the group counts in the parallel-group design
and the sequence counts in the two-period crossover design,
and independently preserve the allocation counts within
each cycle in the repeated crossover design.
Each test uses 1000 randomization draws and nominal level \(0.10\).
The residual-moment formulas and sign-selection rule used in these
experiments are given in Appendix~\ref{app:other_design_details}.

Figure~\ref{fig:additional_designs} shows that all four
methods have empirical Type I error rates close to the
nominal level across all three designs.
Power increases with sample size, and RT (CATE) has the
highest power at every sample size considered.
These experiments cover different treatment patterns.
The parallel-group design compares a group that begins
treatment at a common time with a never-treated group.
In the two-period crossover design, all units receive
both treatment and no treatment, in opposite orders.
The repeated crossover design randomizes these orders
again in every cycle, allowing several treatment changes
and leaving the counts of complete treatment sequences
random.
The results show that the proposed approach can be useful
beyond staggered adoption, including settings where treatment
is withdrawn and reassigned.

\begin{figure}[t]
    \centering
    \begin{minipage}{0.9\textwidth}
        \centering
        \begin{subfigure}[t]{0.38\linewidth}
            \centering
            \includegraphics[width=\linewidth]{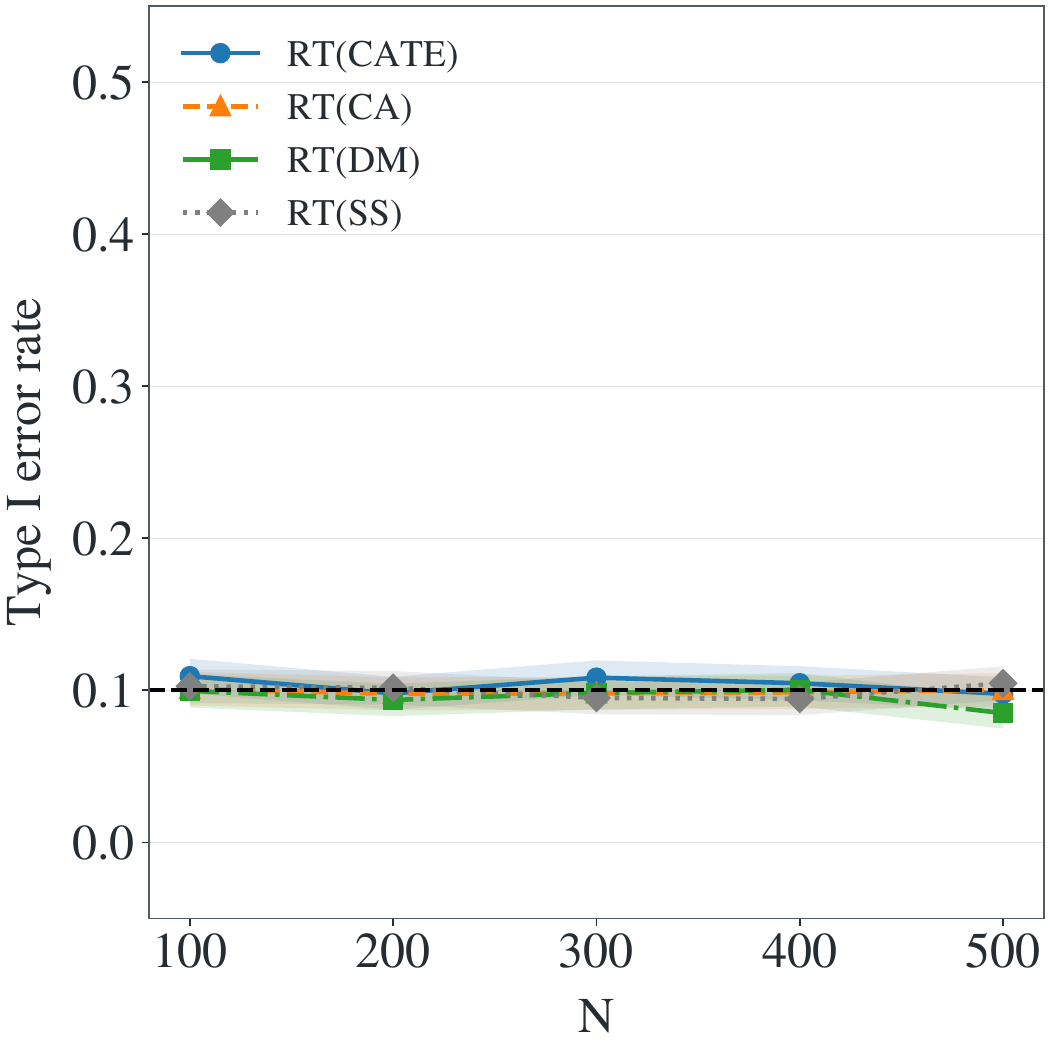}
            \caption{Parallel-group: size.}
            \label{fig5:sub1}
        \end{subfigure}
        \hspace{0.02\linewidth}
        \begin{subfigure}[t]{0.38\linewidth}
            \centering
            \includegraphics[width=\linewidth]{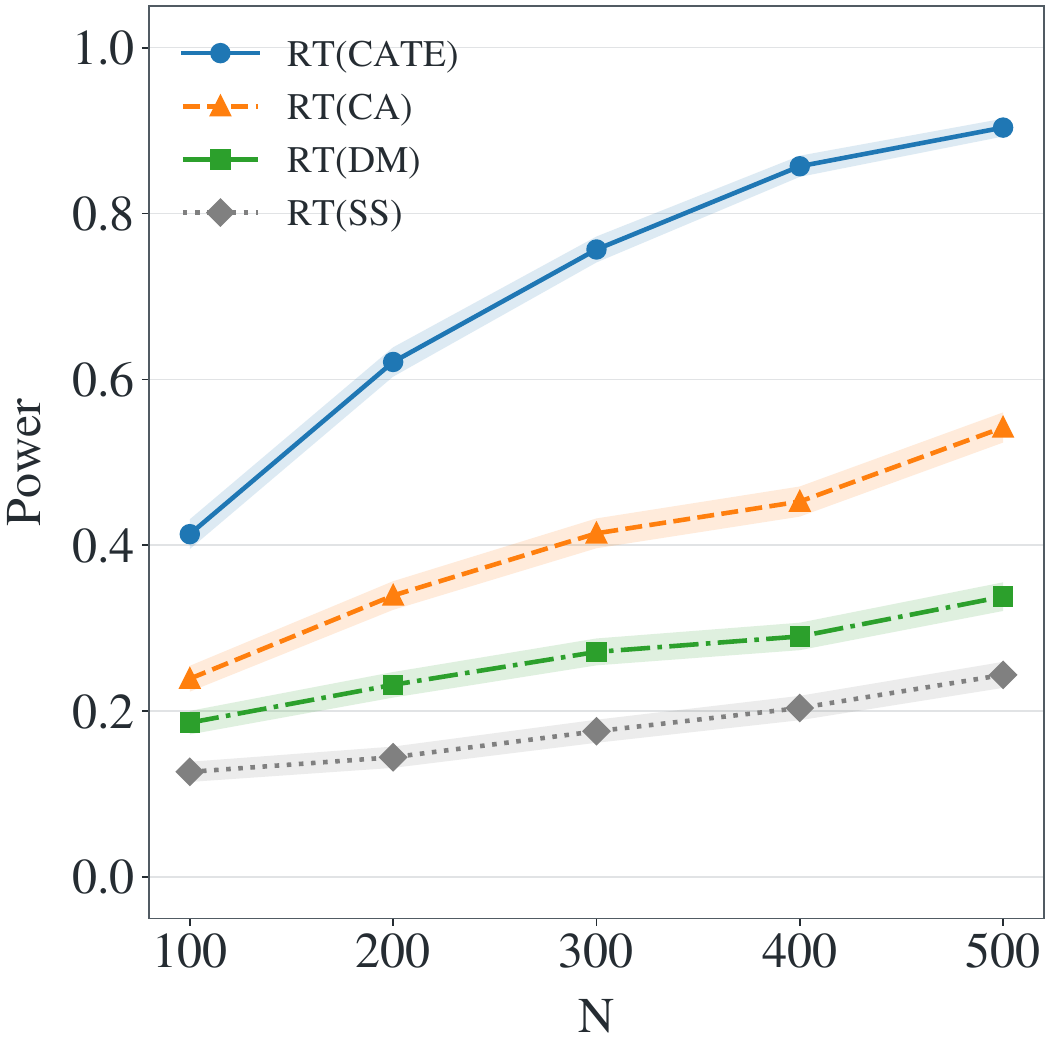}
            \caption{Parallel-group: power.}
            \label{fig5:sub2}
        \end{subfigure}
        \par\medskip
        \begin{subfigure}[t]{0.38\linewidth}
            \centering
            \includegraphics[width=\linewidth]{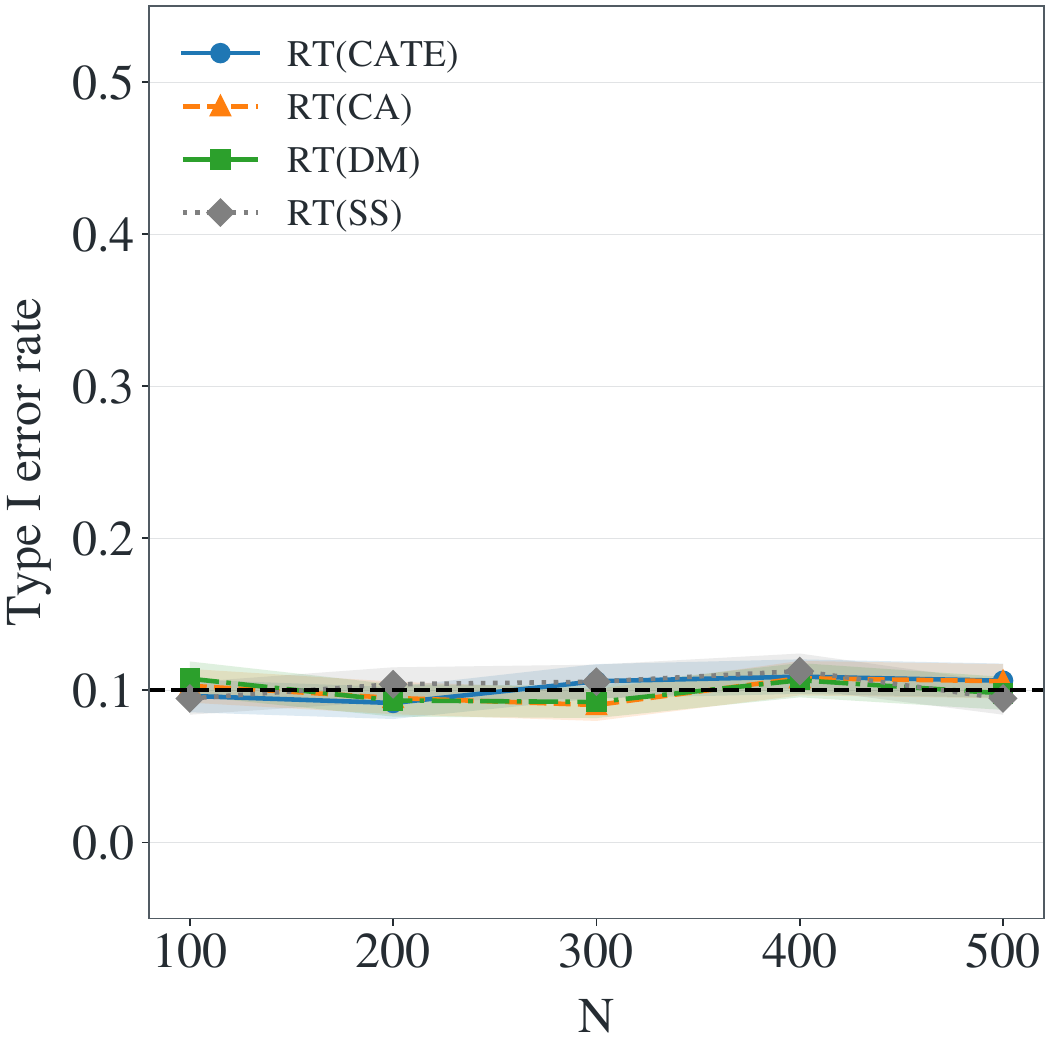}
            \caption{AB/BA crossover: size.}
            \label{fig5:sub3}
        \end{subfigure}
        \hspace{0.02\linewidth}
        \begin{subfigure}[t]{0.38\linewidth}
            \centering
            \includegraphics[width=\linewidth]{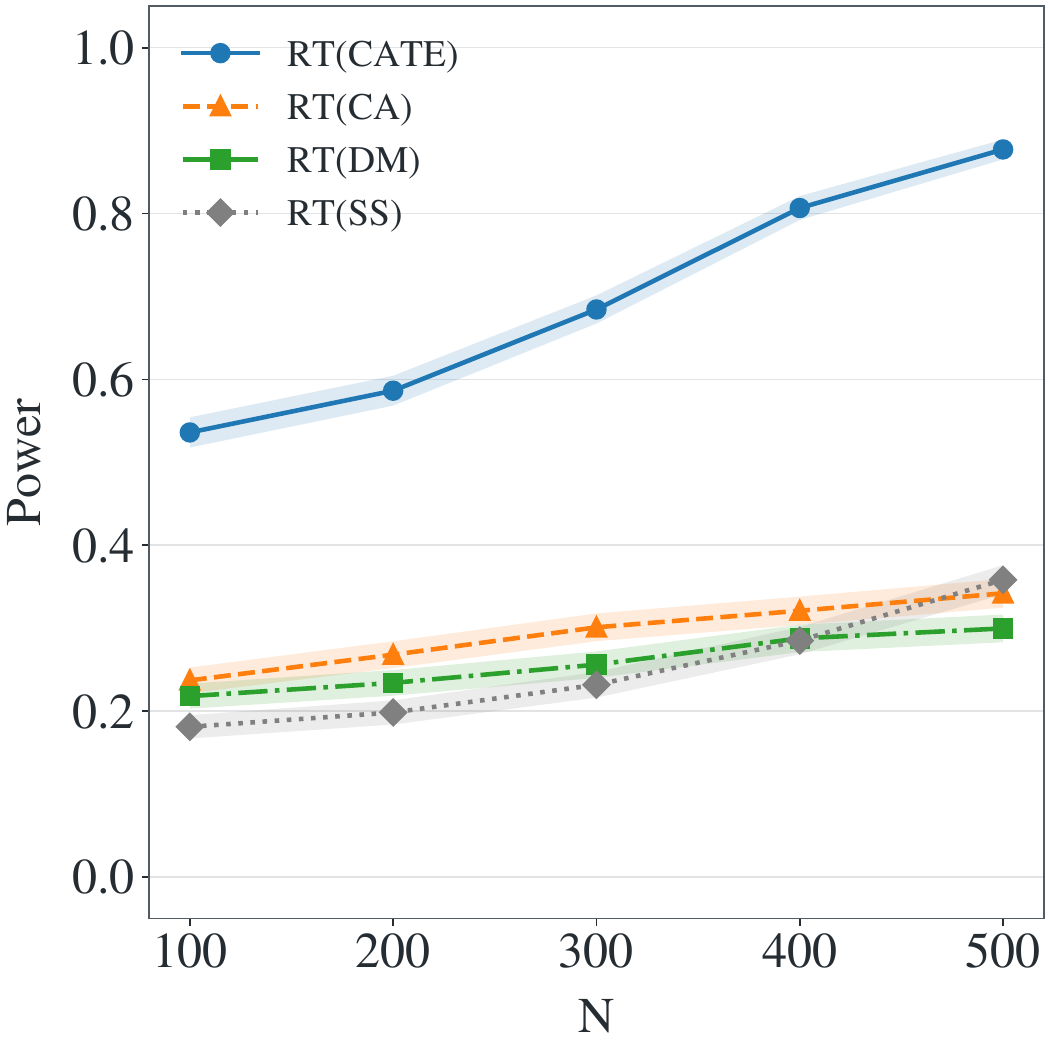}
            \caption{AB/BA crossover: power.}
            \label{fig5:sub4}
        \end{subfigure}
        \par\medskip
        \begin{subfigure}[t]{0.38\linewidth}
            \centering
            \includegraphics[width=\linewidth]{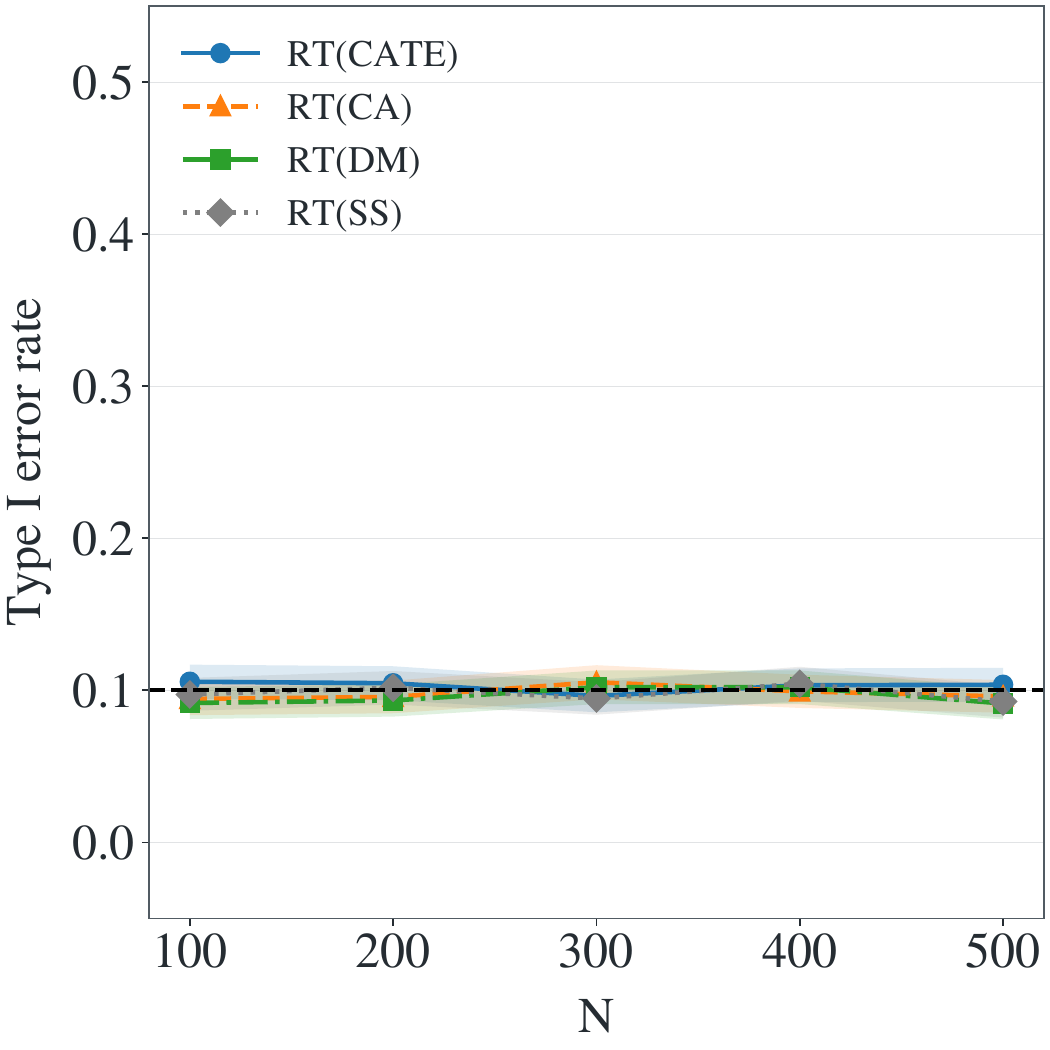}
            \caption{Repeated crossover: size.}
            \label{fig5:sub5}
        \end{subfigure}
        \hspace{0.02\linewidth}
        \begin{subfigure}[t]{0.38\linewidth}
            \centering
            \includegraphics[width=\linewidth]{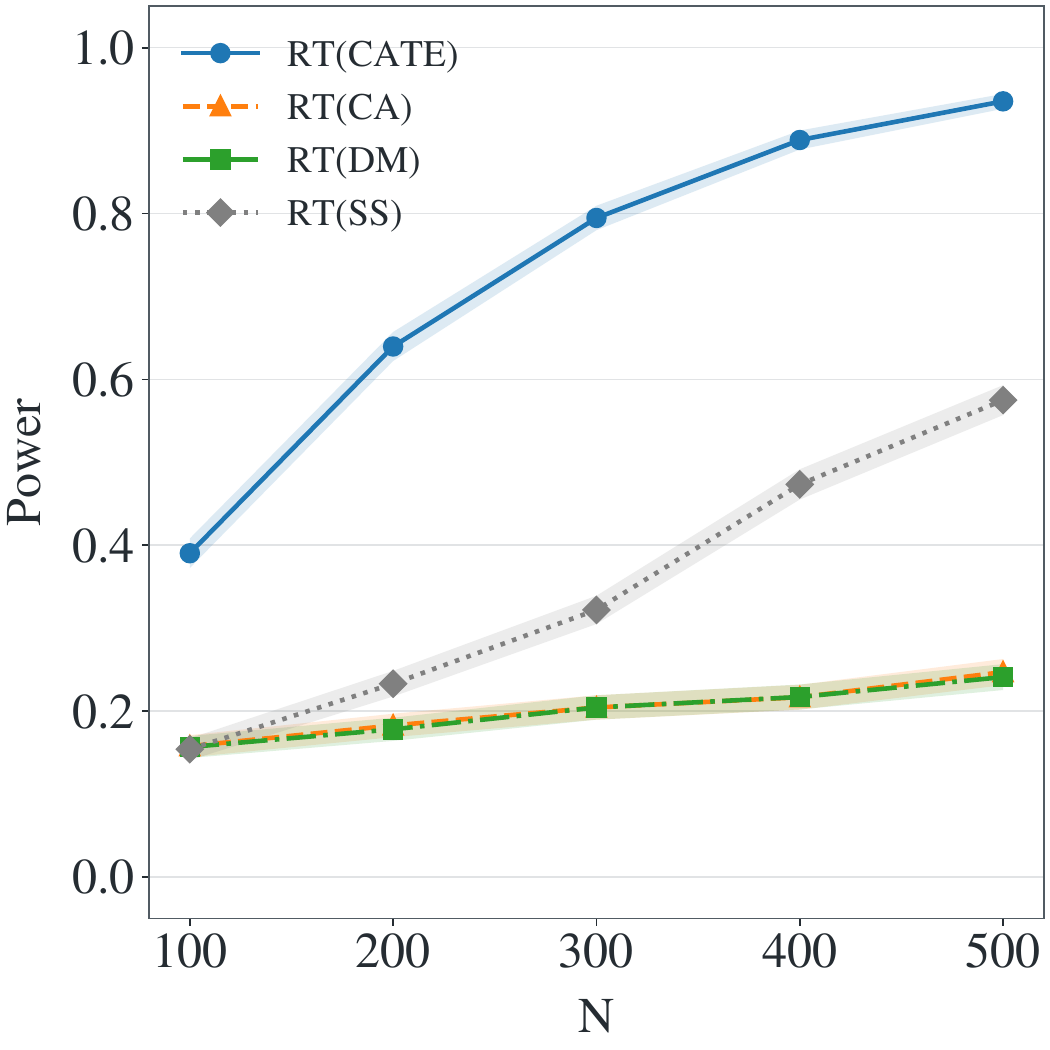}
            \caption{Repeated crossover: power.}
            \label{fig5:sub6}
        \end{subfigure}
    \end{minipage}
\caption{Size and power of randomization tests for the global
sharp null at nominal level \(0.10\). Each panel uses 3000
replications per sample size and 1000 randomization draws
per test. Shaded bands show the empirical rejection rate
plus or minus two Monte Carlo standard errors.}
    \label{fig:additional_designs}
\end{figure}

\clearpage


\vspace{-5pt}
\section{Experimental details}

\subsection{Unsigned CATE estimation for lag-specific tests}
\label{app:lag_cate_estimation}

The lagged staggered-adoption simulations in
\Cref{subsubsec:size_and_power} test \(H_{0,1,l}\) for
\(l=0,\ldots,4\) and \(H_{0,t,0}\) for \(t=1,\ldots,5\),
combining the latter \(p\)-values across adoption times.
The Chicago cohort comparisons in \Cref{subsec:real_data}
test \(H_{0,t,1}\), concerning effects one quarter after training.
Here we describe how the unsigned CATE estimator
can be adapted to the fitting data permitted for lag-specific tests.

For a given \((t,l)\), let \(\mathcal I\) and \(\mathcal J\)
be the subsets of units defined in \eqref{equ:set_I}. By
\Cref{prop:lag_conditional_validity}, nuisance fitting may use
\(\mathcal D_{t,l}\) in \eqref{eq:lag_invariant_data}, together
with \(\mathcal I\) and \(A_{\mathcal J}\).
One option, when \(\mathcal J\) is nonempty, is to use the full
outcome trajectories of units in \(\mathcal J\).
We estimate outcome means on this sample and estimate its residual
cross-moments by regressing products of the resulting residuals
on the covariates. Under a completely randomized design, the conditional cohort
probabilities for this fitting population are
\begin{equation}\label{eq:lag_fitting_probabilities}
\pi_a^{\mathcal J}
=
\frac{1}{|\mathcal J|}
\sum_{i\in\mathcal J}\mathbf 1_{\{A_i=a\}},
\qquad a\in[T+1].
\end{equation}
In the lagged simulations, the errors are independent of the
covariates and treatment start times, with covariance \(0.6^2I_5\).
Hence \(\mathbb E[\epsilon_{i,r}\epsilon_{i,s}\mid X_i,A_i]=0\)
for \(r\ne s\), and the zero error-covariance condition also
holds within \(\mathcal J\).
Let \(c_{r,s}^{\mathcal J}(x)\) denote the population residual
cross-moment after centering the outcomes by their conditional
means within \(\mathcal J\). Under the additive model and this
error condition, \eqref{equ:key_identity} gives, for
\(1\le s<r\le T\),
\[
\begin{aligned}
c_{r,s}^{\mathcal J}(x)
=
\sum_{a=1}^{s}\pi_a^{\mathcal J}
\tau_{r-a}(x)\tau_{s-a}(x)-
\left\{\sum_{a=1}^{r}\pi_a^{\mathcal J}\tau_{r-a}(x)\right\}
\left\{\sum_{a=1}^{s}\pi_a^{\mathcal J}\tau_{s-a}(x)\right\}.
\end{aligned}
\]
Thus the same moment method applies, with the original cohort
probabilities replaced by \eqref{eq:lag_fitting_probabilities}.
We use the NLS procedure in
\eqref{eq:opt2_nls_pointwise}, with the fitted residual moments
and the coefficients in $c_{r,s}^{\mathcal J}(x)$.
If all solutions to the population moment equations
within \(\Theta\) have the same value of \(|\tau_l(x)|\),
the compactness argument in the proof of
\Cref{thm:opt2_nls_consistency} also establishes
consistency of the estimated target magnitude,
even when the full CATE vector is not identified.
This requires consistent estimation of the residual
moments, a vanishing ridge penalty, and a global
minimizer of the NLS objective.

For testing, set \(\hat m_i=|\hat\tau_{\pm,l}(X_i)|\) for every unit
\(i\in\mathcal I\). We separately fit the outcome mean at time
\(t+l\) on \(\mathcal I\) and use the resulting residuals in
the sign rule \eqref{eq:lag_target_sign}.
The fitted outcome means and effect magnitudes remain fixed across randomization
draws, while signs are recomputed under each candidate assignment.
The conditional \(p\)-value in \eqref{equ:pvaluetl} then has the
validity guarantee in \Cref{prop:lag_conditional_validity}.

\subsection{Estimating equations when \(A_i\) is observed}\label{subsec:first_moment}

When the realized assignments are available for fitting,
\eqref{eq:yit_expression} gives linear moment equations for
estimating the signed CATEs. These equations do not depend on
the null being tested.  For \(k\in\{0,1,\dots,t-1\}\), define
\[
d_{t,k}(x)
:=\mathbb E\!\left[\,R_{i,t}\Big(\mathbf 1_{\{A_i=t-k\}}-\pi_{t-k}(X_i)\Big)\,\middle|\,X_i=x\right].
\]
Multiplying \eqref{eq:yit_expression} by \(\mathbf 1_{\{A_i=t-k\}}-\pi_{t-k}(X_i)\) and taking conditional expectations gives
\[
d_{t,k}(x)
=
\sum_{l=0}^{t-1}
\mathbb E\!\left[
\big(\mathbf 1_{\{A_i=t-l\}}-\pi_{t-l}(X_i)\big)
\big(\mathbf 1_{\{A_i=t-k\}}-\pi_{t-k}(X_i)\big)
\middle| X_i=x\right]
\tau_l(x),
\]
where the noise term drops by \(\mathbb E[\epsilon_{i,t}\mid X_i,A_i]=0\). Using mutual exclusivity of the indicators \(\{\mathbf 1_{\{A_i=a\}}:a\in[T+1]\}\), this simplifies to
\begin{equation}\label{eq:m_t_k}
d_{t,k}(x)
=
\pi_{t-k}(x)\tau_k(x)
-
\pi_{t-k}(x)\sum_{l=0}^{t-1}\pi_{t-l}(x)\tau_l(x).
\end{equation}
Suppose \(1-\pi_{\le t}(x)>0\). Summing \eqref{eq:m_t_k}
over \(k=0,\ldots,t-1\) and rearranging gives
\[
\sum_{l=0}^{t-1}\pi_{t-l}(x)\tau_l(x)
=
\frac{\sum_{k=0}^{t-1}d_{t,k}(x)}
{1-\pi_{\le t}(x)}.
\]
For each \(k\in\{0,\ldots,t-1\}\) with \(\pi_{t-k}(x)>0\),
substituting this identity into \eqref{eq:m_t_k} gives
\[
\tau_k(x)
=
\frac{d_{t,k}(x)}{\pi_{t-k}(x)}
+
\frac{\sum_{j=0}^{t-1}d_{t,j}(x)}
{1-\pi_{\le t}(x)}.
\]

\subsection{Staggered-adoption simulations}
\label{app:sim_details}

Here we give the fitting details for \Cref{subsubsec:size_and_power}.
Gradient-boosted regression trees (GBRT) use 80 trees, a learning
rate of 0.05, depth 1, and a minimum leaf size of 10. RT (CA)
uses GBRT for outcome regression, and RT (SS) uses GBRT to
estimate signed effects on its training sample.

In the lag-invariant setting,
RT (CATE) fits outcome means by quadratic polynomial ridge
regression with penalty \(10^{-3}\) and regresses squared residuals
on the covariates using GBRT. For each \(x\), we regress
\(\hat c_{t,t}(x)\) on \(v_t=\pi_{\le t}(1-\pi_{\le t})\)
across \(t\in[5]\) by least squares, including an intercept.
We truncate the fitted slope at zero and take its square root
to obtain \(\hat m_*(x)\), giving the unsigned CATE estimator
\(\hat\tau_{\pm}^{\mathrm{diag}}(x)\) in \eqref{eq:tau_diagonal}.

In the lagged setting, RT (CATE) fits outcome means by cubic polynomial ridge
regression, with the penalty chosen using only covariates to keep
the effective degrees of freedom at most \(\min(10,n/3)\), where
\(n\) is the fitting sample size. Residual cross-moments are fitted
by GBRT. The convex fit uses a trace penalty of \(10^{-3}\),
followed by NLS refinement with a ridge penalty of \(0.01\).
When testing \(H_{0,t,l}\), we use the eligible set
\(\mathcal I=\{i\in[N]:A_i=t\text{ or }A_i>t+l\}\)
in \eqref{equ:pvaluetl}. We fit the outcome mean at time
\(t+l\) on \(\mathcal I\) using the same ridge procedure.
Sign selection uses only the resulting target-time residuals
and the treatment labels under each candidate assignment.

\subsection{Parallel-group and crossover simulations}
\label{app:other_design_details}

To estimate the squared CATE in the experiments in
\Cref{subsec:simulations_additional}, we use the identity
\[
c_{t,s}(x)
=
\operatorname{Cov}(D_{i,t},D_{i,s})\tau(x)^2
+(\Sigma_\epsilon)_{t,s}.
\]
Under the stated covariance model, \((\Sigma_\epsilon)_{t,s}\)
depends only on the time gap \(|t-s|\). Each residual moment contains a treatment-effect contribution
and an error contribution. We subtract moments with equal
error contributions but different treatment contributions.
The error terms then cancel, leaving a known multiple of
\(\tau(x)^2\). The calculations for the three designs are
given below. Dividing the resulting difference by the known difference
in treatment variances or covariances gives \(\tau(x)^2\).

For the parallel-group design, the treatment variance is zero
at times 1 and 2 and equals \(0.4(1-0.4)=0.24\) at times
3 through 6. Subtracting the average of the two pre-treatment
residual variances from the average of the four post-treatment
residual variances cancels the error variance and gives
\(0.24\tau(x)^2\). Hence,
\[
\tau(x)^2
=
\frac{
\frac14\sum_{t=3}^{6}c_{t,t}(x)
-\frac12\{c_{1,1}(x)+c_{2,2}(x)\}
}{0.24}.
\]
For the two-period crossover design, let \(q=N_{\mathsf{AB}}/N\)
and \(v=q(1-q)\). Treatment indicators are
identical within a treatment period, giving covariance \(v\).
Across the switching boundary, \(D_{i,5}=1-D_{i,4}\), so their
covariance is \(-v\). The four adjacent pairs within periods
are \((1,2)\), \((2,3)\), \((3,4)\), and \((5,6)\), while
\((4,5)\) crosses the boundary. Subtracting the latter residual
moment from the average of the former cancels the common adjacent
error covariance. The treatment covariance difference is
\(2v\), giving
\[
\tau(x)^2
=
\frac{
\frac14\{c_{1,2}(x)+c_{2,3}(x)+c_{3,4}(x)+c_{5,6}(x)\}
-c_{4,5}(x)
}{2v}.
\]
For the repeated crossover design, treatment indicators within
each cycle are complementary and have covariance
\(-0.10(1-0.10)=-0.09\). Assignments are independent
across cycles, so the treatment covariances for \((2,3)\) and
\((4,5)\) are zero. Subtracting the average residual covariance
for the three within-cycle pairs from that for the two adjacent
pairs across cycle boundaries cancels the error covariance.
The treatment covariance difference is \(0.09\), so
\[
\tau(x)^2
=
\frac{1}{0.09}\left\{
\frac{c_{2,3}(x)+c_{4,5}(x)}{2}
-\frac{c_{1,2}(x)+c_{3,4}(x)+c_{5,6}(x)}{3}
\right\}.
\]
For each design, we substitute the fitted residual moments
into its formula for \(\tau(x)^2\). We truncate the result
at zero and take its square root to obtain the unsigned
CATE estimate.
For a candidate treatment path \(d_i\), the sign rule becomes
\[
\hat s_i(d_i)
\in\arg\min_{v\in\{-1,1\}}
\sum_{t=1}^{6}
\{\hat R_{i,t}-(d_{i,t}-q_t)v\hat m_i\}^{2},
\]
where \(q_t=\mathbb P(D_{i,t}=1)\) is the probability that a unit
receives treatment at time \(t\), \(\hat m_i\) is the fitted effect magnitude, and
\(\hat R_{i,t}\) is the fitted outcome residual.

\subsection{Consistency and warm-start experiment}
\label{app:warm_start_dgp}

This appendix describes the data-generating process and estimators used in the warm-start experiment in \Cref{subsubsec:warm_start}. We use \(T=5\) time periods and the completely randomized staggered-adoption cohort proportions
$\pi=(0.1,0.2,0.2,0.2,0.2,0.1),$
where the last component corresponds to the never-treated group. The covariate is one-dimensional:
$X_i\sim\mathrm{Unif}[-1,1].$
The untreated mean is affine in \(X_i\):
\[
\mu_{0,t}(X_i)
=
\alpha_t+\beta_tX_i,
\qquad
\alpha_t=0.55+0.1(t-1),
\qquad
\beta_t=0.3-0.03(t-1),
\]
for \(t=1,\ldots,5\). 
Observed outcomes are generated by
\[
Y_{i,t}
=
\mu_{0,t}(X_i)
+
\sum_{l=0}^{t-1}\mathbf 1_{\{A_i=t-l\}}\tau_l(X_i)
+
\varepsilon_{i,t},
\quad
\varepsilon_{i,t}\overset{\mathrm{iid}}{\sim}N(0,0.1^2).
\]

For the raw consistency experiments, the lagged CATE vector is
\[
(\tau_0(X_i),\ldots,\tau_4(X_i))
=
(1.3,\ 1.15,\ 1.0,\ 0.85,\ 0.7+0.35X_i).
\]
This CATE vector is globally identified up to sign from the
population off-diagonal moments for every \(x\in[-1,1]\).
Using \(\pi_1=0.1\), \(\pi_2=\pi_3=0.2\), and
\((\tau_0,\tau_1,\tau_2)=(1.3,1.15,1.0)\), we calculate
\(c_{2,1}\), \(c_{3,1}\), and \(c_{3,2}\) from
\[
c_{t,s}
=
\sum_{a=1}^{s}\pi_a\tau_{t-a}\tau_{s-a}
-
\left(\sum_{a=1}^{t}\pi_a\tau_{t-a}\right)
\left(\sum_{a=1}^{s}\pi_a\tau_{s-a}\right),
\qquad 1\le s<t\le3.
\]
All three moments are nonzero. The first-column recursion in
\Cref{prop:global_id} uses \(c_{2,1}\) and \(c_{3,1}\) to
determine \(d_1(z)\) and \(d_2(z)\). Substituting these expressions
and \(c_{3,2}\) into \(P_{3,2}(z)\) shows that its only positive
root is \(z=169/100=\tau_0(x)^2\). The remaining first-column
moments \(c_{4,1}(x)\) and \(c_{5,1}(x)\) then determine the last
two lag effects up to the same sign. Consistency of the NLS
estimator additionally requires the moment-estimation and
optimization conditions in \Cref{thm:opt2_nls_consistency}.

We run these experiments over the sample size
\[
N\in\{200,400,800,1600,3200,6400,12800,25600\}.
\]
The trace penalty in the convex relaxation objective is set to
$\lambda_{\mathrm{trace}}=N^{-1/2},$
while the ridge penalty in the NLS objective is always set to zero.

The warm-start experiment uses signed effects $\tau(X_i)=s(X_i)(1.3,1.15,1.0,0.85,0.7+0.35X_i)$ with $s(x)=2\mathbf 1_{\{x\geq 0\}}-1$, the means, errors and design are unchanged. The full sample size is fixed at \(N=25{,}600\). We first fit the nonparametric residual-covariance estimator using all covariates and outcomes. We then reveal the treatment assignments for only a subset $\mathcal I_M\subset\{1,\ldots,N\}$ and  $|\mathcal I_M|=M,$
where
\[
M\in\{25,50,100,200,400,800,1600,3200\}.
\] 
We learn the signed  CATE vector $\tau(x)$
using only these \(M\) assignments.

Denote the units that start treatment at time \(a\) or remain untreated at time \(a+l\) by
\[
\mathcal E_{a,l}=
\{i\in\mathcal I_M: A_i=a \text{ or } A_i>a+l\},
\]
and define the conditional assignment probability
\[
e_{a,l}
=
\frac{\pi_a}{\pi_a+\sum_{b>a+l}\pi_b}.
\]
For a candidate lag-CATE function $f(X_i)=(f_0(X_i),\ldots,f_{T-1}(X_i))$,
define
\[
\ell_M(f)
=
\frac{1}{M}
\sum_{l=0}^{T-1}
\sum_{a=1}^{T-l}
\sum_{i\in\mathcal E_{a,l}}
\left[
\hat R_{i,a+l}
-
\left\{
\mathbf 1_{\{A_i=a\}}-e_{a,l}
\right\}
f_l(X_i)
\right]^2,
\]
where \(\hat R_{i,a+l}=Y_{i,a+l}-\hat\mu_{a+l}(X_i)\), and 
\(\hat\mu_{a+l}(X_i)\) is the fitted outcome function using all covariates and outcomes at time \(a+l\).
The direct R-learner estimator directly learns the lagged CATE from the \(M\) labeled units:
\[
\hat \tau^{\mathrm{direct}}
\in
\argmin_{f\in\mathcal F^T}
\ell_M(f).
\]
The warm-start estimator begins with the assignment-free unsigned estimate \(\hat\tau_{\pm}(X_i)\). Instead of learning the full vector \(f\), it only learns a sign function \(s(\cdot)\in\{-1,1\}\) by minimizing the same loss over functions of the form
$f_s(X_i)=s(X_i)\hat\tau_{\pm}(X_i).$
Thus,
\[
\hat s
\in
\argmin_{s\in\mathcal S}
\ell_M(f_s),
\qquad
\hat\tau(\cdot)
=
\hat s(\cdot)\hat\tau_{\pm}(\cdot).
\]
The warm-start estimator turns the full CATE estimation problem into a lower-dimensional sign-learning problem. Both the direct R-learner and the warm-start sign model are implemented using gradient-boosted regression trees.


\subsection{Parametric estimation}
\label{app:parametric}

In the main text, we present the off-diagonal estimator in a pointwise, nonparametric form. Here we briefly describe a parametric version. The idea is the same, but we restrict the lagged CATE functions to a finite-dimensional linear model.

Suppose that each lagged CATE is linear in a \(p\)-dimensional feature vector \(\phi(x)\):
\[
\tau_l(x)=\beta_l^\top\phi(x),
\qquad l=0,\ldots,T-1.
\]
Let $\beta:=(\beta_0^\top,\ldots,\beta_{T-1}^\top)^\top\in\mathbb R^{pT}$. We write \(\tau(x)=P(x)\beta\), where
\(
P(x):=I_T\otimes \phi(x)^\top.
\)
Substituting this into the off-diagonal quadratic moments gives, for each \((t,s)\in\mathcal P\),
\[
\tau(x)^\top H_{t,s}(x)\tau(x)
=
\beta^\top G_{t,s}(x)\beta,
\qquad
G_{t,s}(x):=P(x)^\top H_{t,s}(x)P(x).
\]
Thus the NLS and convex-relaxation estimators from \Cref{subsubsec:cross_est} can be applied after replacing \(H_{t,s}(x)\) by \(G_{t,s}(x)\) and summing the moment-matching loss over the sample. In the convex version, we estimate the positive
semidefinite matrix \(\beta\beta^\top\).
Let \((\hat\lambda_1,\hat u_1)\) denote the leading
eigenpair of the estimated matrix.
The spectral initialization is
\[
\hat\beta^{\,\mathrm{sp}}=\sqrt{\hat\lambda_1}\hat u_1.
\]
We identify the coefficients using the moment conditions
\begin{equation}\label{eq:projected_moments}
\mathbb E\!\left[
\Big(R_{i,t}R_{i,s}-\beta^\top G_{t,s}(X_i)\beta\Big)\phi_j(X_i)
\right]=0,
\quad
1\le s<t\le T,\quad j=1,\ldots,p.
\end{equation}
These give \(p\binom{T}{2}\) scalar equations for \(pT\) unknown coefficients. For \(T\ge3\), the number of equations is at least the number of unknowns. Since the moments are quadratic, \(\beta\) and \(-\beta\) are observationally equivalent. As shown in \Cref{prop:opt2_identification}, a sufficient local identification condition is that the Jacobian of the stacked moment map in \eqref{eq:projected_moments} with respect to \(\beta\) has full column rank \(pT\) at the true coefficient vector.

We also repeat the consistency experiment of \Cref{subsubsec:warm_start} under this parametric specification. \Cref{fig:parametric-consistency} shows the same pattern as the nonparametric experiment: the estimation error for the coefficient vector \(\beta\) decreases as the sample size grows.

\begin{figure}[h]
    \centering
    \includegraphics[width=0.4\linewidth]{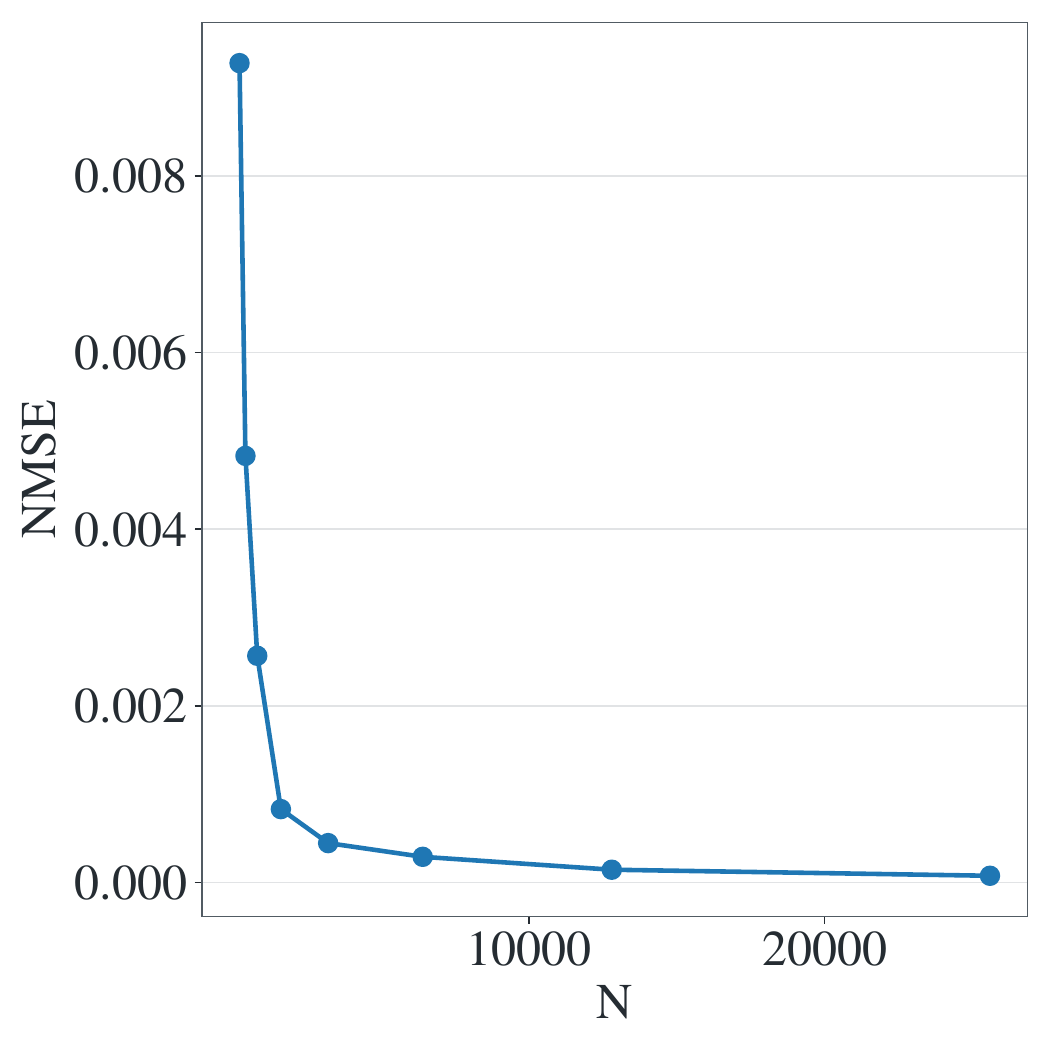}
    \caption{Consistency experiment for the parametric off-diagonal estimator.}
    \label{fig:parametric-consistency}
\end{figure}

\section{Technical proofs}

\subsection{Proof of \Cref{thm:validity}}

\begin{proof}
Let
\(
W:=(X_{[N]},Y_{[N],[T]},\hat\eta).
\)
Under the sharp null \(H_0\), the observed outcomes do not depend on the realized assignment vector \(A_{[N]}\). Conditional on \(W\), the only randomness in \(T(O_{[N]};\hat\eta)\) comes from \(A_{[N]}\), while the only randomness in \(T(\tilde O_{[N]};\hat\eta)\) comes from the simulated assignment \(\tilde A_{[N]}\).

The condition $\hat\eta \independent A_{[N]}\mid X_{[N]},Y_{[N],[T]}$
implies that, conditional on \(W\), the observed assignment vector still follows the design distribution. By construction, \(\tilde A_{[N]}\) is an independent draw from the same design distribution. Hence, conditional on \(W\), the observed and simulated test statistics $Z:=T(O_{[N]};\hat\eta)$
and $\tilde Z:=T(\tilde O_{[N]};\hat\eta)$
have the same distribution.
Fix \(W=w\). Since the assignment space is finite in the randomized designs considered here, \(Z\mid W=w\) has finite support. Let
\[
G_w(z):=\tilde{\mathbb P}(\tilde Z\ge z\mid W=w).
\]
The randomization \(p\)-value is \(p(O_{[N]};\hat\eta)=G_W(Z)\). For any \(\alpha\in(0,1)\), define
\[
D_\alpha(w):=\{z:\,G_w(z)\le \alpha\}.
\]
If \(D_\alpha(w)\) contains no support point of
\(Z\mid W=w\), then
\(\mathbb P\{G_w(Z)\le\alpha\mid W=w\}=0\).
Otherwise, let \(z_\alpha(w)\) be the smallest support point of
\(Z\mid W=w\) that belongs to \(D_\alpha(w)\).
Since \(G_w\) is nonincreasing,
\[
\{G_w(Z)\le\alpha\}\subseteq\{Z\ge z_\alpha(w)\}.
\]
Therefore,
\[
\mathbb P\{G_w(Z)\le \alpha\mid W=w\}
\le
\mathbb P\{Z\ge z_\alpha(w)\mid W=w\}
=
G_w(z_\alpha(w))
\le \alpha.
\]
Taking conditional expectation given \((X_{[N]},Y_{[N],[T]})\) gives
\[
\mathbb P_{H_0}\{p(O_{[N]};\hat\eta)\le\alpha\mid X_{[N]},Y_{[N],[T]}\}\le\alpha,
\]
which proves \eqref{equ:validity}.
\end{proof}

\subsection{Proof of \Cref{prop:lag_conditional_validity}}

\begin{proof}
Fix the covariates and potential outcomes, and condition on
\((\mathcal I,A_{\mathcal J},U)\), where \(U\) denotes any
independent randomness used in fitting.
Under \(H_{0,t,l}\), no interference, and no anticipation, the data
\(\mathcal D_{t,l}\) defined in \eqref{eq:lag_invariant_data} and used
to construct the \(p\)-value remain unchanged across the permitted assignments. 
Consequently,
\[
\mathbb P\!\left(
A_{\mathcal I}=a_{\mathcal I}
\,\middle|\,
\mathcal D_{t,l},\mathcal I,A_{\mathcal J},U
\right)
=
\mathbb P\!\left(
A_{\mathcal I}=a_{\mathcal I}
\,\middle|\,
X_{[N]},\mathcal I,A_{\mathcal J}
\right).
\]
The fitted nuisance functions are fixed under this conditioning,
and the observed assignment follows the same conditional law
used to define \(p_{t,l}\) in \eqref{equ:pvaluetl}.
The upper-tail argument in the proof of \Cref{thm:validity}
therefore gives
\[
\mathbb P_{H_{0,t,l}}\!\left\{
p_{t,l}\le\alpha
\,\middle|\,
\mathcal D_{t,l},\mathcal I,A_{\mathcal J},U
\right\}
\le\alpha.
\]
This bound holds for every fixed potential-outcome array satisfying
the null and every fixed value of \(U\). Averaging over the fitting
randomness and, when random, the potential outcomes gives the
conditional inequality in the proposition.

The result also extends to Monte Carlo randomization.
Let \(A_{[N]}^{(1)},\ldots,A_{[N]}^{(M)}\) be independent draws
from the conditional assignment distribution in
\eqref{equ:pvaluetl}, and define
\begin{equation}\label{eq:lag_mc_pvalue}
\hat p_{t,l}
=
\frac{
1+\sum_{m=1}^{M}
\mathbf 1\!\left\{
T_{t,l}(\mathcal D_{t,l},A_{[N]}^{(m)};\hat\eta)
\ge
T_{t,l}(\mathcal D_{t,l},A_{[N]};\hat\eta)
\right\}
}{M+1}.
\end{equation}
Conditional on \((\mathcal D_{t,l},\mathcal I,A_{\mathcal J},U)\),
the observed statistic and the \(M\) simulated statistics are
exchangeable. The numerator in \eqref{eq:lag_mc_pvalue} is the
upper rank of the observed statistic, with ties counted using
\(\ge\). At most \(\lfloor(M+1)\alpha\rfloor\) of the
\(M+1\) statistics can have upper rank at most \((M+1)\alpha\).
Exchangeability and the same averaging argument give
\[
\mathbb P_{H_{0,t,l}}\!\left\{
\hat p_{t,l}\le\alpha
\,\middle|\,
\mathcal D_{t,l},\mathcal I,A_{\mathcal J}
\right\}
\le
\frac{\lfloor(M+1)\alpha\rfloor}{M+1}
\le\alpha.
\]
\end{proof}

\subsection{Proof of \Cref{prop:diagonal_moment}}

\begin{proof}
Fix \(x\in\mathcal X\). Under Assumption~\ref{assumption:diagonal}, \(\tau_l(x)=\tau_*(x)\) for all \(l\). From \eqref{eq:yit_expression},
\[
R_{i,t}
=
\tau_*(x)\sum_{l=0}^{t-1}\Bigl(\mathbf 1_{\{A_i=t-l\}}-\pi_{t-l}(x)\Bigr)
+\epsilon_{i,t}.
\]
Since
\[
\sum_{l=0}^{t-1}\mathbf 1_{\{A_i=t-l\}}=\mathbf 1_{\{A_i\le t\}},
\qquad
\sum_{l=0}^{t-1}\pi_{t-l}(x)=\pi_{\le t}(x),
\]
we have
\[
R_{i,t}
=
\tau_*(x)\Bigl(\mathbf 1_{\{A_i\le t\}}-\pi_{\le t}(x)\Bigr)+\epsilon_{i,t}.
\]
Therefore,
\begin{align*}
c_{t,t}(x)
&=\mathbb E[R_{i,t}^2\mid X_i=x] \\
&=\tau_*^2(x)\,
\mathbb E\!\left[\Bigl(\mathbf 1_{\{A_i\le t\}}-\pi_{\le t}(x)\Bigr)^2\mid X_i=x\right]
+\mathbb E[\epsilon_{i,t}^2\mid X_i=x],
\end{align*}
where the cross term is zero because \(\mathbb E[\epsilon_{i,t}\mid X_i,A_i]=0\). Since \(\mathbf 1_{\{A_i\le t\}}\mid X_i=x\) is Bernoulli with success probability \(\pi_{\le t}(x)\),
\[
\mathbb E\!\left[\Bigl(\mathbf 1_{\{A_i\le t\}}-\pi_{\le t}(x)\Bigr)^2\mid X_i=x\right]
=
\pi_{\le t}(x)(1-\pi_{\le t}(x)).
\]
Using \(\mathbb E[\epsilon_{i,t}^2\mid X_i=x]=\sigma_\epsilon^2(x)\), we obtain
\[
c_{t,t}(x)=\tau_*^2(x)v_t(x)+\sigma_\epsilon^2(x),
\qquad
v_t(x)=\pi_{\le t}(x)(1-\pi_{\le t}(x)).
\]
\end{proof}

\subsection{Proof of \Cref{prop:diagonal_identification}}

\begin{proof}
By Proposition~\ref{prop:diagonal_moment}, we have
\[
c_{t,t}(x)=\tau_*^2(x)v_t(x)+\sigma_\epsilon^2(x),
\qquad t\in[T].
\]
Averaging over \(t\) gives $\bar c(x)=\tau_*^2(x)\bar v(x)+\sigma_\epsilon^2(x).$
Subtracting yields
\[
c_{t,t}(x)-\bar c(x)
=
\tau_*^2(x)\{v_t(x)-\bar v(x)\}.
\]
If \(v_t(x)\neq \bar v(x)\), then
\[
\tau_*^2(x)=\frac{c_{t,t}(x)-\bar c(x)}{v_t(x)-\bar v(x)}.
\]
Thus \(m_*(x)=|\tau_*(x)|\) is identified by taking the nonnegative square root.
\end{proof}

\subsection{Proof of \Cref{prop:diag_consistency}}

\begin{proof}
Since \(T\) is fixed and \(\hat c_{t,t}(x)\xrightarrow{\;p\;}c_{t,t}(x)\) for all \(t\),
\[
\hat{\bar c}(x)=\frac1T\sum_{t=1}^T\hat c_{t,t}(x)
\xrightarrow{\;p\;}
\bar c(x).
\]
Hence
\[
\frac{\hat c_{t,t}(x)-\hat{\bar c}(x)}{v_t(x)-\bar v(x)}
\xrightarrow{\;p\;}
\frac{c_{t,t}(x)-\bar c(x)}{v_t(x)-\bar v(x)}
=\tau_*^2(x),
\]
where the last equality follows from Proposition~\ref{prop:diagonal_identification}. Since the map \(u\mapsto \sqrt{\max\{u,0\}}\) is continuous,
\[
\hat m_*(x)
=
\sqrt{\max\left\{\frac{\hat c_{t,t}(x)-\hat{\bar c}(x)}{v_t(x)-\bar v(x)},0\right\}}
\xrightarrow{\;p\;}
\sqrt{\tau_*^2(x)}=m_*(x).
\]
\end{proof}

\subsection{Proof of \Cref{prop:offdiag_moment}}\label{sect:proof_offdiag_moment}

\begin{proof}
Fix \(x\in\mathcal X\) and \(1\le s<t\le T\). Throughout, condition on \(X_i=x\). Let
\[
\Delta_a:=\mathbf 1_{\{A_i=a\}}-\pi_a(x),
\qquad a\in[T].
\]
Then \eqref{eq:yit_expression} gives
\[
R_{i,t}=\sum_{l=0}^{t-1}\Delta_{t-l}\tau_l(x)+\epsilon_{i,t},
\qquad
R_{i,s}=\sum_{k=0}^{s-1}\Delta_{s-k}\tau_k(x)+\epsilon_{i,s}.
\]
Under Assumption~\ref{assumption:offdiagonal}, \(\mathbb E[\epsilon_{i,t}\epsilon_{i,s}\mid X_i=x]=0\) for \(t\neq s\). The cross terms involving one \(\epsilon\) are also zero by \(\mathbb E[\epsilon_{i,t}\mid X_i,A_i]=0\). Therefore,
\[
c_{t,s}(x)
=
\sum_{l=0}^{t-1}\sum_{k=0}^{s-1}
\mathbb E[\Delta_{t-l}\Delta_{s-k}\mid X_i=x] \tau_l(x)\tau_k(x).
\]
For \(a,b\in[T]\), mutual exclusivity of the treatment start indicators gives
\[
\mathbb E[\Delta_a\Delta_b\mid X_i=x]
=
\pi_a(x)\mathbf 1_{\{a=b\}}-\pi_a(x)\pi_b(x).
\]
Substituting \(a=t-l\) and \(b=s-k\), and using \(t>s\), we obtain
\begin{equation}\label{eq:cts_formula_rewrite_app}
c_{t,s}(x)
=
\sum_{k=0}^{s-1}\pi_{s-k}(x)\tau_{t-s+k}(x)\tau_k(x)
-
\bar\tau_t(x)\bar\tau_s(x),
\end{equation}
where
$\bar\tau_t(x):=\sum_{l=0}^{t-1}\pi_{t-l}(x)\tau_l(x).$
To write the moment as a quadratic form, define
\[
\pi^{(t)}(x):=(\pi_t(x),\pi_{t-1}(x),\dots,\pi_1(x),0,\dots,0)^\top\in\mathbb R^T,
\]
\[
\pi^{(s)}(x):=(\pi_s(x),\pi_{s-1}(x),\dots,\pi_1(x),0,\dots,0)^\top\in\mathbb R^T,
\]
\[
D_s(x):=\operatorname{diag}(\pi_s(x),\pi_{s-1}(x),\dots,\pi_1(x),0,\dots,0)\in\mathbb R^{T\times T},
\]
and let \(J_{t-s}\in\mathbb R^{T\times T}\) be the shift matrix with entries
$(J_{t-s})_{a,b}=\mathbf 1_{\{a=b+t-s\}}.$
Then
\[
\tau(x)^\top J_{t-s}D_s(x)\tau(x)
=
\sum_{k=0}^{s-1}\pi_{s-k}(x)\tau_{t-s+k}(x)\tau_k(x),
\]
and
\[
\tau(x)^\top \pi^{(t)}(x)\pi^{(s)}(x)^\top\tau(x)=\bar\tau_t(x)\bar\tau_s(x).
\]
Thus, with
\[
G_{t,s}(x):=J_{t-s}D_s(x)-\pi^{(t)}(x)\pi^{(s)}(x)^\top,
\]
we have \(c_{t,s}(x)=\tau(x)^\top G_{t,s}(x)\tau(x)\). Since only the symmetric part matters in a scalar quadratic form, define
\[
H_{t,s}(x):=\frac12\{G_{t,s}(x)+G_{t,s}(x)^\top\}.
\]
Then
$c_{t,s}(x)=\tau(x)^\top H_{t,s}(x)\tau(x),$
which proves the quadratic representation.
\end{proof}

\subsection{Proof of \Cref{prop:opt2_identification}}

\begin{proof}
Fix \(x\in\mathcal X\).
By assumption, \(J_x(\tau(x))\) has full column rank \(T\).
Hence there exists a set of \(T\) coordinates
\(\mathcal I\subseteq\mathcal P\) such that the corresponding
submatrix of \(J_x(\tau(x))\) is invertible.
Let \(P_{\mathcal I}\) be the coordinate projection onto these
entries and define
\[
G:=P_{\mathcal I}\circ F_x:\mathbb R^T\to\mathbb R^T.
\]
By the inverse function theorem, there exists an open
neighborhood \(\mathcal N_x\) of \(\tau(x)\) on which
\(G\), and therefore \(F_x\), is injective.
Since \(c(x)=F_x(\tau(x))\), we have
\[
\{u\in\mathcal N_x:F_x(u)=c(x)\}=\{\tau(x)\}.
\]
Since each component of \(F_x\) is a quadratic form,
\(F_x(-u)=F_x(u)\) for every \(u\in\mathbb R^T\).
Thus $\{u\in-\mathcal N_x:F_x(u)=c(x)\}=\{-\tau(x)\}.$
Combining these two equalities gives
\[
\{u\in\mathcal N_x\cup(-\mathcal N_x):F_x(u)=c(x)\}
=\{\tau(x),-\tau(x)\}.
\]
\end{proof}

\subsection{Proof of \Cref{prop:global_id}}

\begin{proof}
Fix \(x\), and suppress it in the notation. For any candidate vector
\(u=(u_0,\ldots,u_{T-1})^\top\), the first-column moment formula in
\eqref{eq:cts_formula_rewrite_app} gives
\[
(F_x(u))_{l+1,1}
=
\pi_1u_0
\left\{(1-\pi_1)u_l
-\sum_{a=2}^{l+1}\pi_a u_{l+1-a}\right\},
\qquad l=1,\ldots,T-1.
\]
If \(u_0=0\), all these moments vanish. The nonzero first-column
moment in \eqref{eq:global_id_conditions} therefore implies that
every vector satisfying \(F_x(u)=c(x)\) has \(u_0\ne0\).

For such a vector, set \(z=u_0^2>0\). The products \(u_0u_l\)
satisfy the recursion in \eqref{eq:scalar_id_recursion}, starting
from \(u_0^2=z\). Since \(1-\pi_1>0\), the recursion determines
these products uniquely, so \(u_0u_l=d_l(z)\). Consequently,
\[
u=\pm\frac{d(z)}{\sqrt z},
\]
with the same sign for all coordinates. For every remaining
moment pair \(1<s<t\le T\),
\[
(F_x(u))_{t,s}-c_{t,s}
=
\frac{d(z)^\top H_{t,s}d(z)-zc_{t,s}}{z}
=
\frac{P_{t,s}(z)}{z}.
\]
Since \(F_x(u)=c(x)\), we have \(P_{t,s}(z)=0\) for every
\(1<s<t\le T\). Together with \(z>0\), this gives
\(z\in\mathcal Z_x\) by the definition of \(\mathcal Z_x\).

Conversely, let \(z\in\mathcal Z_x\) and set
\(u=\pm d(z)/\sqrt z\), using the same sign for all coordinates.
Because \(d_0(z)=z\), we have \(u_0=\pm\sqrt z\) and
\(u_0u_l=d_l(z)\). Substituting these products into the first-column
moment formula and using \eqref{eq:scalar_id_recursion} gives
\[
\begin{aligned}
(F_x(u))_{l+1,1}
&=
\pi_1\left\{
(1-\pi_1)d_l(z)
-\sum_{a=2}^{l+1}\pi_a d_{l+1-a}(z)
\right\}\\
&=c_{l+1,1},
\qquad l=1,\ldots,T-1.
\end{aligned}
\]
For the remaining moments, \eqref{eq:scalar_id_polynomial}
and \(z\in\mathcal Z_x\) give
\[
\begin{aligned}
(F_x(u))_{t,s}
&=\frac{d(z)^\top H_{t,s}d(z)}{z}\\
&=c_{t,s}+\frac{P_{t,s}(z)}{z}
=c_{t,s},
\qquad 1<s<t\le T.
\end{aligned}
\]
Thus \(F_x(u)=c(x)\). The two directions together establish
\[
\{u\in\mathbb R^T:F_x(u)=c(x)\}
=
\left\{\pm\frac{d(z)}{\sqrt z}:z\in\mathcal Z_x\right\}.
\]

The population relation ensures that \(\tau_0(x)^2>0\) belongs
to \(\mathcal Z_x\). If \(|\mathcal Z_x|=1\), the only solutions
are \(\tau(x)\) and \(-\tau(x)\), establishing global
identification up to sign.
\end{proof}

\subsection{Proof of \Cref{prop:rank_sufficient_design}}

\begin{proof}
Fix \(x\) and write \(r=\rho\). Since \(F_x\) is homogeneous of
degree two and \(\theta(x)\ne0\), dividing every candidate vector
by \(\theta(x)\) reduces the problem to the true vector
\[
\gamma=(1,r,r^2,\ldots,r^{T-1})^\top.
\]
In this rescaled problem, write \(c=F_x(\gamma)\) for the population
moments and let \(u\) denote a candidate vector. Write \(h=1-q\),
and define
\[
\delta=hr-q,\qquad
\varepsilon=hr^2-q(1+r)=r\delta-q,\qquad
\zeta=hr^3-q(1+r+r^2)=r\varepsilon-q.
\]
The first-column moments at \(\gamma\) include
\[
c_{2,1}=q\delta,\qquad
c_{3,1}=q\varepsilon,\qquad
c_{4,1}=q\zeta.
\]
If \(\delta\ne0\), then \(c_{2,1}\ne0\). If \(\delta=0\), then
\(\varepsilon=-q\), so \(c_{3,1}=-q^2\ne0\). In either case,
the scalar characterization in \Cref{prop:global_id} applies.

We now find the possible roots \(z=u_0^2\). The first-column
recursion has the explicit solution
\begin{equation}\label{eq:geometric_id_affine}
d_j(z)=r^j+a_j(z-1),
\qquad
a_0=1,\quad a_j=\frac{q}{h^j}\quad(j\ge1).
\end{equation}
To verify this expression, first note that \(d_j(1)=r^j\), since
the true vector generates the moments. The coefficient of \(z\)
in the recursion satisfies
\[
a_0=1,\qquad
a_j=\frac{q}{h}\sum_{k=0}^{j-1}a_k.
\]
The values in \eqref{eq:geometric_id_affine} satisfy this recursion
because \(\sum_{k=0}^{j-1}a_k=h^{-(j-1)}\).

It is enough to show that \(P_{3,2}\) and \(P_{4,2}\) have no common
positive root other than \(1\), since every \(z\in\mathcal Z_x\)
must satisfy both equations. We use these two moment pairs because
they are available when \(T\ge4\) and their polynomial coefficients
have simple expressions. Recall the quadratic representation in \eqref{equ:key_identity}.
For a vector \(d\in\mathbb R^T\), write
\(Q_{t,s}(d):=d^\top H_{t,s}d\).
Under the equal-probability design, we have
\[
\begin{split}
Q_{3,2}(d)
&=
q(d_0d_1+d_1d_2)
-q^2(d_0+d_1+d_2)(d_0+d_1),\\
Q_{4,2}(d)
&=
q(d_0d_2+d_1d_3)
-q^2(d_0+d_1+d_2+d_3)(d_0+d_1).
\end{split}
\]
Accordingly, \(P_{3,2}(z)=Q_{3,2}(d(z))-zc_{3,2}\) and
\(P_{4,2}(z)=Q_{4,2}(d(z))-zc_{4,2}\). Both polynomials vanish
at the true value \(z=1\). We determine their second roots from
their leading and constant coefficients.

The leading coefficients are obtained by substituting
\(a=(a_0,\ldots,a_{T-1})^\top\) into the quadratic forms. Using
\[
a_0+a_1=\frac1h,\qquad
a_0+a_1+a_2=\frac1{h^2},\qquad
a_0+a_1+a_2+a_3=\frac1{h^3}
\]
gives
\[
\begin{split}
A_{3,2}
&=Q_{3,2}(a)
=\frac{q^2}{h}+\frac{q^3}{h^3}
-\frac{q^2}{h^3}
=-\frac{q^3}{h^2},\\
A_{4,2}
&=Q_{4,2}(a)
=\frac{q^2}{h^2}+\frac{q^3}{h^4}
-\frac{q^2}{h^4}
=-\frac{q^3}{h^3}.
\end{split}
\]
For the constant coefficients, evaluate at \(z=0\). Here
\(d_0(0)=0\) and \(d_1(0)=\delta/h\). The recursion also gives
\[
hd_2(0)-qd_1(0)=\varepsilon,\qquad
hd_3(0)-q\{d_1(0)+d_2(0)\}=\zeta.
\]
Consequently,
\[
\begin{split}
C_{3,2}
&=
qd_1(0)\{hd_2(0)-qd_1(0)\}
=\frac{q\delta\varepsilon}{h},\\
C_{4,2}
&=
qd_1(0)
\bigl[hd_3(0)-q\{d_1(0)+d_2(0)\}\bigr]
=\frac{q\delta\zeta}{h}.
\end{split}
\]
Since each polynomial has one root at \(1\), these coefficients
give the factorizations
\begin{equation}\label{eq:geometric_id_factorization}
\begin{aligned}
P_{3,2}(z)
&=-\frac{q^3}{h^2}(z-1)(z-z_{3,2}),
&\qquad z_{3,2}&=-\frac{h\delta\varepsilon}{q^2},\\
P_{4,2}(z)
&=-\frac{q^3}{h^3}(z-1)(z-z_{4,2}),
&\qquad z_{4,2}&=-\frac{h^2\delta\zeta}{q^2}.
\end{aligned}
\end{equation}

It remains to compare the second roots. From the definitions
of \(\delta,\varepsilon,\zeta\),
\[
\begin{split}
\varepsilon-h\zeta
&=(1-hr)\varepsilon+hq\\
&=(h-\delta)(r\delta-q)+hq\\
&=\delta\{hr+q-r\delta\}\\
&=\delta\{q(1+r^2)+r(1-r)\}.
\end{split}
\]
Therefore,
\[
z_{4,2}-z_{3,2}
=
\frac{h\delta^2}{q^2}
\{q(1+r^2)+r(1-r)\}.
\]
The expression in braces is strictly positive because \(q>0\)
and \(0<r<1\). If \(\delta\ne0\), the second roots are distinct,
so \(1\) is the only root common to the two polynomials.
If \(\delta=0\), then \(z_{3,2}=0\);
\eqref{eq:geometric_id_factorization} shows that the only
positive root of \(P_{3,2}\) is again \(1\).

Thus every vector \(u\) satisfying \(F_x(u)=F_x(\gamma)\)
must have \(z=u_0^2=1\).
By \eqref{eq:geometric_id_affine}, we have \(d_j(1)=r^j\)
for every \(j=0,\ldots,T-1\). Hence \Cref{prop:global_id}
gives
\[
u=\pm\frac{d(1)}{\sqrt{1}}
=\pm(1,r,\ldots,r^{T-1})^\top
=\pm\gamma.
\]
Now consider any vector \(v\in\mathbb R^T\) satisfying
\(F_x(v)=c(x)=F_x(\tau(x))\).
Since \(F_x\) is homogeneous of degree two and
\(\tau(x)=\theta(x)\gamma\), we have
\[
F_x\left(\frac{v}{\theta(x)}\right)
=\frac{F_x(v)}{\theta(x)^2}
=F_x(\gamma).
\]
It follows from the normalized result above that
\(v/\theta(x)=\pm\gamma\), or equivalently,
\(v=\pm\tau(x)\).
Conversely, both \(\tau(x)\) and \(-\tau(x)\) satisfy
the moment equations because \(F_x\) is quadratic.
The only vectors \(v\in\mathbb R^T\) satisfying
\(F_x(v)=c(x)\) are \(\tau(x)\) and \(-\tau(x)\).
\end{proof}

\subsection{Proof of \Cref{prop:gaussian_sign_recovery}}
\label{sect:proof_sign_recover}

\begin{proof}
Fix a unit \(i\) and condition on \(X_i\) and an observed assignment
\(A_i=a\) with \(\pi_a(X_i)>0\). Write
\(\nu=\nu_i(a)=\Pi(X_i,a)\tau_\pm(X_i)\).
The CATE \(\tau_\pm(X_i)\) has the correct relative signs across
lag coordinates, and \(\tau(X_i)=s_i\tau_\pm(X_i)\).
The residual equation \eqref{eq:yit_expression} therefore gives
\(R_i=s_i\nu+\epsilon_i\). The oracle sign rule compares
\[
L(s):=\|R_i-s\nu\|_2^2,\qquad s\in\{-1,1\}.
\]
At the correct and incorrect signs, respectively,
\[
L(s_i)=\|\epsilon_i\|_2^2,
\qquad
L(-s_i)=\|\epsilon_i+2s_i\nu\|_2^2.
\]
Their difference is
\[
L(-s_i)-L(s_i)
=
4\|\nu\|_2^2+4s_i\nu^\top\epsilon_i.
\]
Regardless of the fixed tie rule, selecting the incorrect sign
implies that its loss is no larger than the correct sign's loss.
Hence
\begin{equation}\label{eq:sign_error_event}
\{\hat s_i^{*}(a)\ne s_i\}
\subseteq
\{s_i\nu^\top\epsilon_i\le-\|\nu\|_2^2\}.
\end{equation}
If \(\nu=0\), the two losses coincide identically and the claimed
bound equals one. Suppose now that \(\|\nu\|_2>0\).
For every \(\lambda\ge0\), the sub-Gaussian assumption implies
\[
\mathbb E\!\left[
\exp\{-\lambda s_i\nu^\top\epsilon_i\}\mid X_i,A_i=a
\right]
\le
\exp\!\left\{
\frac{\lambda^2\sigma^2\|\nu\|_2^2}{2}
\right\}.
\]
For any \(\lambda>0\), Markov's inequality and the
sub-Gaussian bound give
\[
\begin{aligned}
&\mathbb P\!\left\{
s_i\nu^\top\epsilon_i\le-\|\nu\|_2^2
\mid X_i,A_i=a
\right\}\\
&\quad\le
\exp\{-\lambda\|\nu\|_2^2\}
\mathbb E\!\left[
\exp\{-\lambda s_i\nu^\top\epsilon_i\}\mid X_i,A_i=a
\right]\\
&\quad\le
\exp\!\left\{
-\lambda\|\nu\|_2^2
+\frac{\lambda^2\sigma^2\|\nu\|_2^2}{2}
\right\}.
\end{aligned}
\]
The exponent is minimized at \(\lambda=1/\sigma^2\). Combining
this choice with \eqref{eq:sign_error_event} proves
\eqref{eq:sign_margin_bound}.

Suppose that, conditional on \(X_i,A_i=a\),
\(\epsilon_i\sim N(0,\sigma^2I_T)\) and \(\|\nu\|_2>0\). Then
\[
\frac{s_i\nu^\top\epsilon_i}{\sigma\|\nu\|_2}
\sim N(0,1).
\]
The loss difference has a continuous distribution, so the probability of a tie is zero. Hence the event on
the right-hand side of \eqref{eq:sign_error_event} has the same probability as the sign error event, regardless
of the tie rule. Standardizing by \(\sigma\|\nu\|_2\) gives
\[
\mathbb P\!\left\{
\hat s_i^{*}(a)\ne s_i\mid X_i,A_i=a
\right\}
=
\Phi\!\left(-\frac{\|\nu\|_2}{\sigma}\right),
\]
which is \eqref{eq:sign_gaussian_error}. If \(\nu=0\), the loss provides no
distinction between the signs.

Suppose \(\pi_c(X_i)=1/(T+1)\) for \(c\in[T+1]\) and
\(\tau_{\pm,l}(X_i)=\theta(X_i)\rho^l\), with
\(\theta(X_i)>0\) and \(0<\rho<1\). Using the definition of \(\Pi(X_i,a)\) given after
\eqref{eq:yit_expression}, we can write the coordinate of \(\nu\)
at the treatment start time \(a\le T\) as
\[
\begin{aligned}
\nu_a
&=
\sum_{l=0}^{a-1}
\left\{\mathbf 1_{\{a=a-l\}}-\pi_{a-l}(X_i)\right\}
\tau_{\pm,l}(X_i)\\
&=
\theta(X_i)\sum_{l=0}^{a-1}
\left\{\mathbf 1_{\{l=0\}}-\frac{1}{T+1}\right\}\rho^l\\
&=
\theta(X_i)\left\{
1-\frac{1}{T+1}\sum_{l=0}^{a-1}\rho^l
\right\}.
\end{aligned}
\]
Since \(\sum_{l=0}^{a-1}\rho^l\le a\) and \(a<T+1\),
\[
1-\frac{1}{T+1}\sum_{l=0}^{a-1}\rho^l
\ge 1-\frac{a}{T+1}>0.
\]
Consequently,
\[
\|\nu\|_2^2
\ge \nu_a^2
\ge
\theta(X_i)^2\left(1-\frac{a}{T+1}\right)^2.
\]
Substituting this lower bound into \eqref{eq:sign_margin_bound}
proves \eqref{eq:sign_geometric_bound}.
\end{proof}

\subsection{Proof of \Cref{thm:opt2_nls_consistency}}

\begin{proof}
Fix \(x\in\mathcal X_0\), and write
\[
c_0:=c(x),
\qquad
\hat c:=\hat c(x),
\qquad
\tau_0:=\tau(x).
\]
Let
\[
L_x^0(\tau):=\ell^{\,\mathrm{nls}}(\tau;c,x,0),
\qquad
\hat L_x(\tau):=\ell^{\,\mathrm{nls}}(\tau;\hat c,x,\rho_N),
\]
where \(\rho_N=o(1)\). The minimization is over \(\Theta=\{\tau\in\mathbb R^T:\|\tau\|_2\le M\}\). Set \(r_N:=\|\hat c-c_0\|_2\), so \(r_N\xrightarrow{\;p\;}0\).

By Assumption~\ref{assump:opt2_reg},
\(C_H:=\max_{(t,s)\in\mathcal P}\|H_{t,s}(x)\|_{\mathrm{op}}<\infty\).
For any \(\tau\in\Theta\),
\[
|(F_x(\tau))_{t,s}|
=|\tau^\top H_{t,s}(x)\tau|
\le C_HM^2,
\]
so \(\sup_{\tau\in\Theta}\|F_x(\tau)\|_2<\infty\). For any \(\tau\in\Theta\),
\begin{align*}
\hat L_x(\tau)-L_x^0(\tau)
&=
\|\hat c-F_x(\tau)\|_2^2
-
\|c_0-F_x(\tau)\|_2^2
+
\rho_N\|\tau\|_2^2\\
&=
2(c_0-F_x(\tau))^\top(\hat c-c_0)
+
\|\hat c-c_0\|_2^2
+
\rho_N\|\tau\|_2^2.
\end{align*}
Thus,
\[
\sup_{\tau\in\Theta}|\hat L_x(\tau)-L_x^0(\tau)|
\le
2r_N\sup_{\tau\in\Theta}\|c_0-F_x(\tau)\|_2+r_N^2+\rho_NM^2
\xrightarrow{\;p\;}0.
\]

Fix \(\varepsilon>0\), and define
\[
\Theta_\varepsilon
:=
\left\{\tau\in\Theta:
\min\{\|\tau-\tau_0\|_2,\|\tau+\tau_0\|_2\}\ge\varepsilon
\right\}.
\]
If \(\Theta_\varepsilon\) is empty, every feasible estimate already has sign-invariant distance less than \(\varepsilon\) from \(\tau_0\). Otherwise, \(\Theta_\varepsilon\) is compact, and the only minimizers of the continuous function \(L_x^0\) are \(\tau_0\) and \(-\tau_0\). Hence
\[
\eta_\varepsilon
:=
\inf_{\tau\in\Theta_\varepsilon}L_x^0(\tau)-L_x^0(\tau_0)
>
0.
\]
On the event
\(\sup_{\tau\in\Theta}|\hat L_x(\tau)-L_x^0(\tau)|\le\eta_\varepsilon/3\),
which has probability tending to one, every \(\tau\in\Theta_\varepsilon\) has strictly larger sample objective than \(\tau_0\). Therefore no global minimizer \(\hat\tau^{\,\mathrm{nls}}(x)\) lies in \(\Theta_\varepsilon\). Since \(\varepsilon>0\) is arbitrary,
\[
d_{\pm}(\hat \tau_{\pm}^{\,\mathrm{nls}}(x),\tau(x))
\xrightarrow{\;p\;}0.
\]
\end{proof}

\makeatletter
\@ifundefined{lemma}{\newtheorem{lemma}{Lemma}}{}
\makeatother

\subsection{Proof of Theorem~\ref{thm:oracle_power}}
\label{app:cate_gaussian_power}
\label{app:combined_power_comparison}

We retain the balanced design and Gaussian outcome assumptions of the
theorem. The CATEs have the expression
\[
\tau_{l,N}(x)=\{\delta_N+\theta_Nu(x)\}\rho^l,
\qquad l=0,\ldots,T-1,
\]
where \(T\ge3\), \(0<\rho\le1/2\), \(\theta_N>0\), and \(u\)
is bounded, with \(\mathbb E[u(X)]=0\) and \(\mathbb E[u(X)^2]=1\).
The rate conditions are
\[
\sqrt N\,|\delta_N|/\sigma\to0,\qquad
\theta_N/\sigma\to0,\qquad
h_N:=\sqrt N\,\theta_N^2/\sigma^2\to h\in[0,\infty).
\]
The proof uses four main lemmas presented below.
Lemma~\ref{lem:balanced_randomization}
gives a normal approximation for balanced randomization.
Lemma~\ref{lem:signed_oracle_CA_power} establishes the
signed-oracle and CA power formulas.
Lemma~\ref{lem:sign_selected_power} gives the
sign-selected power formula, and
Lemma~\ref{lem:geometric_response} proves that
the constant \(B_g\) is positive under the stated conditions.

\begin{proof}[Proof of Theorem~\ref{thm:oracle_power}]
\Cref{lem:signed_oracle_CA_power,lem:sign_selected_power} give
\[
\begin{aligned}
\beta_{\mathrm{LR},N}^{*}
&=1-\Phi(z_{1-\alpha}-\sqrt{A_gN}\,\theta_N/\sigma)+o(1),\\
\beta_{\mathrm{LR},N}^{\pm}
&=1-\Phi(z_{1-\alpha}-h_NB_g)+o(1),\\
\beta_{\mathrm{CA},N}^{*}&=\alpha+o(1).
\end{aligned}
\]
Here \(A_g>0\), and Lemma~\ref{lem:geometric_response} proves
\(B_g>0\). If \(h>0\), then
\(N\theta_N^2/\sigma^2=\sqrt N\,h_N\to\infty\).
The signed-oracle power tends to one, while the
sign-selected power tends to
\(1-\Phi(z_{1-\alpha}-hB_g)>\alpha\). The CA power tends to
\(\alpha\). If \(h=0\), the sign-selected power also tends to
\(\alpha\); this includes heterogeneous effects smaller than the
\(N^{-1/4}\) scale.
\end{proof}

\paragraph{Notation.}
Write \(b=(1,\rho,\ldots,\rho^{T-1})^\top\) and
\(d_c=\Pi(c)b\) for \(c\in[T+1]\). Thus
\[
(d_c)_r=\mathbf1_{\{c\le r\}}b_{r-c}
-\frac1{T+1}\sum_{l=0}^{r-1}b_l,
\qquad r\in[T],
\]
where the first term is zero when \(c>r\). In particular,
\(\sum_{c=1}^{T+1}d_c=0\). With \(Z_i=\epsilon_i/\sigma\),
the true residuals satisfy
\[
R_i=\{\delta_N+\theta_Nu(X_i)\}d_{A_i} + \sigma Z_i.
\]
The vectors \(Z_i\) are independent standard normal vectors,
independent of the covariates and assignments. Conditional on the assignment vector \(A_{[N]}\), the pairs
\((X_i,Z_i)\) are independent and identically distributed,
and every cohort contains \(N/(T+1)\) units.
As defined in the statistic \eqref{equ:global_statistic},
the aggregation weights are 
\[
w_l=\frac{1}{T(T-l)},
\qquad l=0,\ldots,T-1.
\]
Since \(g(x)=u(x)b\) and
\(\mathbb E[u(X)^2]=1\),
the constant in the theorem is
\[
A_g=\frac1{T+1}\sum_{c=1}^{T+1}\|d_c\|_2^2>0.
\]
Its positivity follows from \((d_{T+1})_1=-1/(T+1)\).

\subsubsection{Approximation of the randomization distribution}
\label{app:power_permutation_approximation}

\begin{lemma}\label{lem:balanced_randomization}
Let \(x_i(c)\) be a possibly random real score for unit
\(i\in[N]\) and candidate cohort \(c\in[T+1]\).
Conditional on the score table, draw \(\widetilde A\) uniformly
over assignments with \(N/(T+1)\) units in every cohort.
Suppose that the randomization variance per unit,
\(N^{-1}\widetilde{\operatorname{Var}}
\{\sum_i x_i(\widetilde A_i)\}\), is bounded away from zero
in probability, and that
\[
\frac1{N(T+1)}\sum_{i,c}|x_i(c)|^3=O_p(1).
\]
Then, conditional on the score table,
\[
\sup_z\left|
\widetilde{\mathbb P}\!\left\{
\frac{\sum_i x_i(\widetilde A_i)
-\widetilde{\mathbb E}\{\sum_i x_i(\widetilde A_i)\}}
{\sqrt{\widetilde{\operatorname{Var}}\{\sum_i x_i(\widetilde A_i)\}}}
\le z\right\}-\Phi(z)\right|\xrightarrow{\;p\;}0.
\]
\end{lemma}

\begin{proof}
Within this proof, subtract both sets of means:
\[
x_i^\circ(c)=x_i(c)-\frac1{T+1}\sum_a x_i(a)
-\frac1N\sum_j x_j(c)+\frac1{N(T+1)}\sum_{j,a}x_j(a).
\]
The rows and columns of \(x^\circ\) both sum to zero.
Subtracting the row means removes the randomization mean
\((T+1)^{-1}\sum_{i,c}x_i(c)\). Subtracting the column means does not
change the resulting sum, because a balanced assignment uses
each column exactly \(N/(T+1)\) times. Therefore
\[
\sum_i x_i(\widetilde A_i)-\frac1{T+1}\sum_{i,c}x_i(c)
=\sum_i x_i^\circ(\widetilde A_i).
\]
For distinct units \(i\ne j\), balanced assignment gives
\[
\widetilde{\mathbb P}(\widetilde A_i=c)=\frac1{T+1},
\qquad
\widetilde{\mathbb P}(\widetilde A_i=c,\widetilde A_j=a)
=\frac{N/(T+1)-\mathbf1_{\{a=c\}}}{(T+1)(N-1)}.
\]
The row sums are zero, so the constant part of the joint
probability contributes zero to each cross term. The column sums
are also zero, giving
\[
\begin{aligned}
\sum_{i\ne j}\widetilde{\mathbb E}
\{x_i^\circ(\widetilde A_i)x_j^\circ(\widetilde A_j)\}
&=-\frac1{(T+1)(N-1)}\sum_{i\ne j,c}x_i^\circ(c)x_j^\circ(c)\\
&=\frac1{(T+1)(N-1)}\sum_{i,c}\{x_i^\circ(c)\}^2.
\end{aligned}
\]
The diagonal terms sum to
\((T+1)^{-1}\sum_{i,c}\{x_i^\circ(c)\}^2\).
Adding these terms gives the exact randomization variance:
\begin{equation}\label{eq:proof_randomization_variance}
\mathcal V_N(x):=\widetilde{\operatorname{Var}}
\left\{\sum_i x_i(\widetilde A_i)\right\}
=\frac{N}{N-1}\sum_i\frac1{T+1}
\sum_c\{x_i^\circ(c)\}^2.
\end{equation}
Represent a balanced assignment by a uniform permutation of
\(N\) slots, with \(N/(T+1)\) slots of each cohort type.
For \(\mathcal V_N(x)>0\), the bound of
\citet{chen2015combinatorial}, equation~(1.5), gives
\begin{equation}\label{eq:proof_randomization_normal_bound}
\sup_z\left|
\widetilde{\mathbb P}\!\left\{
\frac{\sum_i x_i^\circ(\widetilde A_i)}{\sqrt{\mathcal V_N(x)}}
\le z\right\}-\Phi(z)\right|
\le\frac{451}{(T+1)\mathcal V_N(x)^{3/2}}
\sum_{i,c}|x_i^\circ(c)|^3.
\end{equation}
Jensen's inequality bounds the average third absolute moment of
each row mean, column mean, and overall mean by the corresponding
average for the original table. The inequality
\(|a+b+c+d|^3\le16(|a|^3+|b|^3+|c|^3+|d|^3)\)
then gives
\[
\frac1{N(T+1)}\sum_{i,c}|x_i^\circ(c)|^3=O_p(1).
\]
Together with the variance assumption, this makes the bound in
\eqref{eq:proof_randomization_normal_bound} \(O_p(N^{-1/2})\).
It therefore tends to zero.
\end{proof}

\subsubsection{Signed-oracle and CA power}
\label{app:oracle_power}

\begin{lemma}
\label{lem:signed_oracle_CA_power}
Under the assumptions of Theorem~\ref{thm:oracle_power}, for every fixed
\(\alpha\in(0,1)\),
\[
\beta_{\mathrm{LR},N}^{*}
=1-\Phi\{z_{1-\alpha}-\sqrt{A_gN}\,\theta_N/\sigma\}+o(1),
\qquad
\beta_{\mathrm{CA},N}^{*}=\alpha+o(1).
\]
\end{lemma}

\begin{proof}
The proof compares the observed and randomization distributions for
each test. We use the following normalization only within this proof:
\[
t_N=\frac{\sqrt{\delta_N^2+\theta_N^2}}{\sigma},
\qquad
U_i=\frac{\delta_N+\theta_Nu(X_i)}{\sigma t_N}.
\]
Then \(t_N>0\), \(t_N\to0\), \(\mathbb E[U_i^2]=1\), and
\(|U_i|\le1+\|u\|_\infty\). In particular, all moments of
\(U_i\) are uniformly bounded. The residual representation becomes
\(R_i/\sigma=Z_i+t_NU_id_{A_i}\). This normalization allows
\(\delta_N\) and \(\theta_N\) to approach zero at different rates.

\paragraph{1. The signed-oracle likelihood test.}
For a candidate assignment \(a\), the outcome error under the
specified alternative is \(R_i-\sigma t_NU_id_{a_i}\). Expanding
its squared norm in the Gaussian log likelihood and removing terms
independent of \(a\) gives
\[
\ell_N(a)=t_N\sum_iU_i d_{a_i}^{\top}(R_i/\sigma)
-\frac{t_N^2}{2}\sum_iU_i^2\|d_{a_i}\|_2^2.
\]
Since \(\sum_cd_c=0\), its randomization mean is
\(-t_N^2A_g\sum_iU_i^2/2\). The observed likelihood score minus
this mean is therefore
\begin{equation}\label{eq:proof_signed_observed}
\begin{split}
\bar T_N^*:=\ell_N(A)-\widetilde{\mathbb E}[\ell_N(\widetilde A)] =t_N\sum_iU_id_{A_i}^{\top}Z_i
+\frac{t_N^2}{2}\sum_iU_i^2\{\|d_{A_i}\|_2^2+A_g\}.
\end{split}
\end{equation}

\emph{Observed distribution.}
Conditional on \(X_{[N]},A_{[N]}\), the first term in
\eqref{eq:proof_signed_observed} is Gaussian. Its conditional variance
is \(t_N^2\sum_iU_i^2\|d_{A_i}\|_2^2\). Independence across
units, bounded fourth moments, and the fixed cohort counts imply
\[
\mathbb E\!\left[\sum_iU_i^2\|d_{A_i}\|_2^2    \, \bigg|  \,  A_{[N]}\right]
=NA_g,
\qquad
\operatorname{Var}\!\left[\sum_iU_i^2\|d_{A_i}\|_2^2 \, \bigg|  \, A_{[N]}\right]
=O(N).
\]
Hence \(N^{-1}\sum_iU_i^2\|d_{A_i}\|_2^2\to A_g\) in
probability. Averaging the conditional Gaussian distribution gives
\[
\frac{\sum_iU_id_{A_i}^{\top}Z_i}{\sqrt{NA_g}}
\xrightarrow{\;d\;} N(0,1).
\]
The second term in \eqref{eq:proof_signed_observed} has conditional
mean \(Nt_N^2A_g\) and variance \(O(Nt_N^4)\). Thus
\(\mathbb E(\bar T_N^*\mid A_{[N]})=Nt_N^2A_g\), and the centered
second term divided by \(t_N\sqrt{NA_g}\) is
\(O_p(t_N)=o_p(1)\). We obtain
\begin{equation}\label{eq:proof_signed_observed_clt}
\frac{\bar T_N^*}{t_N\sqrt{NA_g}}-\sqrt{NA_g}\,t_N
\xrightarrow{\;d\;} N(0,1).
\end{equation}
\emph{Randomization distribution.}
Hold the observed data fixed, subtract the candidate mean for each
unit, and divide the likelihood score by \(t_N\). The resulting
candidate table is
\[
x_i(c)=U_i d_c^{\top}Z_i+t_NU_i^2
\left\{d_c^{\top}d_{A_i}
-\frac12\bigl(\|d_c\|_2^2-A_g\bigr)\right\}.
\]
Its rows sum to zero. The second term is uniformly \(O(t_N)\),
so its average square tends to zero. The first term has mean zero,
second moment \(\|d_c\|_2^2\), and uniformly bounded fourth
moments. Chebyshev's inequality therefore gives
\[
\frac1{N(T+1)}\sum_{i,c}x_i(c)^2=A_g+o_p(1),
\qquad
\frac1N\sum_i x_i(c)=o_p(1).
\]
Subtracting the column means changes the average square by
\[
(T+1)^{-1}\sum_c\{N^{-1}\sum_i x_i(c)\}^2=o_p(1).
\]
Moreover, the average third absolute moment of \(x_i(c)\) is
bounded in probability, and Jensen's inequality gives the same bound
after centering. Lemma~\ref{lem:balanced_randomization}, together
with \eqref{eq:proof_randomization_variance}, now yields
\[
\sup_z\left|
\widetilde{\mathbb P}\!\left\{
\frac{\ell_N(\widetilde A)-\widetilde{\mathbb E}[\ell_N(\widetilde A)]}
{t_N\sqrt{NA_g}}\le z\right\}-\Phi(z)\right|
\xrightarrow{\;p\;}0.
\]
The upper-tail critical value on this scale is consequently
\(z_{1-\alpha}+o_p(1)\).

To combine the two limits, note that the distribution functions in
\eqref{eq:proof_signed_observed_clt} converge uniformly to \(\Phi\).
The bounded normal density implies that the probability of any
interval of length \(2\varepsilon\) is at most
\(2\varepsilon/\sqrt{2\pi}+o(1)\), uniformly over its location.
On the event that the random critical value differs from
\(z_{1-\alpha}\) by at most \(\varepsilon\), replacing it by
\(z_{1-\alpha}\) changes the rejection event only within such an
interval. First let \(N\to\infty\), then let
\(\varepsilon\downarrow0\). This proves
\[
\beta_{\mathrm{LR},N}^{*}
=1-\Phi\{z_{1-\alpha}-\sqrt{NA_g}\,t_N\}+o(1),
\]
including when \(\sqrt{NA_g}\,t_N\) diverges. Finally,
\[
0\le\sqrt{NA_g}\,t_N-\sqrt{NA_g}\,\theta_N/\sigma
\le\sqrt{NA_g}\,|\delta_N|/\sigma\to0.
\]
Replacing the shift by \(\sqrt{NA_g}\,\theta_N/\sigma\)
therefore gives
\[
\beta_{\mathrm{LR},N}^{*}
=1-\Phi\{z_{1-\alpha}-\sqrt{A_gN}\,\theta_N/\sigma\}+o(1).
\]

\paragraph{2. The oracle CA test.}
For this part, collect the weights of the CA score at each outcome
time into a vector \(B_c\in\mathbb R^T\):
\[
B_{c,r}=\begin{cases}
(T+1)w_{r-c},&c\le r,\\[2pt]
-\dfrac{T+1}{T+1-r}\displaystyle\sum_{l=0}^{r-1}w_l,&c>r.
\end{cases}
\]
The first case is the treated contribution of cohort \(c\).
The second combines its control contributions while it remains
untreated. Thus \(N\) times the candidate CA statistic is
\(\sum_iB_{a_i}^{\top}R_i\). The weights cancel at every
outcome time, giving \(\sum_cB_c=0\), so the randomization
mean is zero. Define its noise variance constant by
\begin{equation}\label{eq:proof_CA_variance}
D_T=\frac1{T+1}\sum_c\|B_c\|_2^2
=(T+1)\sum_{r=1}^T\left\{
\sum_{l=0}^{r-1}w_l^2+
\frac{\left(\sum_{l=0}^{r-1}w_l\right)^2}{T+1-r}\right\}>0.
\end{equation}

In the average \((T+1)^{-1}\sum_cB_c^{\top}d_c\), the common
mean trajectory subtracted in \(d_c\) cancels because
\(\sum_cB_c=0\). Only treated coordinates remain. Hence
\[
\frac1{T+1}\sum_cB_c^{\top}d_c
=\sum_{r=1}^T\sum_{l=0}^{r-1}w_lb_l
=\sum_{l=0}^{T-1}(T-l)w_lb_l
=\frac1T\sum_{l=0}^{T-1}b_l.
\]
The final equality uses \((T-l)w_l=1/T\).

\emph{Observed distribution.}
Substituting the residual representation gives
\[
\sum_iB_{A_i}^{\top}R_i
=\sigma\sum_iB_{A_i}^{\top}Z_i
+\frac{N\delta_N}{T}\sum_{l=0}^{T-1}b_l
+\theta_N\sum_i u(X_i)B_{A_i}^{\top}d_{A_i}.
\]
Conditional on assignment, the first term is exactly Gaussian with
variance \(N\sigma^2D_T\). The final term has mean zero and
variance \(O(N\theta_N^2)\), by independence of the covariates
and \(\mathbb E[u(X)]=0\). After division by
\(\sigma\sqrt{ND_T}\), the second term tends to zero because
\(\sqrt N\,|\delta_N|/\sigma\to0\), and the final term is
\(O_p(\theta_N/\sigma)=o_p(1)\). Therefore
\[
\frac{\sum_iB_{A_i}^{\top}R_i}{\sigma\sqrt{ND_T}}
\xrightarrow{\;d\;} N(0,1).
\]

\emph{Randomization distribution.}
The candidate CA table divided by \(\sigma\) is
\[
x_i(c)=B_c^{\top}Z_i+t_NU_iB_c^{\top}d_{A_i}.
\]
Its rows sum to zero. The second term is uniformly \(O(t_N)\),
whereas the first has mean zero, second moment \(\|B_c\|_2^2\),
and bounded fourth moments. The same moment calculation as above
gives
\[
\frac1{N(T+1)}\sum_{i,c}x_i(c)^2=D_T+o_p(1),
\qquad
\frac1N\sum_i x_i(c)=o_p(1).
\]
The average third absolute moment is bounded in probability.
Subtracting the column means again changes the average square by
\(o_p(1)\), so Lemma~\ref{lem:balanced_randomization} gives
\[
\sup_z\left|
\widetilde{\mathbb P}\!\left\{
\frac{\sum_iB_{\widetilde A_i}^{\top}R_i}
{\sigma\sqrt{ND_T}}\le z\right\}-\Phi(z)\right|
\xrightarrow{\;p\;}0.
\]
The observed and reference limits are both standard normal. Using
the critical-value argument from the first part proves
$\beta_{\mathrm{CA},N}^{*}=\alpha+o(1).$
\end{proof}

\subsubsection{Power with sign selection}
\label{app:sign_selected_power}

\begin{lemma}
\label{lem:sign_selected_power}
Under the assumptions of Theorem~\ref{thm:oracle_power}, there is a
finite constant \(B_g\), depending only on \(T\), \(\rho\), and the
distribution of \(u(X)\), such that
\begin{equation}\label{eq:proof_sign_power_final}
\beta_{\mathrm{LR},N}^{\pm}
=1-\Phi(z_{1-\alpha}-h_NB_g)+o(1).
\end{equation}
The constant \(B_g\) is defined in \eqref{eq:proof_AB_constants} below.
\end{lemma}

\begin{proof}[Proof of Lemma~\ref{lem:sign_selected_power}]
The proof uses three lemmas. Lemma~\ref{lem:sign_score_representation}
writes the statistic after sign selection and centering.
Lemma~\ref{lem:sign_score_moments} establishes its quadratic mean
response and a positive variance bound.
Lemma~\ref{lem:sign_normal_limits} gives normal approximations for
the observed statistic and its randomization distribution.
Let \(\bar T_N\) be the centered and rescaled observed statistic in
Lemma~\ref{lem:sign_score_representation}, and let \(V_N^{\pm}\)
be its randomization variance. Centering and positive rescaling
preserve the randomization \(p\)-value.
Lemma~\ref{lem:sign_normal_limits} gives
\[
\frac{\bar T_N}{\sqrt{V_N^{\pm}}}-h_NB_g\xrightarrow{\;d\;} N(0,1)
\]
and a uniform standard normal approximation for the corresponding
randomization distribution. This approximation and continuity of
\(\Phi\) imply that the ordinary upper-tail randomization
\(p\)-value equals
\[
1-\Phi(\bar T_N/\sqrt{V_N^{\pm}})+o_p(1).
\]
Because \(h_NB_g\to hB_g\), continuity at the normal critical
value gives \eqref{eq:proof_sign_power_final}. The score
representation includes sign selection under every candidate
assignment; consistent recovery of individual signs is not required.
\end{proof}

\paragraph{Normalization for the supporting lemmas.}
Put
\begin{equation}\label{eq:sign_power_normalization}
s_N=(\delta_N^2+\theta_N^2)^{1/2},\qquad
t_N=s_N/\sigma,\qquad
U_i=\frac{\delta_N+\theta_Nu(X_i)}{s_N}.
\end{equation}
Then \(s_N>0\), \(\mathbb E[U_i^2]=1\), and
\(|U_i|\le1+\|u\|_\infty\). Also,
\begin{equation}\label{eq:sign_power_normalization_rate}
t_N\to0,\qquad
\sqrt N\,t_N^2=h_N+\sqrt N\,\delta_N^2/\sigma^2=h_N+o(1),
\end{equation}
since
\(\sqrt N\,\delta_N^2/\sigma^2
=(\sqrt N\,\delta_N/\sigma)^2/\sqrt N\to0\).
This normalization covers \(h=0\), when
\(\delta_N/\theta_N\) need not tend to zero.
These quantities are used only in the supporting lemmas.

For a candidate cohort \(c\in[T+1]\), define
\[
(v_c)_r=
\begin{cases}
w_{r-c}b_{r-c},&c\le r,\\
0,&c>r,
\end{cases}
\qquad
k_c=\sum_{r=c}^T w_{r-c}
\left\{\frac{2b_{r-c}}{T+1}\sum_{l=0}^{r-1}b_l-b_{r-c}^2\right\}.
\]
Empty sums are zero. For \(z\in\mathbb R^T\), let
\begin{equation}\label{eq:proof_sign_reduced_score}
\zeta_c(z)=\operatorname{sgn}(d_c^\top z)v_c^\top z,
\qquad
Q_c(z)=\zeta_c(z)-\frac1{T+1}\sum_{a=1}^{T+1}\zeta_a(z).
\end{equation}
Let \(F_i(c)\) be unit \(i\)'s contribution to \(N\) times the
sign-selected LR statistic under candidate cohort \(c\), divided
by \(2\sigma s_N\) and centered over candidate cohorts. Explicitly,
\[
\begin{aligned}
F_i(c)
:={}&
\frac{1}{2\sigma s_N}
\left.
\sum_{l=0}^{T-1}w_l
\sum_{t=1}^{T-l}\psi_{i,t,l}^{\mathrm{LR}}(O_i;\eta)
\right|_{A_i=c}\\
&-
\frac{1}{2\sigma s_N(T+1)}
\sum_{a=1}^{T+1}
\left.
\sum_{l=0}^{T-1}w_l
\sum_{t=1}^{T-l}\psi_{i,t,l}^{\mathrm{LR}}(O_i;\eta)
\right|_{A_i=a}.
\end{aligned}
\]
Here the covariates and outcomes are held fixed, while the sign
and corresponding control mean are recomputed under each candidate
cohort. By construction, \(\sum_{c=1}^{T+1}F_i(c)=0\).
The centered and rescaled statistic under assignment \(a\) is
\(\sum_iF_i(a_i)\), and its observed value is
\(\bar T_N=\sum_iF_i(A_i)\).

\begin{lemma}[Score representation]
\label{lem:sign_score_representation}
For every \(i\in[N]\) and \(c\in[T+1]\),
\begin{equation}\label{eq:proof_sign_centered_table}
F_i(c)=|U_i|Q_c(Z_i+t_NU_id_{A_i})
+\frac{t_N}{2}U_i^2
\left(k_c-\frac1{T+1}\sum_{a=1}^{T+1}k_a\right).
\end{equation}
\end{lemma}

\begin{proof}
The unsigned CATE vector for unit \(i\) is \(s_N|U_i|b\).
For candidate cohort \(c\), the loss in \eqref{equ:s_x} is
\[
\|R_i-vs_N|U_i|d_c\|_2^2
=\|R_i\|_2^2+s_N^2U_i^2\|d_c\|_2^2
-2vs_N|U_i|d_c^\top R_i,
\qquad v\in\{-1,1\}.
\]
For \(U_i\ne0\), its minimizing sign is
\(\operatorname{sgn}(d_c^\top R_i)\).
For \(U_i=0\), either sign gives a zero LR contribution.
Every \(d_c\) is nonzero: for \(c\le T\), its \(c\)-th
coordinate is at least \(1-2/(T+1)>0\), while
\((d_{T+1})_1=-1/(T+1)\). Thus Gaussian sign ties have
probability zero; a fixed rule can be used at a tie.

The candidate control mean at time \(r\) is
\[
\hat\mu_{0,r}(X_i;c)
=
\mu_r(X_i)-\frac1{T+1}\sum_{l=0}^{r-1}\hat\tau_l(X_i;c)
=
\mu_r(X_i)
-\frac{\operatorname{sgn}(d_c^\top R_i)s_N|U_i|}{T+1}
\sum_{l=0}^{r-1}b_l,
\]
where
\(\hat\tau_l(X_i;c)=\operatorname{sgn}(d_c^\top R_i)s_N|U_i|b_l\)
is the signed effect under candidate cohort \(c\).

Although the marginal outcome means and unsigned CATEs are known,
the candidate control mean need not equal the true control mean.
For \(c\le T\), only the terms with \(t=c\) contribute to unit
\(i\)'s LR score. The candidate control mean gives
\[
Y_{i,r}-\hat\mu_{0,r}(X_i;c)
=
R_{i,r}
+\frac{\operatorname{sgn}(d_c^\top R_i)s_N|U_i|}{T+1}
\sum_{j=0}^{r-1}b_j.
\]
Using this expression, the signed effect
\(\hat\tau_{r-c}(X_i;c)
=\operatorname{sgn}(d_c^\top R_i)s_N|U_i|b_{r-c}\),
and the change of index \(r=c+l\), we obtain
\[
\begin{aligned}
&  \left.
\sum_{l=0}^{T-1}w_l
\sum_{t=1}^{T-l}
\psi_{i,t,l}^{\mathrm{LR}}(O_i;\eta)
\right|_{A_i=c}  =
\sum_{r=c}^{T}w_{r-c}
\left\{
2\hat\tau_{r-c}(X_i;c)
\bigl[Y_{i,r}-\hat\mu_{0,r}(X_i;c)\bigr]
-\hat\tau_{r-c}(X_i;c)^2
\right\}\\
& =
2s_N|U_i|\operatorname{sgn}(d_c^\top R_i)
\sum_{r=c}^{T}w_{r-c}b_{r-c}R_{i,r}+
s_N^2U_i^2
\sum_{r=c}^{T}w_{r-c}
\left\{
\frac{2b_{r-c}}{T+1}\sum_{j=0}^{r-1}b_j
-b_{r-c}^2
\right\}\\
&=
2s_N|U_i|\operatorname{sgn}(d_c^\top R_i)v_c^\top R_i
+s_N^2U_i^2k_c.
\end{aligned}
\]
The expansion uses that the selected sign has square one.
The final equality follows from the definitions of \(v_c\)
and \(k_c\). This is unit \(i\)'s contribution to \(N\) times
the aggregated statistic in \eqref{equ:global_statistic}.
For \(c=T+1\), all LR scores are zero, and the same expression
holds because \(v_{T+1}=0\) and \(k_{T+1}=0\).

Dividing the unit contribution by \(2\sigma s_N\) and
subtracting its average over candidate cohorts gives
\[
\begin{aligned}
F_i(c)
&=
\frac{|U_i|}{\sigma}
\left[
\operatorname{sgn}(d_c^\top R_i)v_c^\top R_i
-\frac1{T+1}\sum_{a=1}^{T+1}
\operatorname{sgn}(d_a^\top R_i)v_a^\top R_i
\right]\\
&\quad+
\frac{s_N}{2\sigma}U_i^2
\left(k_c-\frac1{T+1}\sum_{a=1}^{T+1}k_a\right)\\[4pt]
&=
|U_i|
\left[
\zeta_c(R_i/\sigma)
-\frac1{T+1}\sum_{a=1}^{T+1}\zeta_a(R_i/\sigma)
\right]
+\frac{t_N}{2}U_i^2
\left(k_c-\frac1{T+1}\sum_{a=1}^{T+1}k_a\right)\\[4pt]
&=
|U_i|Q_c(R_i/\sigma)
+\frac{t_N}{2}U_i^2
\left(k_c-\frac1{T+1}\sum_{a=1}^{T+1}k_a\right)\\[4pt]
&=
|U_i|Q_c(Z_i+t_NU_id_{A_i})
+\frac{t_N}{2}U_i^2
\left(k_c-\frac1{T+1}\sum_{a=1}^{T+1}k_a\right).
\end{aligned}
\]
Here we used \(\sigma>0\), \(t_N=s_N/\sigma\), the definitions
of \(\zeta_c\) and \(Q_c\), and
\(R_i/\sigma=Z_i+t_NU_id_{A_i}\).
This proves \eqref{eq:proof_sign_centered_table}.
\end{proof}

Let \(Z\sim N(0,I_T)\), independently of the covariates. Define the
quadratic mean response and the variance per unit by
\begin{equation}\label{eq:proof_geometric_constants}
\gamma_T(b)=\frac1{2(T+1)}\sum_{c=1}^{T+1}
\mathbb E\!\left[
Q_c(Z)\{(d_c^\top Z)^2-\|d_c\|_2^2\}\right] \quad \text{and }
\end{equation}
\begin{equation}\label{eq:proof_sign_noise}
\begin{split}
\omega_N^2
&=\frac1{T+1}\sum_{c=1}^{T+1}
\operatorname{Var}\{|U_1|Q_c(Z)\}\\
&=\frac1{T+1}\sum_{c=1}^{T+1}
\left[\mathbb E\{Q_c(Z)^2\}
-(\mathbb E[|U_1|])^2\{\mathbb E [Q_c(Z)]\}^2\right].
\end{split}
\end{equation}

\begin{lemma}[Score moments]\label{lem:sign_score_moments}
Under the assumptions of Lemma~\ref{lem:sign_selected_power},
\begin{equation}\label{eq:proof_sign_observed_mean}
\mathbb E(\bar T_N\mid A_{[N]})
=Nt_N^2\mathbb E[|U_1|^3]\,\gamma_T(b)+O(Nt_N^4),
\end{equation}
and
\begin{equation}\label{eq:proof_sign_observed_variance}
\frac1N\operatorname{Var}(\bar T_N\mid A_{[N]})=\omega_N^2+o(1).
\end{equation}
These bounds hold uniformly over balanced assignments. Moreover,
\begin{equation}\label{eq:proof_sign_variance_bounds}
0<\frac1{T+1}\sum_{c=1}^{T+1}\operatorname{Var}\{Q_c(Z)\}
\le\omega_N^2
\le\frac1{T+1}\sum_{c=1}^{T+1}\mathbb E\{Q_c(Z)^2\}<\infty.
\end{equation}
\end{lemma}

\begin{proof}
We calculate the mean by expanding the Gaussian density. We then
compare the variance with the variance obtained by removing the
small shift from the residuals.

By \eqref{eq:proof_sign_reduced_score} and the triangle
and Cauchy--Schwarz inequalities,
\[
|Q_c(z)|
\le |v_c^\top z|
+\frac1{T+1}\sum_{a=1}^{T+1}|v_a^\top z|
\le C\|z\|_2,
\]
where \(C=2\max_{a\in[T+1]}\|v_a\|_2\) depends only on
\(T\) and \(\rho\). We use this bound throughout the
following proofs.

\paragraph{The mean.}
For real \(v\), the Gaussian density identity gives
\[
\mathbb E [Q_c(Z+vd_c)]
=\mathbb E\!\left[Q_c(Z)
\exp\{vd_c^\top Z-v^2\|d_c\|_2^2/2\}\right].
\]
Fix \(v_0>0\) and consider \(|v|\le v_0\). In the preceding
density identity, differentiation with respect to \(v\) acts
only on the exponential factor. Its \(k\)-th derivative is
the same exponential multiplied by a polynomial in
\(d_c^\top z-v\|d_c\|_2^2\) of degree at most \(k\).
By the Cauchy--Schwarz inequality,
\[
|d_c^\top z-v\|d_c\|_2^2|
\le \|d_c\|_2\|z\|_2+v_0\|d_c\|_2^2
\le C(1+\|z\|_2).
\]
Thus, for \(k=0,\ldots,4\), the polynomial factor is bounded
in absolute value by \(C(1+\|z\|_2^4)\). Moreover,
\[
\exp\left\{vd_c^\top z-\frac{v^2\|d_c\|_2^2}{2}\right\}
\le \exp\{v_0\|d_c\|_2\|z\|_2\}
\le \exp(C\|z\|_2).
\]
Combining these bounds with \(|Q_c(z)|\le C\|z\|_2\) gives
\[
\begin{aligned}
\left|
\frac{\partial^k}{\partial v^k}
\left[
Q_c(z)\exp\left\{
vd_c^\top z-\frac{v^2\|d_c\|_2^2}{2}
\right\}
\right]
\right|
&=
|Q_c(z)|
\left|
\frac{\partial^k}{\partial v^k}
\exp\left\{
vd_c^\top z-\frac{v^2\|d_c\|_2^2}{2}
\right\}
\right|\\
&\le C\|z\|_2(1+\|z\|_2^4)\exp(C\|z\|_2)\\
&\le C(1+\|z\|_2^5)\exp(C\|z\|_2),
\end{aligned}
\]
where \(C\) may depend on \(T\), \(\rho\), and \(v_0\),
but not on \(v\) or \(z\).
The derivative bound above is integrable under the standard
Gaussian distribution, since
\[
\exp\left\{C\|z\|_2-\frac{\|z\|_2^2}{2}\right\}
\le
\exp(C^2)\exp\left\{-\frac{\|z\|_2^2}{4}\right\},
\]
and
\[
\int_{\mathbb R^T}(1+\|z\|_2^5)
\exp\left\{-\frac{\|z\|_2^2}{4}\right\}\,dz<\infty.
\]
We may therefore differentiate under the expectation up to
order four, with
\[
\sup_{|v|\le v_0}
\left|
\frac{d^4}{dv^4}\mathbb E[Q_c(Z+vd_c)]
\right|<\infty.
\]
This ensures an \(O(v^4)\) remainder in the
third-order Taylor expansion about zero.

Replacing \(z\) by \(-z\) reverses both factors in
\(\zeta_c(z)\), so \(Q_c(-z)=Q_c(z)\) outside the
sign-tie sets. These sets have probability zero under
the shifted Gaussian distributions. Since \(Z\) and
\(-Z\) have the same distribution, we have
\[
\mathbb E[Q_c(Z-vd_c)]
=\mathbb E[Q_c(-Z-vd_c)]=\mathbb E[Q_c(Z+vd_c)].
\]
Thus the expectation is an even function of \(v\), and its
first and third derivatives at zero vanish. Evaluating the
second derivative using the Gaussian density identity and
applying Taylor's formula gives
\begin{equation}\label{eq:proof_sign_mean_expansion}
\begin{aligned}
\mathbb E[Q_c(Z+vd_c)] =\mathbb E[Q_c(Z)]+\frac{v^2}{2}\mathbb E\!\left[
Q_c(Z)\{(d_c^\top Z)^2-\|d_c\|_2^2\}
\right]+O(v^4),
\end{aligned}
\end{equation}
uniformly over the finitely many cohorts.
Conditional on \(A_i=c\), substitute \(v=t_NU_i\) into
\eqref{eq:proof_sign_mean_expansion} and average over \(X_i\).
Using the expression of $F_i(c)$ in \eqref{eq:proof_sign_centered_table} gives
\[
\begin{aligned}
\mathbb E\{F_i(c)\mid A_i=c\}
&=\mathbb E[|U_i|]\,\mathbb E [Q_c(Z)]+\frac{t_N^2\mathbb E[|U_i|^3]}{2}
\mathbb E\!\left[Q_c(Z)\{(d_c^\top Z)^2-\|d_c\|_2^2\}\right]\\
&\quad+\frac{t_N}{2}
\left(k_c-\frac1{T+1}\sum_{a=1}^{T+1}k_a\right)+O(t_N^4).
\end{aligned}
\]
By the normalization in \eqref{eq:sign_power_normalization},
\(U_i\) is uniformly bounded and \(\mathbb E[U_i^2]=1\).
Substituting \(v=t_NU_i\) into
\eqref{eq:proof_sign_mean_expansion} and multiplying by
\(|U_i|\) gives a remainder bounded by \(Ct_N^4|U_i|^5\).
Since \(U_i\) is uniformly bounded, its expectation is
\(O(t_N^4)\) uniformly over cohorts.
Moreover, \(\mathbb E[U_i^2]=1\) gives the third term.
By the definition of \(Q_c\) in
\eqref{eq:proof_sign_reduced_score}, we have
\[
\sum_{c=1}^{T+1}Q_c(z)=0,\qquad
\sum_{c=1}^{T+1}
\left(k_c-\frac1{T+1}\sum_{a=1}^{T+1}k_a\right)=0.
\]
Since every cohort contains \(N/(T+1)\) units, the first
and third terms cancel when summed over units.
Summing the second term over all units and using the
definition of \(\gamma_T(b)\) in
\eqref{eq:proof_geometric_constants} gives the mean
in \eqref{eq:proof_sign_observed_mean}.

\paragraph{The variance.}
Remove the residual shift and write
\(F_i^0(c)=|U_i|Q_c(Z_i)\). Independence of \(U_i\) and
\(Z_i\), together with \(\mathbb E[U_i^2]=1\), gives
\[
\frac1{T+1}\sum_{c=1}^{T+1}\operatorname{Var}\{F_i^0(c)\}
=\omega_N^2.
\]
For almost every \(Z\), all \(Q_c\) are continuous at \(Z\).
Since \(U_i\) is uniformly bounded and \(t_N\to0\),
the shifts \(t_NU_id_j\) tend uniformly to zero.
Moreover, the quadratic term in
\eqref{eq:proof_sign_centered_table} is bounded by \(Ct_N\).
Thus \(F_i(c)-F_i^0(c)\to0\) almost surely conditional
on \(A_i=j\).
By the linear growth bound on \(Q_c\) and the uniform
boundedness of \(U_i\), we have
\[
|F_i(c)-F_i^0(c)|^2\le C(1+\|Z_i\|_2)^2.
\]
Dominated convergence therefore gives
\begin{equation}\label{eq:proof_sign_mean_square}
\max_{c,j}\mathbb E\!\left[
|F_i(c)-F_i^0(c)|^2\mid A_i=j\right]\to0.
\end{equation}
By the Cauchy--Schwarz inequality and the uniformly bounded
second moments,
\[
\max_{c,j}\left|
\operatorname{Var}\{F_i(c)\mid A_i=j\}
-\operatorname{Var}\{F_i^0(c)\}\right|\to0.
\]
Conditional on \(A_{[N]}\), the contributions from different
units are independent. Since every cohort contains
\(N/(T+1)\) units, we have
\[
\begin{aligned}
\frac1N\operatorname{Var}(\bar T_N\mid A_{[N]})
&=\frac1N\sum_{i=1}^N
\operatorname{Var}\{F_i(A_i)\mid A_{[N]}\}\\
&=\frac1{T+1}\sum_{c=1}^{T+1}
\operatorname{Var}\{F_1(c)\mid A_1=c\}\\
&=\frac1{T+1}\sum_{c=1}^{T+1}
\operatorname{Var}\{F_1^0(c)\}+o(1)\\
&=\omega_N^2+o(1).
\end{aligned}
\]
This proves \eqref{eq:proof_sign_observed_variance}
uniformly over balanced assignments.

It remains to establish \eqref{eq:proof_sign_variance_bounds}.
Since \((\mathbb E[|U_i|])^2\le\mathbb E[U_i^2]=1\),
we have
\[
\begin{aligned}
\operatorname{Var}\{Q_c(Z)\}
&=\mathbb E[Q_c(Z)^2]-(\mathbb E[Q_c(Z)])^2\\
&\le \mathbb E[Q_c(Z)^2]
-(\mathbb E[|U_i|])^2(\mathbb E[Q_c(Z)])^2\\
&\le \mathbb E[Q_c(Z)^2].
\end{aligned}
\]
Averaging over \(c\) and using \eqref{eq:proof_sign_noise}
gives
\[
\frac1{T+1}\sum_{c=1}^{T+1}\operatorname{Var}\{Q_c(Z)\}
\le\omega_N^2
\le\frac1{T+1}\sum_{c=1}^{T+1}\mathbb E[Q_c(Z)^2]
<\infty,
\]
where the last inequality follows from the linear growth
bound
\[
\mathbb E[Q_c(Z)^2]\le C^2\mathbb E[\|Z\|_2^2]=C^2T.
\]

For the strict lower bound, since \(v_{T+1}=0\), we have
\(\zeta_{T+1}=0\). By \eqref{eq:proof_sign_reduced_score},
\[
Q_1(Z)-Q_{T+1}(Z)=\zeta_1(Z).
\]
It follows that
\[
\begin{aligned}
\operatorname{Var}\{\zeta_1(Z)\}
&=\operatorname{Var}\{Q_1(Z)-Q_{T+1}(Z)\}\\
&\le 2\operatorname{Var}\{Q_1(Z)\}
+2\operatorname{Var}\{Q_{T+1}(Z)\}\\
&\le 2\sum_{c=1}^{T+1}\operatorname{Var}\{Q_c(Z)\}.
\end{aligned}
\]
Since \(\zeta_1(Z)^2=(v_1^\top Z)^2\) almost surely
and \(v_1\ne0\), \(\zeta_1(Z)\) is not constant almost
surely. Hence
\[
\frac1{T+1}\sum_{c=1}^{T+1}\operatorname{Var}\{Q_c(Z)\}
\ge\frac{\operatorname{Var}\{\zeta_1(Z)\}}{2(T+1)}>0.
\]
This proves \eqref{eq:proof_sign_variance_bounds}.
\end{proof}

\begin{lemma}[Normal approximations]\label{lem:sign_normal_limits}
Under the assumptions of Lemma~\ref{lem:sign_selected_power}, let
\[
V_N^{\pm}:=\widetilde{\operatorname{Var}}
\left\{\sum_iF_i(\widetilde A_i)\right\}
\]
be the conditional randomization variance of the sum of
the scores in \eqref{eq:proof_sign_centered_table}, and let
\begin{equation}\label{eq:proof_AB_constants}
B_g=\frac{\mathbb E[|u(X)|^3]\,\gamma_T(b)}
{\left\{\dfrac1{T+1}\sum_{c=1}^{T+1}
\operatorname{Var}\bigl(|u(X)|Q_c(Z)\bigr)\right\}^{1/2}},
\end{equation}
where \(Z\sim N(0,I_T)\) is independent of \(X\).
Then
\begin{equation}\label{eq:proof_sign_normalized_limit}
\frac{\bar T_N}{\sqrt{V_N^{\pm}}}-h_NB_g
\xrightarrow{\;d\;} N(0,1).
\end{equation}
Moreover, conditional on the score table, the randomization
distribution satisfies
\begin{equation}\label{eq:proof_sign_randomization_limit}
\sup_z\left|
\widetilde{\mathbb P}\!\left\{
\frac{\sum_iF_i(\widetilde A_i)}{\sqrt{V_N^{\pm}}}\le z
\right\}
-\Phi(z)\right|\xrightarrow{\;p\;}0.
\end{equation}
\end{lemma}

\begin{proof}
By the same argument used to establish
\eqref{eq:proof_sign_variance_bounds}, with \(U_i\)
replaced by \(u(X)\), we have
\[
0<\frac1{T+1}\sum_{c=1}^{T+1}
\operatorname{Var}\{Q_c(Z)\}
\le\frac1{T+1}\sum_{c=1}^{T+1}
\operatorname{Var}\{|u(X)|Q_c(Z)\}<\infty.
\]
The numerator in \eqref{eq:proof_AB_constants} is also
finite because \(u\) is bounded and \(Q_c\) has at most
linear growth. Thus \(B_g\) is well defined and finite.

Recall that \(F_i^0(c)=|U_i|Q_c(Z_i)\).
By \eqref{eq:proof_sign_centered_table}, the linear growth
bound on \(Q_c\), and the uniform boundedness of \(U_i\),
we have
\[
|F_i(c)|+|F_i^0(c)|\le C(1+\|Z_i\|_2).
\]
Consequently, both scores have uniformly bounded moments
of every fixed order.

Since \(\sum_{c=1}^{T+1}F_i(c)=0\),
\eqref{eq:proof_randomization_variance} gives
\[
V_N^{\pm}=\frac{N}{N-1}\sum_{i=1}^N\frac1{T+1}
\sum_{c=1}^{T+1}
\left\{F_i(c)-\frac1N\sum_{j=1}^NF_j(c)\right\}^2.
\]
By \eqref{eq:proof_sign_mean_square} and the law of
total expectation, for every \(c\in[T+1]\),
\[
\mathbb E\left[
\frac1N\sum_{i=1}^N|F_i(c)-F_i^0(c)|^2
\right]
\le \max_j\mathbb E\!\left[
|F_1(c)-F_1^0(c)|^2\mid A_1=j
\right]=o(1).
\]
Markov's inequality therefore gives
\[
\frac1N\sum_i|F_i(c)-F_i^0(c)|^2=o_p(1).
\]
For fixed \(c\), the variables \(F_i^0(c)\) are independent
across units with uniformly bounded fourth moments.
By Chebyshev's inequality,
\[
\begin{aligned}
\frac1N\sum_iF_i^0(c)
=\mathbb E[|U_1|]\,\mathbb E[Q_c(Z)]+o_p(1),\qquad
\frac1N\sum_iF_i^0(c)^2 =\mathbb E[Q_c(Z)^2]+o_p(1),
\end{aligned}
\]
where we use the independence of \(U_i\) and \(Z_i\)
and the fact that \(\mathbb E[U_i^2]=1\).

By the Cauchy--Schwarz inequality,
\[
\left|\frac1N\sum_i\{F_i(c)-F_i^0(c)\}\right|
\le
\left\{\frac1N\sum_i|F_i(c)-F_i^0(c)|^2\right\}^{1/2}
=o_p(1).
\]
Similarly,
\[
\begin{aligned}
& \left|\frac1N\sum_i\{F_i(c)^2-F_i^0(c)^2\}\right|
\le
\frac1N\sum_i|F_i(c)-F_i^0(c)|^2\\
& +
2\left\{\frac1N\sum_i|F_i(c)-F_i^0(c)|^2\right\}^{1/2}
\left\{\frac1N\sum_iF_i^0(c)^2\right\}^{1/2} =o_p(1),
\end{aligned}
\]
where the last equality follows from the mean-square
approximation and \(N^{-1}\sum_iF_i^0(c)^2=O_p(1)\).
Combining these results gives
\[
\begin{aligned}
\frac1N\sum_iF_i(c)
=\mathbb E[|U_1|]\,\mathbb E[Q_c(Z)]+o_p(1),\qquad
\frac1N\sum_iF_i(c)^2 =\mathbb E[Q_c(Z)^2]+o_p(1).
\end{aligned}
\]
Substituting these empirical moment limits into the
randomization variance formula gives
\begin{equation}\label{eq:proof_sign_reference_variance}
\begin{aligned}
\frac{V_N^{\pm}}{N}
&=\frac{N}{N-1}\frac1{T+1}\sum_{c=1}^{T+1}
\left\{
\frac1N\sum_i F_i(c)^2
-\left(\frac1N\sum_i F_i(c)\right)^2
\right\}\\
&=\frac1{T+1}\sum_{c=1}^{T+1}
\left\{
\mathbb E[Q_c(Z)^2]
-(\mathbb E[|U_1|])^2
(\mathbb E[Q_c(Z)])^2
\right\}+o_p(1)\\
&=\omega_N^2+o_p(1),
\end{aligned}
\end{equation}
where the last equality follows from
\eqref{eq:proof_sign_noise}.

Conditional on \(A_{[N]}\), the summands in \(\bar T_N\) are
independent with uniformly bounded fourth moments. Their total
third absolute centered moment is \(O(N)\). By
Lemma~\ref{lem:sign_score_moments} and
\eqref{eq:proof_sign_variance_bounds}, their total variance is
bounded below by a positive multiple of \(N\). Lyapunov's
central limit theorem therefore gives
\begin{equation}\label{eq:proof_sign_observed_clt}
\frac{\bar T_N-\mathbb E[\bar T_N\mid A_{[N]}]}
{\sqrt N\,\omega_N}
\xrightarrow{\;d\;} N(0,1).
\end{equation}
The bounds hold uniformly over assignments with the specified
cohort counts, so the limit also holds after averaging over
\(A_{[N]}\).

The remainder in \eqref{eq:proof_sign_observed_mean}
is \(o(\sqrt N)\), since $\sqrt N\,t_N^4=(\sqrt N\,t_N^2)t_N^2\to0.$
Combining the mean expansion in
\eqref{eq:proof_sign_observed_mean} with
\eqref{eq:proof_sign_observed_clt} and
\eqref{eq:proof_sign_reference_variance} gives
\[
\frac{\bar T_N}{\sqrt{V_N^{\pm}}}
-\frac{\sqrt N\,t_N^2\mathbb E[|U_1|^3]\,\gamma_T(b)}{\omega_N}
\xrightarrow{\;d\;} N(0,1).
\]
Here the mean shift is bounded because \(\sqrt N\,t_N^2=O(1)\),
\(U_1\) is uniformly bounded, and \(\omega_N\) is bounded
away from zero.

If \(h>0\), then \(\theta_N/\sigma\) has order \(N^{-1/4}\),
whereas \(\delta_N/\sigma=o(N^{-1/2})\). Hence
\(\delta_N/\theta_N\to0\) and \(s_N/\theta_N\to1\).
By the normalization in \eqref{eq:sign_power_normalization},
\[
\frac{\delta_N+\theta_Nu(x)}{s_N}\to u(x)
\]
uniformly in \(x\), since \(u\) is bounded. It follows that
\[
\mathbb E[|U_1|]\to\mathbb E[|u(X)|],
\qquad
\mathbb E[|U_1|^3]\to\mathbb E[|u(X)|^3].
\]
Substituting the first limit into \eqref{eq:proof_sign_noise}
gives
\[
\omega_N^2\to
\frac1{T+1}\sum_{c=1}^{T+1}
\operatorname{Var}\bigl(|u(X)|Q_c(Z)\bigr).
\]
The limit is strictly positive because \(\omega_N^2\)
is bounded away from zero. Since
\(\sqrt N\,t_N^2=h_N+o(1)\) and \(h_N=O(1)\), we obtain
\[
\begin{aligned}
\frac{\sqrt N\,t_N^2\mathbb E[|U_1|^3]\,\gamma_T(b)}
{\omega_N}
&=
h_N\frac{\mathbb E[|u(X)|^3]\,\gamma_T(b)}
{\left\{\dfrac1{T+1}\sum_{c=1}^{T+1}
\operatorname{Var}\bigl(|u(X)|Q_c(Z)\bigr)\right\}^{1/2}}
+o(1)\\
&=h_NB_g+o(1),
\end{aligned}
\]
where the last equality follows from the definition of
\(B_g\) in \eqref{eq:proof_AB_constants}.

If \(h=0\), the boundedness of \(U_1\) and the positive
lower bound on \(\omega_N\) give
\[
\left|
\frac{\sqrt N\,t_N^2\mathbb E[|U_1|^3]\,\gamma_T(b)}
{\omega_N}
\right|
\le C\sqrt N\,t_N^2\to0.
\]
Since \(B_g\) is finite and \(h_N\to0\), the mean shift
is again \(h_NB_g+o(1)\).
This case does not require \(\delta_N/\theta_N\) to converge.
Thus \eqref{eq:proof_sign_normalized_limit} follows from
Slutsky's theorem in both cases.

Finally, by the bound \(|F_i(c)|\le C(1+\|Z_i\|_2)\),
\[
\frac1{N(T+1)}\sum_{i=1}^N\sum_{c=1}^{T+1}|F_i(c)|^3
\le \frac{C}{N}\sum_{i=1}^N(1+\|Z_i\|_2)^3
=O_p(1).
\]
Equation~\eqref{eq:proof_sign_reference_variance} and the
positive lower bound in \eqref{eq:proof_sign_variance_bounds}
verify the variance condition in
Lemma~\ref{lem:balanced_randomization}. Applying that lemma
proves \eqref{eq:proof_sign_randomization_limit}.

\end{proof}
\subsubsection{Positivity of \(B_g\) under geometric decay}
\label{app:geometric_positive_response}

\begin{lemma}\label{lem:geometric_response}
For every fixed \(T\ge3\) and \(0<\rho\le1/2\), the constant
\(B_g\), defined in \eqref{eq:proof_AB_constants}, is strictly positive.
\end{lemma}

\begin{proof}
By the definition of \(B_g\) in \eqref{eq:proof_AB_constants},
it suffices to show that \(\gamma_T(b)>0\): its denominator
is positive, and \(\mathbb E[|u(X)|^3]>0\) by 
\(\mathbb E[u(X)^2]=1\).
Recall from \eqref{eq:proof_geometric_constants} that
\[
\gamma_T(b)=\frac1{2(T+1)}\sum_{c=1}^{T+1}
\mathbb E\!\left[
Q_c(Z)\{(d_c^\top Z)^2-\|d_c\|_2^2\}
\right],
\]
where \(Z\sim N(0,I_T)\).
Fix \(c\le T\) and define
\[
H_a=\sum_{r=\max(a,c)}^T w_{r-c}\rho^{2r-a-c},
\qquad
\overline H=\frac1{T+1}\sum_{a=1}^{T+1}H_a.
\]
Since \(d_a\) is the centered treatment trajectory,
\(v_c^\top d_a=H_a-\overline H\).
Note that \(H_c>0\), \(H_{T+1}=0\), and
\(H_a=\rho^{c-a}H_c\) for \(a<c\).

For \(c\le a<T\), we have the recursion $H_a=\rho^{a-c}w_{a-c}+\rho H_{a+1}.$
When \(c<T\), the last two nonzero terms satisfy
\[
\frac{H_T}{H_{T-1}}
=\frac{\rho}{w_{T-c-1}/w_{T-c}+\rho^2}
\le\frac{\rho}{1/2+\rho^2}
\le\frac23.
\]
For \(c\le a\le T-2\), using
\(w_{a+1-c}/w_{a-c}\le3/2\) and \(\rho\le1/2\) gives
\[
H_{a+1}-\frac34 H_a
\le \rho\left(H_{a+2}-\frac34 H_{a+1}\right).
\]
Backward induction therefore yields \(H_{a+1}<3H_a/4\)
for every \(c\le a<T\).

Put \(y_a=H_a/H_c\) and
\(\bar y=(T+1)^{-1}\sum_a y_a\).
Then \(y_c=1\), \(y_{T+1}=0\), and \(0\le y_a\le3/4\)
for \(a\ne c\).
Moreover, among the \(T-1\) coordinates other than
\(y_c=1\) and \(y_{T+1}=0\), at most one exceeds \(9/16\).
Consequently,
\[
\bar y\le\frac{1+3/4+(T-2)9/16}{T+1}
<\frac58,
\qquad
\frac1{T+1}\sum_a y_a^2
\le\frac34\bar y+\frac1{4(T+1)}.
\]
Since \(T\ge3\), we obtain
\[
\begin{aligned}
(1-\bar y)^2-\frac1{T+1}\sum_a(y_a-\bar y)^2
\ge 2\bar y^2-\frac{11}{4}\bar y+\frac{15}{16}=2\left(\frac58-\bar y\right)
     \left(\frac34-\bar y\right)>0.
\end{aligned}
\]
Thus, writing \(D_c=v_c^\top d_c=H_c-\overline H>0\),
we have
\begin{equation}\label{eq:proof_projection_variance}
\frac1{T+1}\sum_{a=1}^{T+1}(v_c^\top d_a)^2<D_c^2.
\end{equation}

We now use this inequality to establish the positivity of
\(\gamma_T(b)\).
Fix \(c\le T\) and \(a\in[T+1]\), and let
\(W=d_c^\top Z/\|d_c\|_2\).
By Gaussian conditioning,
\[
\begin{aligned}
&\mathbb E\!\left[
(v_c^\top Z)\{(d_a^\top Z)^2-\|d_a\|_2^2\}
\mid W
\right]=
\frac{D_c(d_c^\top d_a)^2}{\|d_c\|_2^3}(W^3-3W)
+\frac{2(d_c^\top d_a)(v_c^\top d_a)}{\|d_c\|_2}W.
\end{aligned}
\]
Recall that \(\zeta_c(Z)=\operatorname{sgn}(W)v_c^\top Z\).
Using \(\mathbb E[|W|]=2\varphi(0)\) and
\(\mathbb E[|W|^3]=4\varphi(0)\), where \(\varphi\)
is the standard normal density, we obtain
\[
\begin{aligned}
\frac12\mathbb E\!\left[
\zeta_c(Z)\{(d_a^\top Z)^2-\|d_a\|_2^2\}
\right]&=\varphi(0)
\left\{
\frac{2(d_c^\top d_a)(v_c^\top d_a)}{\|d_c\|_2}
-\frac{D_c(d_c^\top d_a)^2}{\|d_c\|_2^3}
\right\}\\
&=\frac{\varphi(0)\|d_c\|_2}{D_c}
\left\{
(v_c^\top d_a)^2
-\left(v_c^\top d_a
-\frac{D_cd_c^\top d_a}{\|d_c\|_2^2}\right)^2
\right\}\\
&\le\frac{\varphi(0)\|d_c\|_2}{D_c}
(v_c^\top d_a)^2.
\end{aligned}
\]
Equality holds when \(a=c\), in which case
\(v_c^\top d_c=D_c\). Hence
\[
\frac12\mathbb E\!\left[
\zeta_c(Z)\{(d_c^\top Z)^2-\|d_c\|_2^2\}
\right]
=\varphi(0)\|d_c\|_2D_c.
\]

Substituting the definition of \(Q_c\) in
\eqref{eq:proof_sign_reduced_score}
into the expression for \(\gamma_T(b)\),
exchanging the summation indices, and using
\(\zeta_{T+1}=0\), we obtain
\[
\begin{aligned}
\gamma_T(b)
&=\frac1{2(T+1)}\sum_{c=1}^T
\Bigg\{
\mathbb E\!\left[
\zeta_c(Z)\{(d_c^\top Z)^2-\|d_c\|_2^2\}
\right]\\
&\qquad-\frac1{T+1}\sum_{a=1}^{T+1}
\mathbb E\!\left[
\zeta_c(Z)\{(d_a^\top Z)^2-\|d_a\|_2^2\}
\right]
\Bigg\}.
\end{aligned}
\]
Applying the equality for \(a=c\) to the first expectation
and the upper bound to the expectations being subtracted
gives
\[
\begin{aligned}
\gamma_T(b)
&\ge\frac{\varphi(0)}{T+1}\sum_{c=1}^T
\frac{\|d_c\|_2}{D_c}
\left\{
D_c^2-\frac1{T+1}\sum_{a=1}^{T+1}
(v_c^\top d_a)^2
\right\}>0,
\end{aligned}
\]
where the last inequality follows from \(D_c>0\) and
\eqref{eq:proof_projection_variance}.
This proves \(B_g>0\).
\end{proof}
\end{document}